\documentclass[journal]{IEEEtran}
\usepackage[T1]{fontenc}% optional T1 font encoding
\usepackage{cite}
\usepackage{amsmath,amssymb,amsfonts}
\usepackage{algorithmic}
\usepackage{algorithm}
\usepackage{graphicx}
\usepackage{textcomp}
\usepackage{xcolor}
\usepackage{color}
\usepackage{amsthm}
\usepackage{cite}
\usepackage{caption}
\usepackage{booktabs}
\usepackage{makecell}
\usepackage{threeparttable}
\usepackage{subfig}
\usepackage[numbers,sort&compress]{natbib}
\renewcommand{\algorithmicrequire}{\textbf{Input:}}  % Use Input in the format of Algorithm  
\renewcommand{\algorithmicensure}{\textbf{Output:}} % Use Output in the format of Algorithm  
\begin{document}

\title{Communication-Efficient Design for Quantized Decentralized Federated Learning}%: Towards Reducing Quantization Distortion and Communicated Bits}

%\author{Wei Liu, Li Chen}
	\author{Li Chen,~\IEEEmembership{Senior Member,~IEEE}, Wei Liu, Yunfei Chen,~\IEEEmembership{Senior Member,~IEEE}, and Weidong Wang       
		
		\thanks{
			% This work was supported by National Natural Science Foundation of China (Grant No. 61601432), the Fundamental Research Funds for the Central Universities.
			
			Li Chen, Wei Liu and Weidong Wang are with Department of Electronic Engineering and Information Science, University of Science and Technology of China, Hefei 230027, China
			(e-mail: chenli87@ustc.edu.cn, liuwei93@mail.ustc.edu.cn, wdwang@ustc.edu.cn).	
			%	N. Zhao is with the School of Info. and Commun. Eng., Dalian University of Technology, Dalian, China (email: zhaonan@dlut.edu.cn).
			
			Yunfei Chen is with Department of Engineering, University of Durham, South Road, Durham, DH1 3LE, UK (e-mail: yunfei.chen@durham.ac.uk).
			%	F.R. Yu is with the Department of Systems and Computer Engineering, Carleton University, Ottawa, ON, K1S 5B6, Canada (email: richard.yu@carleton.ca).
	}}

        % <-this % stops a space
%\thanks{This paper was produced by the IEEE Publication Technology Group. They are in Piscataway, NJ.}% <-this % stops a space
%\thanks{Manuscript received April 19, 2021; revised August 16, 2021.}}

\markboth{Journal of \LaTeX\ Class Files,~Vol.~14, No.~8, August~2021}%
{Shell \MakeLowercase{\textit{et al.}}: A Sample Article Using IEEEtran.cls for IEEE Journals}


\maketitle

\begin{abstract}
Decentralized federated learning (DFL) is a variant of federated learning, where edge nodes only communicate with their one-hop neighbors to learn the optimal model. However, as information exchange is restricted in a range of one-hop in DFL, inefficient information exchange leads to more communication rounds to reach the targeted training loss. This greatly reduces the communication efficiency. In this paper, we propose a new non-uniform quantization of model parameters to improve DFL convergence. Specifically, we apply the Lloyd-Max algorithm to DFL (LM-DFL) first to minimize the quantization distortion by adjusting the quantization levels adaptively. Convergence guarantee of LM-DFL is established without convex loss assumption. Based on LM-DFL, we then propose a new doubly-adaptive DFL, which jointly considers the ascending number of quantization levels to reduce the amount of communicated information in the training and adapts the quantization levels for non-uniform gradient distributions. Experiment results based on MNIST and CIFAR-10 datasets illustrate the superiority of LM-DFL with the optimal quantized distortion and show that doubly-adaptive DFL can greatly improve communication efficiency. 
	
%Compared with federated learning, Decentralized Federated Learning (DFL) catches a higher scalability because of peer interactions, and avoids single node failure due to robust connections.
%%the total decentralized paradigms of training and communication, which flattens DFL framework into single-layer interaction. 
%Meanwhile, it brings a fundamental barrier that is the more frequent but inefficient gradient communications with peers for information exchange. Existing algorithms merely uniformly quantize gradients with an invariable number of quantized levels while ignoring the non-uniform gradient distributions and the slowing of convergence rate during training course. In this paper, we first propose the DFL with Lloyd-Max algorithm (LM-DFL) to deal with the variable gradient distribution, which applies Lloyd-Max quantizer to minimize the quantized distortion. We give the quantization variance and establish the convergence guarantee of LM-DFL with convex loss assumption. Furthermore, we design Doubly-Adaptive DFL, which adopts ascending number of quantized bits combined with variable quantized levels to deal with non-uniform gradient distributions and convergence rate, respectively. 
\end{abstract}

\begin{IEEEkeywords}
Decentralized federated learning, doubly-adaptive quantization, Lloyd-Max quantizer.
\end{IEEEkeywords}

\section{Introduction}
%Recently, thanks to the explosive growth of the powerful individual computing devices worldwide and the rapid advancement of Internet of Things (IoT), data generated at device terminals are experiencing an exponential increase \cite{chiang2016fog}. Data-driven machine learning is developing into an attractive technique, which makes full use of the tremendous data for predictions and decisions of future events. As a promising
%data-driven machine learning variant, Federated Learning (FL) provides a communication-efficient approach for processing voluminous distributed data and grows in popularity nowadays. 

Due to the explosively growing number of computational devices and the rapid development of social networking applications, the amount of terminal data generated will exceed the capacity of the Internet in the near future \cite{chiang2016fog}. This inspires data-intensive machine learning to sufficiently utilize the immense amount of data. As a distributed prototype of data-intensive machine learning, federated learning (FL) shows great promise by leveraging local resources to train statistical models directly. Recently, FL has shown significant practical values in many fields, e.g., natural language processing \cite{hard2018federated},  vehicle-to-vehicle communications \cite{samarakoon2019distributed} and computer vision \cite{liu2020federated}, etc.
	%\IEEEPARstart{T}{his}

    FL has two different frameworks, centralized and decentralized. Centralized FL (CFL) uses a central server to aggregate the distributed local models, while decentralized FL (DFL) exchanges models among peer nodes by model gossip.
    They are both operated by local updates and aggregations, where nodes perform stochastic gradient descent (SGD) to optimize the model with local dataset and periodical aggregations are conducted to fuse local models. CFL was first proposed in \cite{mcmahan2017communication} for communication-efficient learning and preventing privacy leakage of user data. The work in \cite{liu2020accelerating} proposed an acceleration method for CFL by momentum gradient descent.  Compared to CFL which requires a central server for data aggregation, DFL can avoid the bottleneck of the central server by utilizing inter-node communications but this increases the communication costs of peer nodes \cite{9154332}.

	In order to cope with the increasing communication costs from the large scale training model in CFL, several methods have been studied. One method for communication efficiency is lossy compression of the gradients or models. There are a lot of approaches to achieve the lossy compression from drastic compression to gentle quantization. The work in \cite{ang2020robust} studied the noise effect during model transmission in wireless communication.
    A normalized averaging method was proposed in \cite{wang2020tackling} to eliminate objective inconsistency while preserving fast error convergence. As an extreme compression method, 1bitSGD ignores all the amplitude informations and only retains the sign information of gradient \cite{seide20141}. Although empirical convergence can be achieved under different experimental conditions, its theoretical convergence guarantee has not been established. TernGrad compressed gradients into three numerical levels can encode gradients into 2 bits \cite{wen2017terngrad}. It is unable to achieve a deterministic convergence guarantee. Another way to compress the gradient is to transmit sparse gradients. Top-$k$ method was studied in \cite{stich2018sparsified} where the larger elements of gradients were retained while getting rid of the smaller elements. In \cite{vargaftik2022eden}, the authors provided a robust distributed mean estimation technique that naturally handles heterogeneous communication budgets and packet losses.
    
	Another method for communication efficiency is quantization. QSGD proposed in \cite{alistarh2017qsgd} is an uniform quantization method, which quantizes the elements of gradients with uniform level distribution and achieves the unbiasedness.
	IntSGD avoids communicating float by quantizing floats into an integer \cite{mishchenko2021intsgd}. Nonuniform quantization was studied in \cite{horvoth2022natural} by considering nonuniform distributions of gradients elements. The work in \cite{jhunjhunwala2021adaptive} studied adaptive gradient quantization by changing the sequence of quantization levels and the number of levels. Reference \cite{faghri2020adaptive}
	optimized quantization distortion using coordinate descent. A communication efficient federated learning framework based on  quantized compressed sensing was first proposed in \cite{oh2022communication}. 

    For the DFL, the bottleneck of the communication cost becomes more serious due to the limited communication resources at the clients. 
    Gradient compression and quantization play a critical role in DFL. The work in \cite{tang2018communication} studied compressed communication of DFL. Differential compression was adopted.
	The work in \cite{koloskova2019decentralized} proposed CHOCO-SGD algorithm and established convergence guarantees under strong convexity and non-convexity conditions, respectively. CHOCO-SGD based on DFL with multiple local updates and multiple communications was studied in \cite{liu2022decentralized}. The work in \cite{tang2019deepsqueeze} considered the error-compensated compression in DFL. 
    The work in \cite{lalitha2018fully} allowed the individual users only sample points from small subspaces of the input space.
    Quantized decentralized gradient descent was proposed in \cite{reisizadeh2019exact} by adopting an exponential interval for the quantization levels. Randomized compression and variance reduction were adopted in \cite{kovalev2021linearly} to achieve linear convergence.

	However, none of the aforementioned compression and quantization methods in both CFL and DFL has considered either adaptive quantizer to match the time-varying convergence rate or variable gradient distribution to minimize quantization distortion. Motivated by this observation, we propose a doubly-adaptive DFL framework where both the quantization levels and the number of quantization levels are jointly adapted. In order to optimize the quantization distortion under variable distribution of gradient, we design a DFL framework with Llyod-Max quantizer (LM-DFL) to minimize 
	quantization distortion. We derive a general convergence bound for the quantized DFL frameworks. Then we give the quantization distortion and establish the global convergence guarantee of LM-DFL without assuming any convex loss function.
	In order to reduce communicated bits for efficiency, we derive the optimal number of quantization levels, which shows DFL with the ascending number of quantization levels can converge with the minimal communicated bits. Combined adaptive quantization levels with the adaptive number of quantization levels, we propose doubly-adaptive DFL to improve communication efficiency significantly from two directions, i.e., the number of quantization levels over iterations and the distribution of quantization levels over axis. The main contributions of this work are summarized as follows:
	\begin{itemize} 
		\item \textbf{Developing LM vector quantizer for DFL:} The scalar LM quantizer cannot be applied to the DFL directly. Thus, we design the LM vector quantizer for inter-node communication of DFL, which minimize quantization distortion given the number of the quantization levels.
		
		\item \textbf{Convergence with quantization distortion:}
        We establish a general convergence bound of DFL with any quantization or compression operators without convex loss function. And the results are extended to the proposed LM-DFL to provide its convergence performance.

		%	We establish the convergence of DFL and derive a convergence bound that incorporates network topology and the allocation  of computation and communication steps in a round. It shows that the convergence performance of DFL outperforms that of traditional frameworks.
		%In order to solve model consensus and communication efficiency problems, we proposed DFL design where multiple local updates and multiple inter-node communications are executed. DFL generalizes traditional decentralized frameworks and significantly improves the convergence of decentralized optimization problem. Furthermore, 
		
		\item \textbf{Proposed doubly-adaptive algorithm for DFL:} Based on the convergence bound of the LM-DFL, the doubly-adaptive algorithm is proposed to adapt both the quantization levels and the number of quantization levels. It optimize the convergence performance through matching both the time-varying convergence rate and the variable gradient distribution.

		%In order to improve communication efficiency of DFL, we further introduce compressed communication to DFL as a new scheme, termed C-DFL, where inter-node communication is achieved with compressed information to reduce the communication overhead.
		%	we propose a communication-efficient method named C-DFL with communication compression to optimize communication overhead.
		%For C-DFL, the convergence performance is based on the synergistic influence of compression and allocation of communication and computation steps in a round. We provide a convergence analysis of C-DFL and establish its linear convergence property.
		%	Then we  CS-DFL compresses model updates with arbitrary compression ratios to reduce communication.
		
		%\item \textbf{Evaluation based on MNIST and CIFAR-10 datasets}: We evaluate the performance of DFL and C-DFL based on Convolutional Neural Network (CNN) with using the real datasets MNIST and CIFAR-10. The convergence of DFL and C-DFL are confirmed. Experimental comparisons between DFL and C-SGD illustrate the enhanced performance of DFL. Extensive simulations evaluate the influence of network topology and the impacts of the computation and communication frequency. The simulation of C-DFL verifies its improved communication efficiency.
		
	\end{itemize} 

The remainder part of this paper is organized as follows. We introduce the system model of DFL in Section \uppercase\expandafter{\romannumeral2}. % and present an overview of existing solutions on the problem in Section \uppercase\expandafter{\romannumeral3}. 
The design of LM-DFL is described in detail in Section \uppercase\expandafter{\romannumeral3} and the convergence analysis of LM-DFL is presented in Section \uppercase\expandafter{\romannumeral4}. Then doubly-adaptive DFL and its optimal number of quantization levels are discussed in Section \uppercase\expandafter{\romannumeral5}. Simulation and discussion are presented in Section \uppercase\expandafter{\romannumeral6}. The conclusion is given in Section \uppercase\expandafter{\romannumeral7}.

\subsection{Notations}
In this paper, we use $\textbf{1}$ to denote vector $\left[1, 1, ..., 1\right]^\top$, and define
the consensus matrix $\textbf{J}\triangleq \textbf{1}\textbf{1}^\top/(\textbf{1}^\top \textbf{1})$, which means under the topology represented by $\textbf{J}$, DFL can realize model consensus. All vectors in this paper are column vectors. We assume that the decentralized network contains $N$ nodes.
So we have $\textbf{1}\in \mathbb{R}^N$ and $\textbf{J}\in \mathbb{R}^{N\times N}$. 
We use $\left\|\cdot\right\|$, $\rm\left\|\cdot\right\|_{F}$, $\rm\left\|\cdot\right\|_{op}$ and $\|\cdot\|_2$ to denote the $\l_2$ vector norm, the Frobenius matrix norm, the operator norm and the $l_2$ matrix norm, respectively.

\section{System Model}
%In the section, we introduce general system model of DFL, where the decentralized learning problem is highlighted. Based on the system model, we further give DFL learning strategy.

\subsection{Learning Problem}
Consider a general system model as illustrated in Fig. \ref{fig1}.
%we describe the decentralized federated learning structure. 
The system model is made of 
$N$ edge nodes.
The $N$ nodes have distributed datasets $\mathcal{D}_1,\mathcal{D}_2,...,\mathcal{D}_i,...,\mathcal{D}_N$ with $\mathcal{D}_i$ owned by node $i$. 
We use $\mathcal{D}=\{\mathcal{D}_1, \mathcal{D}_2,
..., \mathcal{D}_N\}$ to denote the global dataset of all nodes.
Assuming $\mathcal{D}_i\cap\mathcal{D}_j=\emptyset$ for $i\neq j$, we define $D\triangleq|\mathcal{D}|$ and $D_i\triangleq|\mathcal{D}_i|$ for $i=1,2,...,N$, where $|\cdot|$ denotes the size of set, and $\textbf{x}\in\mathbb{R}^{d}$
denotes the model parameter. 

\begin{figure}[!t]
	\centering
	\includegraphics[scale=0.55]{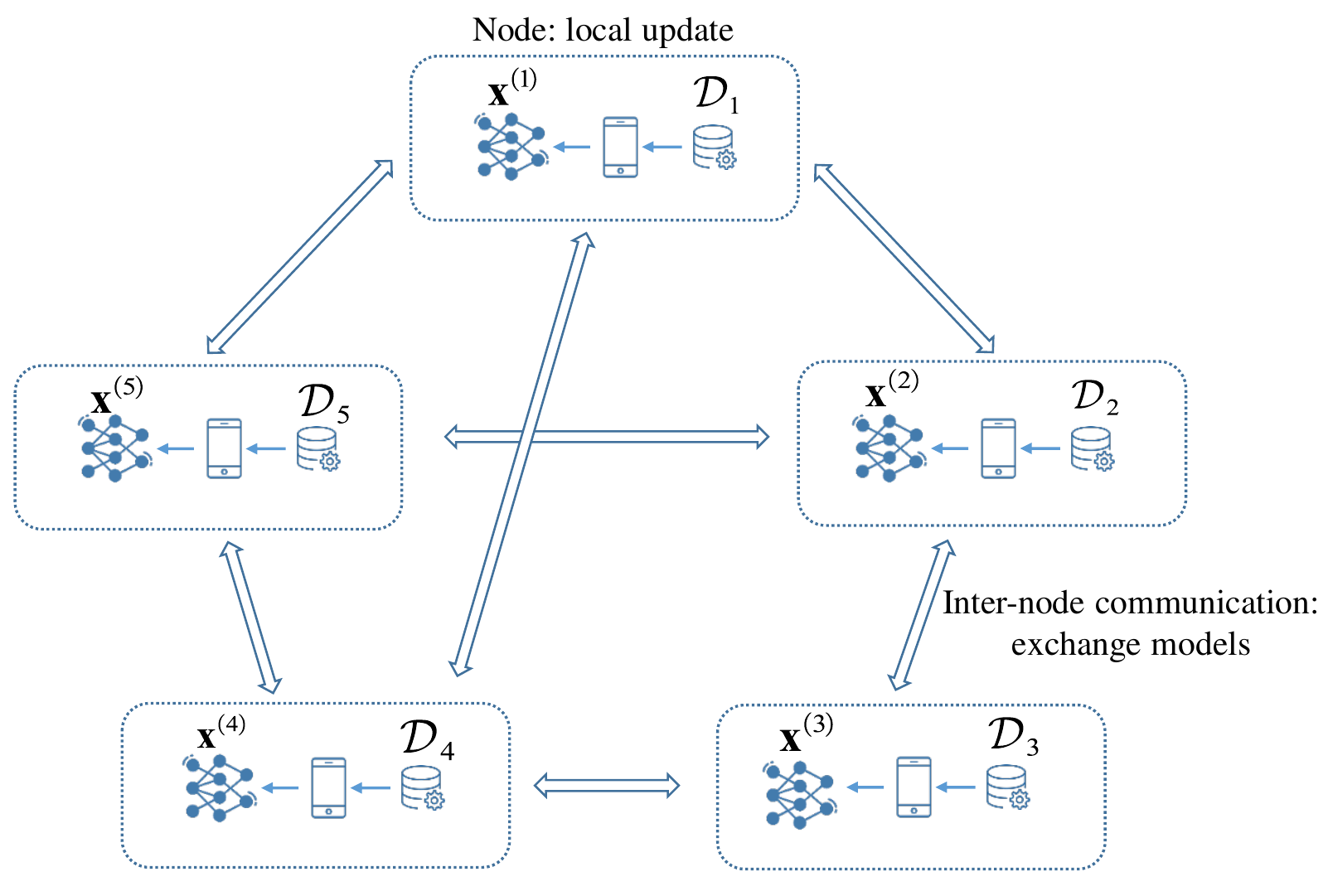}        %这个是在LaTeX文件夹中的相对路
	\caption{An example of DFL system model with a topology of 5 nodes.}
	\label{fig1}
\end{figure}

The loss function of the decentralized network is defined as follows. Let $s_j$ denote the $j$-th data sample of dataset $\mathcal{D}_i$.
The same learning model skeletons are embedded at different nodes. 
Denote $f(\cdot)$ as the loss function for one data sample. 
Loss function for the local dataset $\mathcal{D}_i$ and the global dataset $\mathcal{D}$, denoted by $F_i(\cdot)$ and $F(\cdot)$, are
\begin{align*}
&F_i(\textbf{x})\triangleq\frac{1}{D_i}\sum_{j\in \mathcal{D}_i}f(\textbf{x},s_j),\\
&F(\textbf{x})\triangleq\frac{1}{D}\sum_{i=1}^{N}D_i F_i(\textbf{x}).
\end{align*}

In order to learn about the characteristic information hidden in the distributed datasets,
DFL needs to obtain the optimal model parameters at all nodes. 
%We present the learning problem of DFL. 
The targeted optimization problem is the minimization of the global loss function to get the optimal parameter $\textbf{x}^*$ as follows:
\begin{align}
\label{min_problem}\textbf{x}^*=\mathop {\arg \min}_{\textbf{x}\in\mathbb{R}^d} F(\textbf{x}).
\end{align}
Because of the intrinsic complexity of most machine learning models and common training datasets, finding a closed-form solution to the optimization problem \eqref{min_problem} is usually intractable.
Thus, gradient-based methods are generally used to solve this problem.
Let $\textbf{x}_k$ denote the model parameter at the $k$-th iteration step, $\eta$ denote the learning rate, $\xi^{(i)}_{t}\subset\mathcal{D}_i$ denote the mini-batches randomly sampled on $\mathcal{D}_i$, and $$\widetilde{\nabla} f_i(\textbf{x},\xi^{(i)})\triangleq\frac{1}{|\xi^{(i)}|}\sum_{s_j\in\xi^{(i)}}\nabla f(\textbf{x},s_j)$$ be the stochastic gradient of node $i$.
The iteration of a gradient-based method can be expressed as
\begin{align}
\label{update_rule}\textbf{x}_{t+1}=\textbf{x}_{t}-\eta\left[\frac{1}{N}\sum_{i=1}^{N}\widetilde{\nabla} f_i(\textbf{x}_t,\xi^{(i)}_t)\right].	
\end{align}
We will use $\widetilde{\nabla} f_i(\textbf{x})$ to substitute $\widetilde{\nabla} f_i(\textbf{x},\xi^{(i)})$ in the remaining part of the paper for convenience.

Finally, the targeted optimization problem can be collaboratively solved by local update and model averaging. In the local update step, each node uses the gradient-based method to update its local model.
In the model averaging step, each node uses inter-node communication to communicate its local model parameters to its neighboring nodes simultaneously.

\subsection{Learning Strategy}
We introduce the DFL learning strategy in details as follows. In each iteration $k$, the learning strategy consists of two steps, local update and inter-node communication (communication with neighboring nodes). In local update stage, each node computes the current gradient of its local loss function and updates the model parameter by performing SGD multiple times in parallel.  After local update, each node performs inter-node communication when each node sends its local model parameters to the nodes connected with it, as well as receives the parameters from its neighboring nodes within one hop. 
Then a weighted model averaging of the received parameters based on contribution of connection (the contribution of the neighboring nodes to this node) is performed by the node to obtain the updated parameters. The whole learning process iterates between the two stages, as $k$ increases 1 by 1.
%considering that the elements in each column of $\textbf{C}$ reflect the contribution proportion of each node to a certain node

Define $\textbf{C}\in\mathbb{R}^{N\times N}$ as the confusion matrix which captures the network topology, and it is doubly stochastic, i.e., $\textbf{C}\textbf{1}=\textbf{1},\textbf{C}^{\top}=\textbf{C}$. Its element
$c_{ji}$ denotes the contribution of node $j$ in model averaging at node $i$. If there is no communication between two nodes $i$ and $j$, the corresponding elements $c_{ji}$ and $c_{ij}$ equal to zero.  
In each iteration $k$, DFL performs $\tau$ local updates in parallel at each node during the local update. After $\tau$ local updates, each node needs to perform inter-node communication. %We call $\tau_1$ \textit{the computation frequency} and $\tau_2$ \textit{the communication frequency} in the sequel.
In Fig. \ref{fig2}, we illustrate the learning strategy of DFL. 
%By continuous alternations of local updates and inter-node communications, S-DFL can achieve the minimization of the global loss function $F(\textbf{w})$ presented in \eqref{min_problem}.
We use $\textbf{x}^{(i)}_{k}$  to denote the model parameter after the weighted averaging of inter-node communication (or the initial model parameter of the local update stage in iteration $k$), and use $\textbf{x}^{(i)}_{k,t}$ to denote the model parameter of the $t$-th local updates in iteration $k$ with $t=0,1,2,...,\tau-1$. The calculated stochastic gradient of node $i$ at the $t$-th local update in iteration $k$ is $\widetilde{\nabla} f_i(\textbf{x}^{(i)}_{k,t})$. We set $\textbf{x}^{(i)}_{k,0}=\textbf{x}^{(i)}_{k}$ as the initial value of a new local update.

The learning strategy of DFL can be described as follows.

Local update: In the $k$-th iteration, for local update index $t=0,1,2,...,\tau-1$, local updates are performed at each node in parallel. At node $i$, the updating rule can be expressed as 
\begin{align}\label{localupdate}
\textbf{x}_{k,t+1}^{(i)}=\textbf{x}_{k,t}^{(i)}-\eta \widetilde{\nabla} f_i(\textbf{x}_{k,t}^{(i)}),
\end{align}
where $\eta$ is the learning rate.
Note that each node performs SGD individually using its local dataset at the same rate.

Inter-node communication: After $\tau$ local updates are finished, each node $i$ performs inter-node communication. They 
communicate with their one-hop nodes to exchange the model parameter. After the parameters from all the connected nodes are received, a weighted model averaging
is performed by
\begin{align}\label{internodecom}
\textbf{x}_{k+1}^{(i)}=\sum_{j=1}^{N}c_{ji}\textbf{x}_{k,\tau}^{(j)}.
\end{align}

For convenience, we rewrite the learning strategy into matrix form. In iteration $k$, we use $\textbf{X}_k$ to denote \textit{the initial model parameter matrix}. In the $t$-th local update of iteration $k$, $\textbf{X}_{k,t}$ denotes \textit{the model parameter matrix of local update}. Note that $\textbf{X}_k,\textbf{X}_{k,t}\in \mathbb{R}^{d\times N}$. 
%We first define \textit{the model parameter matrix about iteration} $\textbf{X}_k$ and \textit{the model parameter matrix about local update} in an iteration $\textbf{X}_{k,t}\in \mathbb{R}^{d\times N}$ as follows:
Then 
\begin{align}
\label{X_k}\textbf{X}_k&\triangleq\left[\textbf{x}^{(1)}_{k}\ \textbf{x}^{(2)}_{k} \ ...\ \textbf{x}^{(N)}_{k}\right],\\
\label{X_kt}\textbf{X}_{k,t}&\triangleq\left[\textbf{x}^{(1)}_{k,t}\ \textbf{x}^{(2)}_{k,t}\ ...\ \textbf{x}^{(N)}_{k,t}\right].
\end{align}
Each component of the difference matrix $\textbf{X}_{k,\tau}-\textbf{X}_{k}$ is $\textbf{x}_{k,\tau}^{(i)}-\textbf{x}_{k}^{(i)}$, which can be expressed as
\begin{align*}
\textbf{x}_{k,\tau}^{(i)}-\textbf{x}_{k}^{(i)}=-\eta\sum_{t=0}^{\tau-1}\widetilde{\nabla} f_i(\textbf{x}_{k,t}^{(i)})
\end{align*}
from \eqref{localupdate}. In fact, each node $i$ in DFL transmits the accumulated gradient in the local update step to its neighbors. This is similar to data-parallel SGD \cite{bekkerman2011scaling,chilimbi2014project,recht2011hogwild,chaturapruek2015asynchronous}, which is a synchronous and distributed framework. Each processor aggregates the value of a global $\textbf{x}_k$ and uses $\textbf{x}_k$
to obtain the local gradients, and communicates these updates to all peers for the next global model $\textbf{x}_{k+1}$.

\begin{figure}[!t]
	\centering
	\includegraphics[scale=0.3]{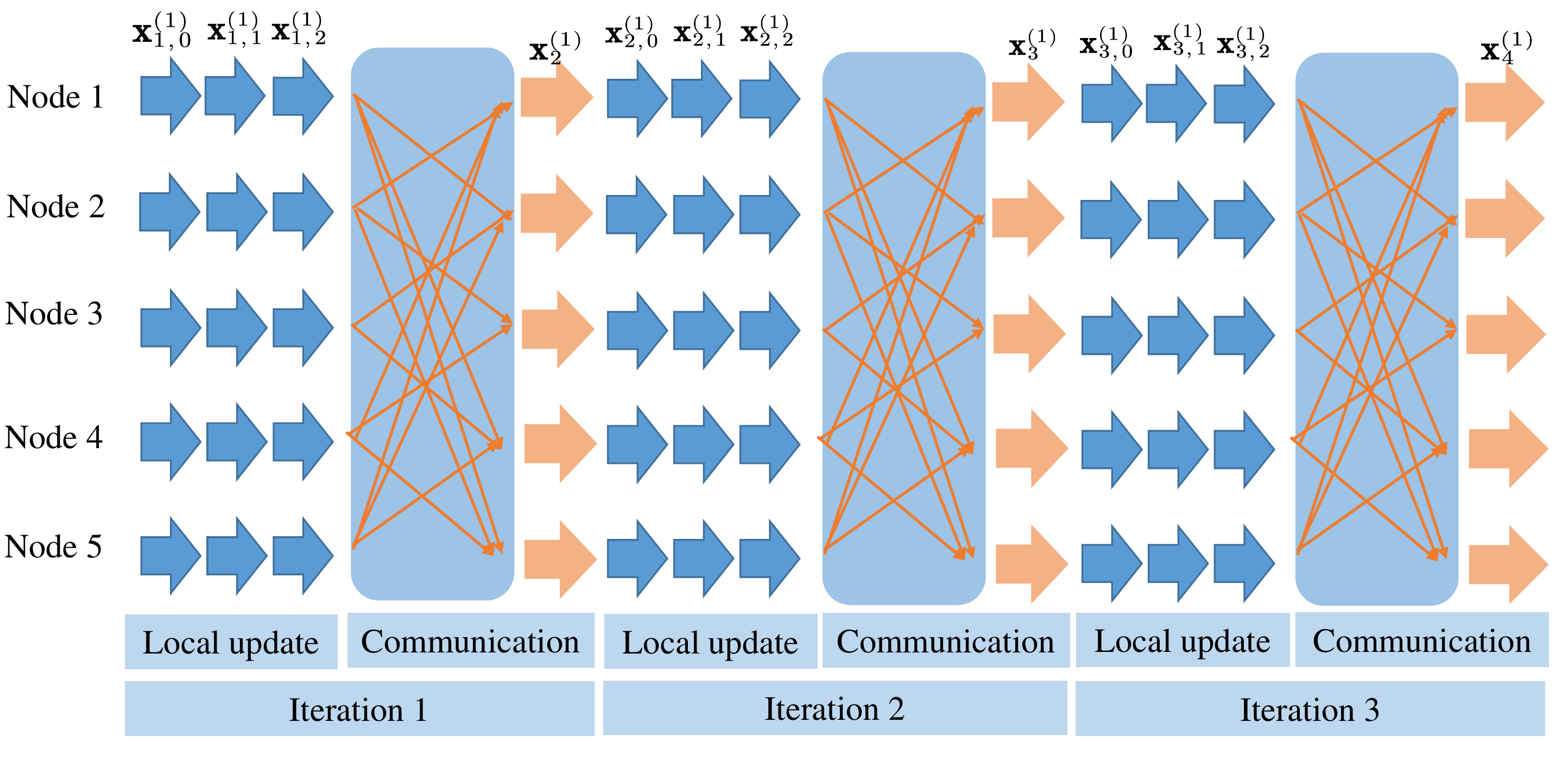}        %这个是在LaTeX文件夹中的相对路
	\caption{An example of DFL learning strategy with $\tau=3$ based on the system model in Fig. \ref{fig1}.}
	\label{fig2}
\end{figure}

Define the stochastic gradient matrix $\textbf{G}_{k,t}\in \mathbb{R}^{d\times N}$ in $t$-th local update of iteration $k$ as 
\begin{align}
\label{G_kt}\textbf{G}_{k,t}\triangleq\left[ \widetilde{\nabla} f_1(\textbf{x}_{k,t}^{(1)})\ \widetilde{\nabla} f_2(\textbf{x}_{k,t}^{(2)})\ ... \ \widetilde{\nabla} f_N(\textbf{x}_{k,t}^{(N)}) \right].
\end{align}
The learning strategy above can be simplified, i.e.,

Local update:
\begin{align}
\label{matrixlocalupdate-DFL}\textbf{X}_{k,t+1}=\textbf{X}_{k,t}-\eta\textbf{G}_{k,t},t=0,1,...,\tau-1.
\end{align}

Inter-node communication:
\begin{align}
\label{matrixupdate}\textbf{X}_{k+1}=\textbf{X}_{k,\tau}\textbf{C}.
\end{align}

%Therefore, considering the learning strategy after finishing each iteration, the global learning strategy in matrix form is presented in \eqref{matrixupdate}.
%\begin{align}
%\label{matrixupdate}\textbf{X}_{k+1}=\textbf{X}_k+(\textbf{X}_{k,\tau}-\textbf{X}_{k})\textbf{C}.
%\end{align}

\section{LM-DFL Design}
In this section, we propose the LM-DFL design. We will first design the Lloyd-Max (LM) vector quantizer to minimize the quantization distortion for a fixed number of quantization levels. Then we will design LM-DFL to solve the learning problem expressed in \eqref{min_problem} by minimizing the quantization distortion to improve convergence performance.

\iffalse
We further combine the two update rules into one to facilitate the following theoretical analysis.
Before describing the update rule of DFL, we define two time-varying matrices $\textbf{C}_t$, $\textbf{G}^{\prime}_t$ as follows:
\begin{align}
\label{W-t}
\textbf{C}_t&\triangleq
\begin{cases}
\textbf{I},&t\in[k]_1 \\
\textbf{C},&t\in[k]_2,
\end{cases}\\
\label{G-t}
\textbf{G}^{\prime}_t&\triangleq
\begin{cases}
\textbf{G}_{t},&t\in[k]_1 \\
\textbf{0},&t\in[k]_2,
\end{cases}
\end{align}
where $\textbf{I}\in\mathbb{R}^{N\times N}$ is the identity matrix and $\textbf{0}\in\mathbb{R}^{d\times N}$ is the zero matrix. According to the definitions of \eqref{W-t} and \eqref{G-t},
\textbf{the global learning strategy} of DFL can be described as 
\begin{align}
\label{global-update-rule}
\textbf{X}_{t+1}=(\textbf{X}_t-\eta \textbf{G}^{\prime}_t)\textbf{C}_t.
\end{align}

Assuming that all nodes have limited communication and computation 
resources, we use a finite number $T$ to denote the total iteration steps. Then the corresponding number of the iteration round is denoted by $K$ with $K=\lfloor\frac{T}{\tau}\rfloor$\footnote{Note that we set the condition $T-K\tau\leq\tau_1$ in this paper for simplified analysis. The other case $T-K\tau>\tau_1$ can be similarly treated based on the former.}.
\fi

\subsection{Introduction of Quantizers}
%\subsubsection{Introduction of Quantization}
For any vector $\textbf{v}\in \mathbb{R}^d$, we use $v_i$ to denote the $i$-th element of  $\textbf{v}$, and define $r_i\triangleq\frac{|v_i|}{\|\textbf{v}\|}$, where $r_i\in[0,1]$. For any scalar $x$, the function $\textup{sign}(x)=\{-1,+1\}$ denotes its sign, with $\textup{sign}(0)=1$.  

Now we consider a general quantizer $Q(\textbf{v})$ for quantizing the vector $\textbf{v}$. The quantizer $Q(\textbf{v})$ quantizes the $i$-th element of $\textbf{v}$ into $h(v_i)$. We define the quantizer $Q(\textbf{v})$ and $h(v_i)$ as,
\begin{align}
\label{Qv}Q(\textbf{v})&\triangleq [h(v_1)\  h(v_2)\ ...\ h(v_d)]^\top,\\
\label{q}h(v_i)&\triangleq\|\textbf{v}\|\cdot \textup{sign}(v_i)\cdot q(r_i),
\end{align}
respectively. The function $q(r_i): [0,1]\to[0,1]$ in \eqref{q} denotes the scalar quantizer. Quantizer $q(r_i)$ uses $s$ adaptable quantization levels, i.e., $\boldsymbol{\ell}=[\ell_1, \ell_2,..,\ell_s]$, with $\ell_j\in[0,1], j=1,2,...,s$. 

Given the vector quantizer $Q(\textbf{v})$, we discuss the total number of bits to encode it. Obviously, one bit is needed for $\textup{sign}(v_i)$ and the scalar $\|\textbf{v}\|_2$ is encoded with full precision, with $32$ bits. We also need $\lceil \log_2 s \rceil$ bits to represent the scalar quantizer $q(r_i)$. We use $C_s$ to denote the total number of bits for encoding $Q(\textbf{v})$. Therefore, $C_s$ can be given by
\begin{align}
\label{numberofbits} C_s=d\lceil \log_2 s \rceil+d+32.
\end{align}
It means that transmitting a single gradient/weight once from one node to a neighbor needs to consume $C_s$ bits in a DFL framework.

\subsection{Existing Quantizers}
%There are many methods for gradient quantization in federated learning field from 1 bits quantization to a quantization with arbitrary quantized levels \cite{seide20141,strom2015scalable,wen2017terngrad,alistarh2017qsgd,horvoth2022natural,faghri2020adaptive}. Quantizer based on merely several number of levels is indeterministic for convergence guarantee \cite{seide20141,strom2015scalable,wen2017terngrad}, let alone takes step to constrain the overlarge quantization distortion. Focusing on the methods of controlling quantization distortion, we present several quantizer with $s$ quantized levels. These quantizers show different choices about the sequence of quantized levels $\boldsymbol{\ell}=[\ell_1, \ell_2,..,\ell_s]$. Quantized Stochastic Gradient Descent (QSGD) proposed in \cite{alistarh2017qsgd} uses a randomly uniform quantizer for gradient quantization. Natural compression in \cite{horvoth2022natural} designs a quantizer, whose sequence of quantized levels is exponentially distributed. Adaptive Level
%Quantization (ALQ) in \cite{faghri2020adaptive} is a method using coordinate descent to minimize the quantization distortion. However, the aforementioned methods either omit the consideration that the distribution of gradient elements shows nonuniform properties substantially \cite{alistarh2017qsgd} \cite{horvoth2022natural}, or adopt a concessive solution to quantization distortion optimization \cite{faghri2020adaptive}.

We introduce the quantizers of QSGD \cite{alistarh2017qsgd}, natural compression \cite{horvoth2022natural} and ALQ \cite{faghri2020adaptive}.

\subsubsection{QSGD}
QSGD is an uniform and unbiased quantizer. Let $Q_{s}(\cdot)$ denote the vector quantizer and $q_{s}(\cdot)$ denote the scalar quantizer. Considering a vector $\textbf{v}\in \mathbb{R}^d$, then 
\begin{align*}
Q_s(\textbf{v})&= [h(v_1)\  h(v_2)\ ...\ h(v_d)]^\top,\\
h(v_i)&=\|\textbf{v}\|\cdot \textup{sign}(v_i)\cdot q_s(r_i).
\end{align*}
The quantization levels are uniformly distributed with $\boldsymbol{\ell}=[0, 1/s, 2/s,..., (s-1)/s,1]$. When $r\in(j/s,(j+1)/s]$, the scalar quantizer $q_s(\cdot)$ is expressed by
\begin{equation*}
q_s(r)=
\begin{cases}
j/s,& \mbox{with probability}\ j+1-sr \\
(j+1)/s,&\mbox{with probability}\ sr-j
\end{cases}
\end{equation*}
which satisfies $\mathbb{E}[q_s(r)]=r$. Hence $\mathbb{E}[Q_s(\textbf{v})]=\textbf{v}$ can be obtained easily, which shows the unbiasedness of QSGD.
The distortion bound is $\mathbb{E}[\|Q_s(\textbf{v})-\textbf{v}\|^2]\leq\min(d/s^2,\sqrt{d}/s)\|\textbf{v}\|^2$ \citep[Lemma 3.1]{alistarh2017qsgd}.

\subsubsection{Natural Compression}
Natural compression is an unbiased and nonuniform quantizer. let $Q_n(\cdot)$ denote the vector quantizer and $q_{n}(\cdot)$ denote the scalar quantizer. For a vector $\textbf{v}\in \mathbb{R}^d$, 
\begin{align*}
Q_n(\textbf{v})&= [h(v_1)\  h(v_2)\ ...\ h(v_d)]^\top,\\
h(v_i)&=\|\textbf{v}\|\cdot \textup{sign}(v_i)\cdot q_n(r_i).
\end{align*}
The quantization levels are set as a binary geometric partition: $\boldsymbol{\ell}=[\ell_s=0,\ell_{s-1},...,\ell_1,\ell_0=1]=[0,2^{1-s},2^{2-s},...,2^{-1},1]$. When $r\in[\ell_{j+1},\ell_{j}]$ for $j\in\{0,1,...,s-1\}$, the scalar quantizer 
$q_n(r)$ is a random variable, presented as 
\begin{align*}
q_n(r)=
\begin{cases}
\ell_j,& \mbox{with probability}\ \frac{r-\ell_{j+1}}{\ell_j-\ell_{j+1}} \\
\ell_{j+1},&\mbox{with probability}\ \frac{\ell_{j}-r}{\ell_j-\ell_{j+1}}
\end{cases}
\end{align*}
which satisfies $\mathbb{E}[q_n(r)]=r$. Further $\mathbb{E}[Q_n(\textbf{v})]=\textbf{v}$ can be obtained easily.
The distortion bound is $\mathbb{E}[\|Q_n(\textbf{v})-\textbf{v}\|^2]\leq 1/8+\min(\sqrt{d}/2^{s-1},d/2^{2(s-1)})$, according to \citep[Theorem 7]{horvoth2022natural}.

\subsubsection{ALQ}
ALQ is an adaptive and unbiased quantizer considering gradient distribution. For a vector $\textbf{v}\in \mathbb{R}^d$, ALQ's level partition is 
$0=\ell_0<\ell_1<\cdots<\ell_s<\ell_{s+1}=1$ and $\boldsymbol{\ell}=[\ell_0,\ell_1,...,\ell_{s+1}]$ with adaptable quantization levels. The vector and scalar quantizers of ALQ are denoted as $Q_a(\cdot)$ and $q_a(\cdot)$ respectively.
For a vector $\textbf{v}\in \mathbb{R}^d$, 
\begin{align*}
Q_a(\textbf{v})&= [h(v_1)\  h(v_2)\ ...\ h(v_d)]^\top,\\
h(v_i)&=\|\textbf{v}\|\cdot \textup{sign}(v_i)\cdot q_a(r_i).
\end{align*}
When $r\in[\ell_j,\ell_{j+1}]$ for $j\in \{0,1,...,s\}$, 
\begin{align*}
q_a(r)=
\begin{cases}
\ell_j,& \mbox{with probability}\ \frac{\ell_{j+1}-r}{\ell_{j+1}-\ell_j} \\
\ell_{j+1}.&\mbox{with probability}\ \frac{r-\ell_j}{\ell_{j+1}-\ell_j}
\end{cases}
\end{align*}
Therefore, ALQ satisfies $\mathbb{E}[Q_a(\textbf{v})]=\textbf{v}$.

In order to optimize the normalized quantization distortion in \eqref{mindistortion}, ALQ uses coordinate descent over the training process by 
\begin{align*}
\ell_j\!(k\!+\!1)\!\!=\!\!\Phi^{-1}\!\!\left(\!\!\Phi(\ell_{j+1}(k))\!-\!\!\!\int_{\ell_{j\!-\!1}(k)}^{\ell_{j\!+\!1}(k)}\!\!\!\!\frac{r-\ell_{j-1}(k)}{\ell_{j+1}(k)-\ell_{j-1}(k)}\!d\Phi(r)\!\!\right)\\\quad \forall j=1,...,s
\end{align*} 
where $\Phi(\cdot)$ is the cumulative distribution function of gradient.

QSGD and natural compression quantizers cannot match a dynamic distribution of model parameters. This causes high quantization distortion. On the other hand, coordinate descent of ALQ updates quantization levels during iterations, but is only asymptotically optimal.

\subsection{LM Quantizer}
In this subsection, we design LM quantizer that dynamically changes quantization levels to minimize quantization distortion during each iteration over the training process.

We first give the definition of quantization distortion.
The mean squared error of vector quantization is used to evaluate the quantization distortion, as 
\begin{align}
\label{mse}\mathbb{E}[\|Q(\textbf{v})-\textbf{v}\|^2]\triangleq\|\textbf{v}\|^2\sum_{i=1}^{d}\mathbb{E}[(q(r_i)-r_i)^2].
\end{align}
Finding a sequence of quantization levels $\boldsymbol{\ell}_s$ to minimize the distortion can enhance the quantization precision effectively. If the normalized coordinates $r_i$'s are i.i.d. given $\|\textbf{v}\|$, the optimization problem to minimize the normalized quantization distortion can be formulated as 
\begin{align}
\label{mindistortion} \min_{\boldsymbol{\ell}_s}\frac{\mathbb{E}[\|Q(\textbf{v})-\textbf{v}\|^2]}{\|\textbf{v}\|^2} =\min_{\boldsymbol{\ell}_s}\sum_{i=1}^{d}\int_{0}^{1}(q(r_i)-r_i)^2\phi(r_i)dr_i,
\end{align}
where $\phi(\cdot)$ is the probability density function of $\textbf{v}$'s elements.
The quantization distortion minimization problem is hard to solve for an arbitrary distribution. Reference \cite{faghri2020adaptive} showed the analytical solution for a special case when optimizing a single level $\ell_j$ given $\ell_{j-1}$ and $\ell_{j+1}$.

\subsubsection{LM Quantized Method}
The works in \cite{max1960quantizing} \cite{lloyd1982least} applied LM quantized method to minimizing the quantization error in data-transmission systems, which will be used in the following.

%The works in \cite{max1960quantizing} and \cite{lloyd1982least} laid a foundation for Lloyd-Max algorithm, which can solve this scalar distortion minimization problem with an exponential rate of convergence \cite{kieffer1982exponential}. 

We define a boundary sequence $\boldsymbol{b}=[b_0,b_1,...,b_{s-1},b_s]$, where $b_0=0,b_s=1,b_j\in(0,1)$, $j=1,2,...,s-1$. Thus, the sequence $\boldsymbol{b}$ has $s$ bins.
The $j$-th level should fall in the $j$-th bin, i.e., $\ell_j\in(b_{j-1},b_j]$ ($\ell_1=0$ belongs to the first bin). According to the expected vector distortion in \eqref{mse}, we consider the expected scalar distortion $\mathbb{E}[(q(r)-r)^2]$.

If $r$ belongs to the $j$-th bin, i.e., $r\in(b_{j-1},b_j]$ ($r=0$ belongs to the first bin), LM scalar quantizer quantizes $r$ into $\ell_j$, i.e., $q(r)=\ell_j$.
Let $D$ denote the expected scalar distortion $\mathbb{E}[(q(r)-r)^2]$ with
\begin{align}\label{minD}
D&=\mathbb{E}[(q(r)-r)^2]
=\sum_{j=1}^{s}\int_{b_{j-1}}^{b_j}(\ell_j-r)^2\phi(r)dr.
\end{align}

The following lemma shows a necessary condition for minimizing $D$ in \eqref{minD}.
\newtheorem{lemma}{Lemma}
\begin{lemma}[\textbf{The Optimal Solution to Minimizing Distortion}]\label{lemmasolution}
Given the probability density function $\phi(r)$ for $r\in[0,1]$ and the number of quantization levels $s$,
then necessary conditions of the optimal solution to \eqref{minD} are
	\begin{align}
	\label{DD1}b_j&=(\ell_{j}+\ell_{j+1})/2,\quad j=1,...,s-1,\\
	\label{DD2}\ell_j&=\frac{\int_{b_{j-1}}^{b_j}r\phi(r)dr}{\int_{b_{j-1}}^{b_j}\phi(r)dr}.\quad j=1,...,s,
	\end{align}
\end{lemma}
\begin{proof}
	Differentiating $D$ with respect to the $b_j$'s and $\ell_j$'s and set derivatives equal to zero, one has
	\begin{align*}
	&\frac{\partial D}{\partial b_j}=(\ell_j-b_j)^2\phi(b_j)-(\ell_{j+1}-b_j)^2\phi(b_j)=0,\\ \notag&\qquad\qquad\qquad\qquad\qquad\qquad\qquad\qquad j=1,...,s-1,\\
	&\frac{\partial D}{\partial \ell_j}=2\int_{b_{j-1}}^{b_j}(\ell_j-r)\phi(r)dr=0,\ \quad\,\, j=1,...,s.
	\end{align*}
	Then, Lemma \ref{lemmasolution} can be obtained from the above.
\end{proof}

%From \eqref{D1} and \eqref{D2}, we can get respectively. 
According to Lemma \ref{lemmasolution}, the boundary $b_j$ is the midpoint between $\ell_{j}$ and $\ell_{j+1}$; the quantization level $\ell_j$ is the centroid of the area of $\phi(r)$ between $b_{j-1}$ and $b_j$. In order to minimize $D$ in \eqref{minD}, one needs to find the sequences $\boldsymbol{b}$ and $\boldsymbol{\ell}$ that satisfy \eqref{DD1} and \eqref{DD2} simultaneously.

\subsubsection{LM Quantizer Design}%\label{LMV}
Based on Lemma \ref{lemmasolution}, the quantized rule of LM quantizer can be designed to satisfy \eqref{DD1} and \eqref{DD2} via iteration. We present the iterative rules of LM quantizer as follows.

\textbf{The iterative rules of LM quantizer:} Given the number of quantization levels $s$, the iterations of LM quantizer are summarized in the following steps:
\begin{enumerate}
\setlength{\itemsep}{0.5ex}
\item[{1.}]  The boundary sequence $\boldsymbol{b}_0=[b_{0,0},b_{1,0},...,b_{s-1,0},b_{s,0}]$ is initialized uniformly in $[0,1]$;
\item[{2.}]   In iteration $p$, the new sequence of quantization levels $\boldsymbol{\ell}_p=[\ell_{1,p}, \ell_{2,p},..,\ell_{s,p}]$ is calculated according to  \eqref{DD2} based on $\boldsymbol{b}_{p-1}$ in iteration $p-1$;
\item[{3.}]   In iteration $p$, the new boundary sequence $\boldsymbol{b}_p=[b_{0,p},b_{1,p},...,b_{s-1,p},b_{s,p}]$ is calculated according to  \eqref{DD1} using $\boldsymbol{\ell}_p$ in Step 2;
\item[{4.}]  Iteration $p$ is incremented and stopping rule is checked. Algorithm is terminated if true, else jumps to Step 2;
\item[{5.}]   Using the final sequences of quantization levels $\boldsymbol{\ell}^*$ and $\boldsymbol{b}^*$, LM quantizer performs: $q_L(r)=\ell_j^*$ if $r\in (b_{j-1}^*,b_j^*]$ ($r=0, q_L(r)=\ell_1^*$). In this paper, we use $q_{L}(r)$ to denote LM scalar quantizer in the remaining parts.
\end{enumerate} 
The algorithm of LM quantizer is given in Algorithm \ref{alg_1}, illustrated  in Fig. \ref{fig_3}.

The above LM quantizer cannot be used in DFL directly because it is only for scalar. The extension to the vector quantizer is as follows.
\subsubsection{Design of LM Vector Quantizer}\label{LMV}

\begin{figure}[!t]
	\centering
	\includegraphics[scale=0.45]{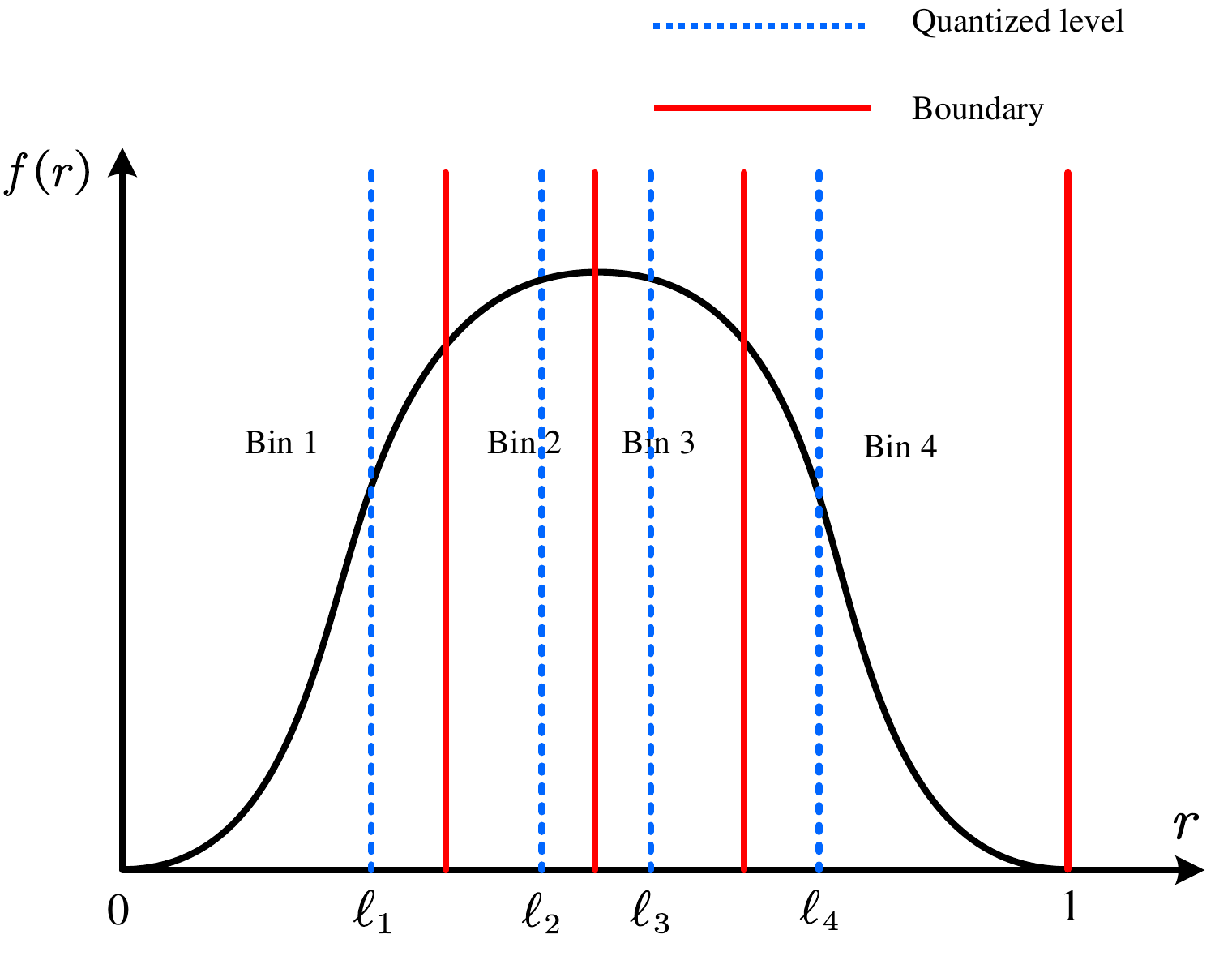}        %这个是在LaTeX文件夹中的相对路
	\caption{LM scalar quantizer.}
	\label{fig_3}
\end{figure}

\begin{algorithm}[!t]
	\caption{\textbf{LM scalar quantizer}} 
	\label{alg_1}
	\begin{algorithmic}[1]
		\REQUIRE ~~\\ %算法的输入参数：Input
		The input scalar $r$ with $r\in[0,1]$ to be quantized\\The number of quantization levels $s$ \\The probability density function $\phi(r)$
		\STATE Set the initial boundary sequence $\boldsymbol{b}_0$
		\STATE $p=1$
		\WHILE{convergence condition is not satisfied}
		%\FOR{$p=1,2,3,4...$}
		\STATE Calculate $\boldsymbol{\ell}_p$ based on $\boldsymbol{b}_{p-1}$ according to\\ $\quad\ell_j=\int_{b_{j-1}}^{b_j}r\phi(r)dr/\int_{b_{j-1}}^{b_j}\phi(r)dr$, with $j=1,...,s$
		\STATE Calculate $\boldsymbol{b}_p$ based on $\boldsymbol{\ell}_p$ according to\\ $\quad b_j=(\ell_{j}+\ell_{j+1})/2$, with $j=1,...,s-1$
		\STATE $p=p+1$
		%\IF{stopping rule is satisfied} 
		%\STATE	goto \textbf{end for}
		\ENDWHILE
		\STATE Performing LM scalar quantization: \\$\quad\quad q_{L}(r)=\ell_j^*$ if $r\in (b_{j-1}^*,b_j^*]$, also $q_{L}(0)=\ell_1^*$.
	\end{algorithmic}
\end{algorithm}

We design the LM vector quantizer used for inter-node communication of DFL. Let $Q_{L}(\textbf{v})$ denote the LM vector quantizer. For a vector $\textbf{v}\in \mathbb{R}^d$ to be quantized, the processing of $\textbf{v}$ at $Q_{L}(\textbf{v})$ has three parts, the $l_2$ norm, the $d$ signs of $v_i$ and the normalized absolute value of $v_i$.
\begin{enumerate}
	\setlength{\itemsep}{0.5ex}
	\item[{1.}] \textbf{The $l_2$ norm:} $Q_{L}(\textbf{v})$ calculates $\|\textbf{v}\|$ and encodes this norm with a full precision of 32 bits.
	\item[{2.}] \textbf{The $d$ signs of $v_i$:} $Q_{L}(\textbf{v})$ takes out the $d$ sign($v_i$)'s and encodes each sign with one bit.
	\item[{3.}] \textbf{The normalized absolute value of $v_i$:} $Q_{L}(\textbf{v})$ quantizes $d$ normalized scalar $r_i$'s with $r_i=|v_i|/\|\textbf{v}\|$ according to the LM scalar quantizer $q_{L}(r_i)$ as presented in Algorithm \ref{alg_1} and encodes the quantization levels with $\lceil \log_2 s \rceil$ bits.
\end{enumerate}

During the inter-node communication of DFL, each node quantizes and encodes its model parameters with $Q_{L}(\textbf{x})$, then sends the encoded bit stream to the neighboring nodes. When the neighbor receives the bit stream, decoding is performed to recover the $l_2$ norm of model parameter $\|\textbf{x}\|$, the $d$ signs of $x_i$ and the normalized absolute value of $x_i$. After decoding the three parts of model parameter, the receiver can receive the model parameters of the transmitter with a certain loss of quantization. %(we just research the quantization distortion with ignoring the transmission loss in this paper). From both the individual perspective of a vector and the global perspective of performance influence on DFL convergence, the theoretical analysis about the quantization distortion of $Q_{L}(\textbf{x})$ will be established in the next section. %We define distortion upper bound for a quantizer, denoted by $\omega$. 

\subsection{LM-DFL}

Based on the LM vector quantizer, we design LM-DFL to minimize the quantization distortion. LM-DFL is a DFL framework with learning strategy in \eqref{matrixupdate}. In the inter-node communication stage, each node in LM-DFL performs LM vector quantization as discussed in Section \ref{LMV}. Thus, one has the following.

\textbf{Local update:} In the $k$-th iteration, for local update index $t=0,1,2,...,\tau-1$, local updates are performed at each node in parallel. At node $i$, the updating rule can be expressed as 
\begin{align}\label{localupdate-LMDFL}
\textbf{x}_{k,t+1}^{(i)}=\textbf{x}_{k,t}^{(i)}-\eta \widetilde{\nabla} f_i(\textbf{x}_{k,t}^{(i)}),
\end{align}
where $\eta$ is the learning rate.
Note that the local updates of LM-DFL are consistent with DFL.

\textbf{Inter-node communication:} When $\tau$ local updates are finished, each node $i$ performs inter-node communication. %They communicate with their connected nodes to exchange the differential model parameter, i.e., $\textbf{x}_{k,\tau}^{(i)}-\textbf{x}_{k}^{(i)}$. 
%Before exchanging the differential model parameter, node $i$ quantizes 
%$\textbf{x}_{k,\tau}^{(i)}-\textbf{x}_{k}^{(i)}$ by LM vector quantizer $Q_{L}(\cdot)$. 
Nodes exchange their quantized differential model parameter $Q_{L}(\textbf{x}_{k,\tau}^{(i)}-\textbf{x}_{k}^{(i)})$ and $Q_L(\textbf{x}_k^{(i)}-\textbf{x}_{k-1,\tau}^{(i)})$ by LM vector quantizer with the connected nodes. 
After the parameters from all the connected nodes are received, 
Node $i$ calculate the model estimated parameters $\hat{\textbf{x}}_{k}^{(j)}$ for all connected node $j$ by
\begin{align}
	\hat{\textbf{x}}_{k}^{(j)}=\hat{\textbf{x}}_{k-1,\tau}^{(j)}+Q_{L}(\textbf{x}_{k-1,\tau}^{(j)}-\textbf{x}_{k-1}^{(j)})
	+Q_L(\textbf{x}_k^{(j)}-\textbf{x}_{k-1,\tau}^{(j)}).
\end{align}
Then a weighted model averaging
is performed by
\begin{align}\label{internodecom-LMDFL}
\textbf{x}_{k+1}^{(i)}=\sum_{j=1}^{N}c_{ji}[\textbf{x}_{k}^{(j)}+Q_{L}(\textbf{x}_{k,\tau}^{(j)}-\textbf{x}_{k}^{(j)})].
\end{align}
Note that the transmitting noise and decoding error are ignored.

LM-DFL learning strategy of matrix form can be provided. According to the definitions of the initial model parameter matrix $\textbf{X}_k$ in \eqref{X_k} and the local update model parameter matrix $\textbf{X}_{k,t}$ in \eqref{X_kt}, 
\textbf{the global learning strategy of LM-DFL} form can be rewritten as 
\begin{align}
\label{matrixupdate-LMDFL}\textbf{X}_{k+1}=[\hat{\textbf{X}}_k+Q_L(\textbf{X}_{k,\tau}-\textbf{X}_{k})]\textbf{C},
\end{align}
where $\hat{\textbf{X}}_k=[\hat{\textbf{x}}_k^{(1)}\ \hat{\textbf{x}}_k^{(2)}\ ...\ \hat{\textbf{x}}_k^{(N)}]$.
The estimated parameter matrix 
\begin{align}\label{eti1}
	\hat{\textbf{X}}_k=\hat{\textbf{X}}_{k-1}+Q_L(\textbf{X}_{k-1,\tau}-\textbf{X}_{k-1})+Q_L(\textbf{X}_k-\textbf{X}_{k-1,\tau}).
\end{align}
We present LM-DFL in Algorithm \ref{alg_2}. Because the nodes in LM-DFL are resource-constrained, we set $K$ iterations in LM-DFL algorithm. A comparison of FL methods with different quantizers is shown in Table \ref{table-comparison}. A further comparison of quantization distortion will be present in the next section.

\begin{algorithm}[!t]
	\caption{\textbf{LM-DFL}} 
	\label{alg_2}
	\begin{algorithmic}[1]
		\REQUIRE ~~\\ %算法的输入参数：Input
		Learning rate $\eta$\\
		Total number of iterations $K$\\
	    Number of local update $\tau$ in an iteration\\
		Confusion matrix $\textbf{C}$
		
		%		\ENSURE ~~\\ %算法的输出：Output
		%		The final global model weight vector $\textbf{w}^\mathrm{f}$
		\STATE Set the initial value of $\textbf{X}_{1,0}$. 
		%let $\widetilde{\textbf{w}}_i(0)=\textbf{w}_i(0)$, $\widetilde{\textbf{d}}_i(0)=\textbf{d}_i(0)$.
		\label{ code:fram:extract1 }%对此行的标记，方便在文中引用算法的某个步骤
		\FOR{$k=1,2,...,K$}
		\FOR{$t=0,1,...,\tau-1$} 
		\STATE Nodes calculate stochastic gradient
		\STATE Nodes perform SGD in parallel: \\ 
		$\textbf{X}_{k,t+1}=\textbf{X}_{k,t}-\eta\textbf{G}_{k,t}$ \hfill$//$\textit{local updates}
		\ENDFOR
		\STATE Each node $i$ calculates $\textbf{x}_{k,\tau}^{(i)}-\textbf{x}_{k}^{(i)}$, $\textbf{x}_{k}^{(i)}-\textbf{x}_{k-1,\tau}^{(i)}$ and computes the statistics to construct their probability density function $\phi_k^{(i)}(x)$
		\STATE Each node $i$ quantizes the differential model parameters to obtain $Q_{L}(\textbf{x}_{k,\tau}^{(i)}-\textbf{x}_{k}^{(i)})$ and
		$Q_{L}(\textbf{x}_{k}^{(i)}-\textbf{x}_{k-1,\tau}^{(i)})$ by LM vector quantizer.
		%performs quantization for differential model parameter by using LM vector quantizer based on  $\phi_k^{(i)}(x)$:\\
		%$\quad Q_{L}:\textbf{x}_{k,\tau}^{(i)}-\textbf{x}_{k}^{(i)} \to Q_{L}(\textbf{x}_{k,\tau}^{(i)}-\textbf{x}_{k}^{(i)})$.
		\STATE Nodes exchange the quantized differential model parameters and calculate the model estimated parameters
		$$	\hat{\textbf{X}}_k=\hat{\textbf{X}}_{k-1}+Q_L(\textbf{X}_{k-1,\tau}-\textbf{X}_{k-1})+Q_L(\textbf{X}_k-\textbf{X}_{k-1,\tau}).$$
		\STATE Then update the initial model parameter for next iteration by \\$\quad\textbf{X}_{k+1}\!=[\hat{\textbf{X}}_k\!+\!Q_L(\textbf{X}_{k,\tau}\!-\!\textbf{X}_{k})]\textbf{C}$. \hfill$//$\textit{communication}
		\ENDFOR
	\end{algorithmic}
\end{algorithm}

\newcommand{\tabincell}[2]{
	\begin{tabular}{@{}#1@{}}
		#2
	\end{tabular}
}
\begin{table*}[!t]  
	\caption{Comparison of different quantized FL methods}
	\centering
	\label{table-comparison}
	\renewcommand\arraystretch{1.5}
	\begin{threeparttable}
		\begin{tabular}{ccccc}  
			
			\toprule[1pt]   
			
			\textbf{Method} & \tabincell{c}{\textbf{Unbiasedness}} & \tabincell{c}{\textbf{Randomness}} & \tabincell{c}{\textbf{Distortion}} & \tabincell{c}{\textbf{Centralized or}\\ \textbf{decentralized structure}} \\  
			
			\midrule   
			
			QSGD \cite{alistarh2017qsgd} &  unbiased & \multicolumn{1}{c}{Random} & \multicolumn{1}{c}{$\min(d/s^2,\sqrt{d}/s)$}&Centralized\\  
			
			Natural Compression \cite{horvoth2022natural}&  unbiased  & \multicolumn{1}{c}{Random} & \multicolumn{1}{c}{$1/8+\min(\sqrt{d}/2^{s-1},d/2^{2(s-1)})$}&Centralized\\    
			
			ALQ\tnote{1} \ \cite{faghri2020adaptive}& unbiased & \multicolumn{1}{c}{Random}& \multicolumn{1}{c}{$\frac{(\ell_{j^*+1}/\ell_{j^*}-1)^2}{4(\ell_{j^*+1}/\ell_{j^*})}$}&Centralized\\
			
			LM-DFL & unbiased& \multicolumn{1}{c}{Deterministic}& \multicolumn{1}{c}{$d/(12s^2)$}&Decentralized\\
			
			\bottomrule[1pt]  
			
		\end{tabular}
		\begin{tablenotes}
			\footnotesize
			\item[1] In ALQ, the distortion is derived with a format of $\ell_{j^*+1}/\ell_{j^*}$, where $j^*=\arg\max_{1\leq j\leq s}\ell_{j+1}/\ell_j$. We present the quantization distortion of the proposed LM-DFL in Appendix D.% of Supplemental Material.
		\end{tablenotes}
	\end{threeparttable}
\end{table*}

\section{Performance of LM-DFL}
In this section, we will first establish a general convergence results of DFL with quantizer $Q(\cdot)$, named as QDFL. Then we will study the distortion bound of the LM vector quantizer and extend the convergence of QDFL to the proposed LM-DFL algorithm.

\subsection{Preliminaries}\label{assumption}
To facilitate the analysis, we make the following assumptions.
\newtheorem{assumption}{Assumption}
\begin{assumption}\label{ass1}
	we assume the following conditions:\\
	\textup{1)} $F_i(\textup{\textbf{x}})$ is $L$-smooth, i.e., $\|\nabla F_i(\textup{\textbf{x}})-\nabla F_i(\textup{\textbf{y}})\|\leq L\|\textup{\textbf{x}}-\textup{\textbf{y}}\|$ for some $L>0$ and any $\textup{\textbf{x}}$, $\textup{\textbf{y}},i$. Hence,  $F(\textup{\textbf{x}})$ is $L$-smooth from the triangle inequality.\\
	\textup{2)} $F(\textup{\textbf{x}})$ has a lower bound, i.e., $F(\textup{\textbf{x}})\geq F_{\textup{inf}}$ for some $F_{\textup{inf}}>0$.\\
	\textup{3)} Gradient estimation is unbiased for stochastic mini-batch sampling, i.e., $\mathbb{E}_{\xi|\textup{\textbf{x}}}[\widetilde{\nabla} f_i(\textup{\textbf{x}})]=\nabla F_i(\textup{\textbf{x}})$. The variance of gradient estimation is $\sigma_i$ for any $i$ and $\textup{\textbf{x}}$, i.e., $\mathbb{E}_{\xi|\textup{\textbf{x}}}\|\widetilde{\nabla} f_i(\textup{\textbf{x}})-\nabla F_i(\textup{\textbf{x}})\|^2=\sigma_i^2$ where $\sigma_i^2>0$. We define $\sigma^2\triangleq\frac{\sum_{i=1}^{N}\sigma_i^2}{N}$.\\
	\textup{4)} Gradient divergence $\|\nabla F_i(\textup{\textbf{x}})-\nabla F(\textup{\textbf{x}})\|$ is bounded by $\delta_i$ for any $i$ and $\textup{\textbf{x}}$, i.e., $\|\nabla F_i(\textup{\textbf{x}})-\nabla F(\textup{\textbf{x}})\|\leq\delta_i$. We define $\delta^2\triangleq\frac{\sum_{i=1}^{N}\delta_i^2}{N}$.\\
	\textup{5)} $\textup{\textbf{C}}$ is a doubly stochastic matrix, which satisfies $\textup{\textbf{C}}\textup{\textbf{1}}=\textup{\textbf{1}},
	\textup{\textbf{C}}^\top=\textup{\textbf{C}}$. The largest eigenvalue of $\textup{\textbf{C}}$ is always $1$ and the other eigenvalues are strictly less than $1$, i.e., $\max\left\{|\lambda_2(\textup{\textbf{C}})|, |\lambda_N(\textup{\textbf{C}})|\right\}<\lambda_1(\textup{\textbf{C}})=1$. For convenience, we define $\zeta\triangleq\max\left\{|\lambda_2(\textup{\textbf{C}})|, |\lambda_N(\textup{\textbf{C}})|\right\}$ to measure the confusion degree of DFL topology.\\%which is a similar way to analyze DFL in \textup{\cite{wang2021cooperative}}.\\
	%\textup{5)} Gradient divergence $\|\nabla F_i(\textup{\textbf{x}})-F(\textup{\textbf{x}})\|$ is bounded by $\delta_i$ for any $i$ and $\textup{\textbf{x}}$, i.e., $\|\nabla F_i(\textup{\textbf{x}})-F(\textup{\textbf{x}})\|\leq\delta_i$. We define $\delta^2\triangleq\frac{\sum_{i=1}^{N}\delta_i^2}{N}$.
	%where $\lambda_i(\textup{\textbf{C}})$ denotes the $i$-th largest eigenvalue of $\textup{\textbf{C}}$.
	% For convenience, we define $$\beta\triangleq\|\textup{\textbf{I}}-\textup{\textbf{C}}\|_2\in[0,2],$$which is a similar way to analyze communication-compressed D-SGD in \textup{\cite{koloskova2019decentralized}}.
\end{assumption}

\newtheorem{definition}{Definition}
We present the definition of the bound of quantization distortion for any quantizer. 
\begin{definition}[\textbf{Bound of Quantization Distortion}]\label{distor}
	For any given $ \textbf{\textup{x}}\in \mathbb{R}^d$, the distortion of any quantizer $Q$: $\mathbb{R}^d\to \mathbb{R}^d$ satisfies
	\begin{align}\label{Q}
	\mathbb{E}\|Q(\textbf{\textup{x}})-\textbf{\textup{x}}\|^2\leq \omega\|\textbf{\textup{x}}\|^2,
	\end{align}
	for $\omega>0$, where $\mathbb{E}$ denotes the expectation about the randomness of quantization. 
\end{definition}

Because the global loss function $F(\textbf{x})$ can be non-convex in some learning platforms such as convolutional neural network (CNN), SGD may converge to a local minimum or a saddle point. %\textcolor{blue}{The non-convex setting has been considered in many decentralized learning works \cite{lian2017can} \cite{wang2019cooperative} \cite{vlaski2021distributed} \cite{vlaski2021distributed1}. However, these works merely consider the decentralized framework with one local update and one inter-node communication (D-SGD) \cite{lian2017can} \cite{vlaski2021distributed} \cite{vlaski2021distributed1} or multiple local updates and one inter-node communication (C-SGD) \cite{wang2019cooperative} in a round.} 
We use the expectation of the gradient norm average of all iteration steps as the indicator of convergence \cite{jiang2018linear} \cite{yu2019linear}.
The algorithm is convergent if the following condition is satisfied.
\begin{definition}[\textbf{Convergence Condition}]\label{def1}
	The algorithm converges to a stationary point if 
	it achieves an $\epsilon$-suboptimal solution, i.e.,
	\begin{align}\label{defF}
	\mathbb{E}\left[\frac{1}{K}\sum_{k=1}^{K}\|\nabla F(\textup{\textbf{u}}_k)\|^2\right]\leq\epsilon.
	\end{align}
\end{definition}

\subsection{Distortion of LM-DFL}

We first establish the unbiasedness of the LM vector quantizer in the following theorem.
\newtheorem{theorem}{Theorem}
\begin{theorem}[\textbf{Unbiasedness}]
	For any given $\textup{\textbf{x}}\in\mathbb{R}^d$, the expectation of LM vector quantizer satisfies
	\begin{align*}
	\mathbb{E}[Q_L(\textup{\textbf{x}})]=\textup{\textbf{x}}
	\end{align*}
	which means the LM vector quantizer is unbiased. 
\end{theorem}
\begin{proof}
	Considering the LM scalar quantizer $q_L(r)$, the axis $[0,1]$ is partitioned into $s$ bins. Bin $j$ is bounded by $[b_{j-1},b_j]$ for $j=1,2,...,s$. If $r$ falls into bin $j$, $q_L(r)=\ell_j$. Therefore, the probability distribution of 
	$r$ is:  $q_L(r)=\ell_j$ with probability $p_j=\int_{b_{j-1}}^{b_j}\phi(r)dr$ for $j=1,2,...,s$. The expectation $\mathbb{E}[q_L(r)]$ satisfies
	\begin{align*}
	\mathbb{E}[q_L(r)]&=\sum_{j=1}^{s}\ell_j\int_{b_{j-1}}^{b_j}\phi(r)dr\\
	&=\sum_{j=1}^{s}\int_{b_{j-1}}^{b_j}r\phi(r)dr\\
	&=\int_{0}^{1}r\phi(r)dr\\
	&=\mathbb{E}[r],
	\end{align*}
	according to the definition of $\ell_j$ in \eqref{DD2}.
	The expectation of LM vector quantizer satisfies 
	\begin{align*}
	\mathbb{E}[Q_L(\textbf{x})]&=\mathbb{E}\{ \|\textbf{x}\|\cdot \textup{sign}(\textbf{x})\circ  q_L(\textbf{r})\}\\
	&=\|\textbf{x}\|\cdot \textup{sign}(\textbf{x})\circ\mathbb{E}[q_L(\textbf{r})]\\
	&=\|\textbf{x}\|\cdot \textup{sign}(\textbf{x})\circ\mathbb{E}[\textbf{r}]\\
	&=\mathbb{E}[\textbf{x}],
	\end{align*}
	where $r=[r_1,r_2,...,r_d]^\top$ and '$\circ$' denotes Hadamard product. 
	Since $\textbf{x}$ is given, we have $\mathbb{E}[Q_L(\textbf{x})]=\textbf{x}$.
\end{proof}

Then we give the quantization distortion of LM-DFL presented in the following theorem.
\begin{theorem}[\textbf{Quantization Distortion of LM-DFL}]
	Let $\textup{\textbf{x}}\in\mathbb{R}^d$ and the number of quantization levels is $s$. The upper bound of quantization distortion can be expressed as 
	\begin{align}
	\label{distortionbound}\mathbb{E}[\|Q_L(\textup{\textbf{x}})-\textup{\textbf{x}}\|^2]\leq\frac{d}{12s^2}\|\textup{\textbf{x}}\|^2.
	\end{align}
\end{theorem}
\begin{proof}
	The detailed proof is presented in Appendix A.%  of Supplemental Material.
\end{proof}
From $\eqref{distortionbound}$, the distortion of the LM vector quantizer increases with the dimension $d$ of $\textbf{x}$ and decreases with the number of quantization levels $s$. %This means that quantization for a vector with a higher dimension causes a worse quantization distortion, while quantization with more levels achieves preciser compression under a smaller distortion.

From Table \ref{table-comparison}, for the same degree of distortion, LM-DFL uses only $0.29s$ levels while QSGD uses $s$ levels. Compared with natural compression, under fine-grained quantization situation with a large magnitude of $s$ where $d/s^2\ll 1/8$, the distortion of LM-DFL is far less than that of natural compression. In Appendix D, LM-DFL's distortion is shown to be $(\frac{\ell_{j^*+1}/\ell_{j^*}-1}{\ell_{j^*+1}/\ell_{j^*}+1})^2$, smaller than that of ALQ (This is because of $(\ell_{j^*+1}/\ell_{j^*}+1)^2\geq 4(\ell_{j^*+1}/\ell_{j^*})$). The definition of $j^*$ is consistent with that of ALQ, i.e., $j^*=\arg\max_{1\leq j\leq s}\ell_{j+1}/\ell_j$. Therefore, LM-DFL achieves a quantizing operator with the minimal quantization distortion.

\subsection{Convergence of LM-DFL}
%First of all, we analyze the convergence of quantized DFL with a general unbiased quantizer $Q$, named as QDFL. Then we give the convergence bound of QDFL. 
Following the same form of \eqref{matrixupdate-LMDFL}, which shows LM-DFL learning strategy, the learning strategy of QDFL can be expressed as 
\begin{align}
\label{QDFLlearningrule}\textbf{X}_{k+1}=[\hat{\textbf{X}}_k+Q(\textbf{X}_{k,\tau}-\textbf{X}_{k})]\textbf{C}.
\end{align}
For convenience, we transform the update rule of QDFL in \eqref{matrixupdate-LMDFL} into model averaging. By
multiplying $\textbf{1}/N$ on both sides of \eqref{QDFLlearningrule}, we obtain 
\begin{align}
\textbf{X}_{k+1}\frac{\textbf{1}}{N}=[\hat{\textbf{X}}_k+Q(\textbf{X}_{k,\tau}-\textbf{X}_{k})]\textbf{C}\frac{\textbf{1}}{N},
\end{align}
where $\textbf{C}$ is eliminated because of $\textbf{C} \textbf{1}=\textbf{1}$ from Section \ref{assumption}. We use $\textbf{u}_k$ to denote the average model: $\textbf{u}_k\triangleq \textbf{X}_k\textbf{1}/{N},$ and use $\hat{\textbf{u}}_k$ to denote the average estimated model: $\hat{\textbf{u}}_k\triangleq \hat{\textbf{X}}_k\textbf{1}/{N}$.
Thus, we can obtain 
\begin{align}\label{DFLQ}
\textbf{u}_{k+1}=\hat{\textbf{u}}_{k}+\frac{1}{N}\sum_{i=1}^{N}Q(\textbf{x}_{k,\tau}^{(i)}-\textbf{x}_{k}^{(i)}).
\end{align}
Using \eqref{eti1} and the unbiased quantized $Q(\cdot)$, taking an expectation for \eqref{eti1} can obtain
\begin{align}
\mathbb{E}\hat{\textbf{X}}_k-\mathbb{E}\hat{\textbf{X}}_{k-1}=\textbf{X}_k-\textbf{X}_{k-1}.
\end{align}
Therefore,  from $l=2$ to $l=k$, one has
\begin{align}\label{est2}
\mathbb{E}\hat{\textbf{X}}_k-\mathbb{E}\hat{\textbf{X}}_{1}=\textbf{X}_k-\textbf{X}_{1}.
\end{align}
In \eqref{eti1}, $k=0$ is undefined and we set the parameter matrices $\hat{\textbf{X}}=\textbf{0},\textbf{X}_0=\textbf{0},\textbf{X}_{0,\tau}=0$. So $\hat{\textbf{X}}_1=Q(\textbf{X}_1)$, and $\mathbb{E}\hat{\textbf{X}}_1=\mathbb{E}Q(\textbf{X}_1)=\textbf{X}_1$. From \eqref{est2}, we have
\begin{align}\label{est3}
	\mathbb{E}\hat{\textbf{X}}_k=\textbf{X}_k.
\end{align}
By multiplying $\textbf{1}/N$ on both sides of \eqref{est3}, we further have
\begin{align}\label{est4}
\mathbb{E}\hat{\textbf{u}}_k=\textbf{u}_k.
\end{align}

In the remaining analysis, we will focus on the convergence of the average model $\textbf{u}_k$,
commonly used to analyze convergence under stochastic sampling setting of gradient descent \cite{yu2019parallel,lian2017can,wang2021cooperative,stich2018local}.

We derive the convergence rate of QDFL based on the distortion definition \ref{distor} for any quantizer.

\begin{lemma}[\textbf{Convergence of QDFL}]\label{the1}
	Consider the average models over iteration $K$ according to the QDFL with any quantizer outlined in \eqref{DFLQ}. Suppose the conditions 1-5 in Assumption \ref{ass1} are satisfied. If the learning rate $\eta$ satisfies,
	\begin{align}
	\label{con-eta} \eta\leq \frac{\sqrt{(\omega+N)^2+4N^2(2\alpha+1)}-\omega-N}{2NL\tau(2\alpha+1)} ,
	\end{align}
	then the expectation of the gradient norm average after $K$ iterations is bounded as
	\begin{align}\label{QDFLcon}
	\notag\mathbb{E}\left[\frac{1}{K}\sum_{k=1}^{K}\|\nabla F(\textup{\textbf{u}}_k)\|^2\right]\leq\frac{2[F(\textup{\textbf{u}}_1)-F_{\textup{inf}}]}{\eta K\tau}+\frac{L\eta\tau\sigma^2(\omega+N)}{N}\\
	+\left(2\alpha +\frac{2}{3}\right)L^2\eta^2\sigma^2\tau^2+\delta^2,
	\end{align}
	where $\alpha=\frac{\zeta^2}{1-\zeta^2}+\frac{\zeta}{(1-\zeta)^2}.$
\end{lemma}
\begin{proof}
	The specific proof the theorem is presented in Appendix B.  %of Supplemental Material.
\end{proof}
From Lemma \ref{the1}, the convergence rate of QDFL is $\mathcal{O}(1/K)$. The convergence is bounded by $\frac{L\eta\tau\sigma^2(\omega+N)}{N}+\left(2\alpha +\frac{2}{3}\right)L^2\eta^2\sigma^2\tau^2+\delta^2$ when $K\to\infty$ for a fixed learning rate $\eta$. This bound reflects the influence of gradient estimation variance $\sigma^2$ and gradient divergence $\delta^2$.

In the existing convergence analysis on decentralized SGD frameworks 
in \cite{koloskova2019decentralized}, the convergence bound is expressed with the order of approximation based on the assumption of strong convexity,  and considers decentralized SGD with only a single local update between two inter-node communications. We provide a deterministic convergence bound of QDFL with $\tau$ local updates between two inter-node communications, which is a common improvement for communication-efficiency in FL area \cite{wang2019adaptive,reisizadeh2020fedpaq,liu2020accelerating,wang2021cooperative}. Moreover, the convergence guarantees of the QDFL do not need the convex loss function assumption.

Based on the general result of QDFL in Lemma \ref{the1}, the convergence bound of LM-DFL is as follows.

\begin{theorem}[\textbf{Convergence of LM-DFL}]\label{LMDFLconvergence}
	Consider the average models over iteration $K$ according to the LM-DFL method outlined in Algorithm \ref{alg_2}. Suppose the conditions 1-5 in Assumption \ref{ass1} are satisfied with i.i.d data distribution ($\delta=0$). If the learning rate is $\eta=1/L\sqrt{K}$ with large $K$,
	then the expectation of the gradient norm average after $K$ iterations is bounded as
	\begin{align}\label{LMDFLcon}
	\notag\mathbb{E}\left[\frac{1}{K}\sum_{k=1}^{K}\|\nabla F(\textup{\textbf{u}}_k)\|^2\right]\leq\frac{2L[F(\textup{\textbf{u}}_1)-F_{\textup{inf}}]}{\tau\sqrt{K}}+\frac{\tau\sigma^2d}{12s^2N\sqrt{K}}\\
	+\frac{\tau\sigma^2}{\sqrt{K}}+\left(2\alpha+\frac{2}{3}\right)\frac{\sigma^2\tau^2}{K},
	\end{align}
	where $\alpha=\frac{\zeta^2}{1-\zeta^2}+\frac{\zeta}{(1-\zeta)^2}.$
\end{theorem}
\begin{proof}
	The result can be derived from Lemma \ref{the1} straightforwardly.
\end{proof}
According to Theorem \ref{LMDFLconvergence}, we analyze the convergence rate of LM-DFL and convergence dependence on LM-DFL topology and quantization distortion, which are presented in the following remarks respectively. 
\newtheorem{remark}{Remark}

\begin{remark}[\textbf{Influence of Quantization Distortion}]\label{re3}
	\textup{Theorem \ref{LMDFLconvergence} implies that the convergence bound is influenced by quantization distortion, which is dependent on $d,s$. The quantization distortion of LM-DFL increases with gradient dimension $d$ and decreases with the number of levels $s$. The convergence bound increases with the quantization distortion from the second and third terms in \eqref{LMDFLcon}. Therefore a larger gradient dimension $d$ will
	make convergence of LM-DFL worse, while a larger number of levels $s$ will make convergence of LM-DFL better. Increasing quantization distortion will lift convergence bound, degrading convergence performance.}
\end{remark}

\begin{remark}[\textbf{Order-wise Convergence Rate}]
	\textup{The result in Theorem \ref{LMDFLconvergence} implies the following order-wise rate
	\begin{align*}
	\notag\mathbb{E}\left[\frac{1}{K}\sum_{k=1}^{K}\|\nabla F(\textup{\textbf{u}}_k)\|^2\right]\leq\mathcal{O}\left(\frac{1}{\sqrt{K}}\right)+\mathcal{O}\left(\frac{1}{K}\right).
	\end{align*}
	When $K$ is large enough, the $1/K$ term will be dominated by the $1/\sqrt{K}$ term.  LM-DFL achieves the convergence rate of $\mathcal{O}(1/\sqrt{K})$ with i.i.d data distribution for non-convex loss functions. Thus, LM-DFL recovers the convergence rate of Cooperative SGD \textup{\cite{wang2021cooperative} \cite{lian2017can}}.}
\end{remark}

\begin{remark}[\textbf{Influence of LM-DFL Topology}]\label{topology}
	\textup{Theorem \ref{LMDFLconvergence} states that the convergence bound is dependent on the network topology of LM-DFL. We can see $\alpha$ is increasing with the second largest absolute eigenvalue $\zeta$. Then the convergence bound increases with the second largest absolute eigenvalue $\zeta$.
	Note that $\zeta$ reflects the exchange rate of all local models during a single inter-node communication step. A larger $\zeta$ means a sparser matrix, resulting to a worse model consensus. When $\zeta=1$, the confusion matrix $\textup{\textbf{C}}=\textup{\textbf{I}}$ where each node can not communicate with any other nodes in the network. When $\zeta=0$,  
	the confusion matrix $\textup{\textbf{C}}=\textup{\textbf{J}}$, which nodes communicate differential model parameters with all other nodes. 
	Therefore, we can conclude that sparser network topology leads to worse convergence.}
\end{remark}

\section{Doubly-Adaptive DFL}
In order to improve the communication efficiency of LM-DFL, we propose a doubly-adaptive DFL to adapt the number of quantization levels $s$ in the training. 

Based on LM-DFL convergence upper bound, we try to optimize the number of quantization levels $s_k$ in $k$-th iteration by minimizing the convergence bound. The result shows doubly-adaptive DFL with ascending number of quantization levels $s$ can achieve a given targeted training loss using much fewer communicated bits. In Fig. \ref{fig_4}, ascending $s$ shows the best convergence performance. 

To find out how the number of quantization levels $s$ affects the convergence bound of LM-DFL,  
we define the number of bits communicated between any node $i$ and $j$ over training, with $i,j=1,2,...,N$, $i\neq j$ and $c_{ij}\neq0$, as $B$.%It is a common to analyze the convergence bound versus number of bits communicated in centralized FL \cite{jhunjhunwala2021adaptive} \cite{mo2022feddq}.

\begin{figure}[!t]
	\centering
	\includegraphics[scale=0.5]{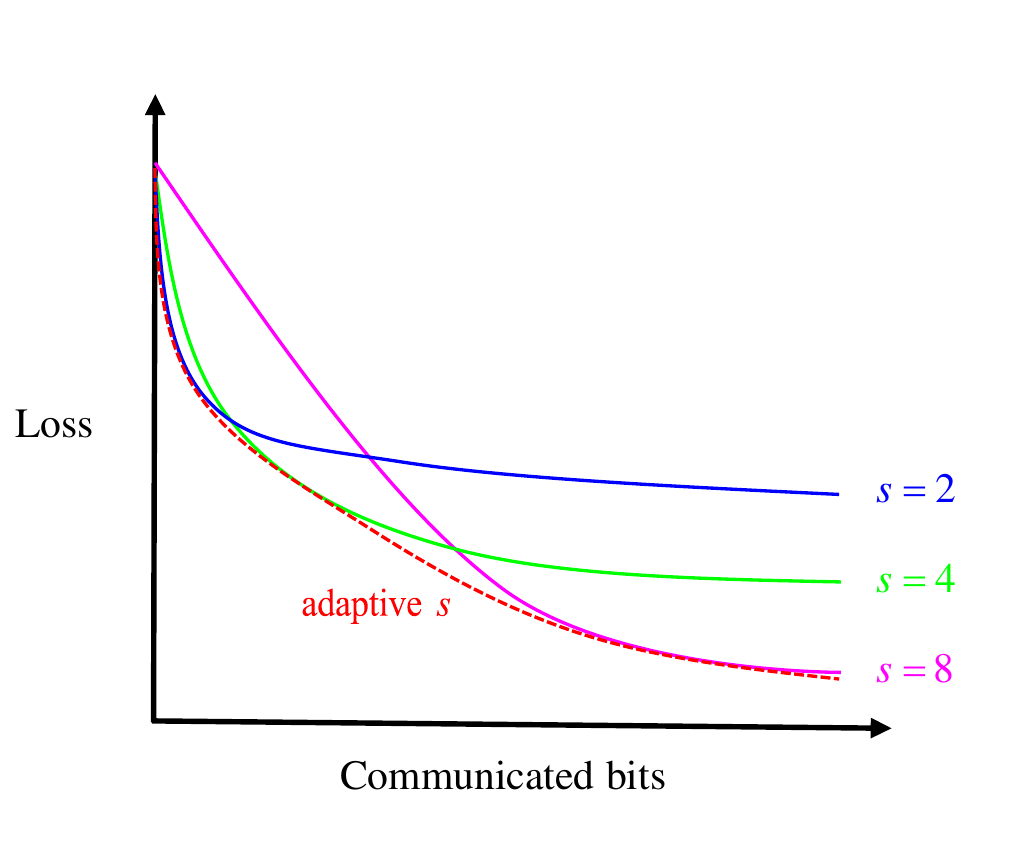}        %这个是在LaTeX文件夹中的相对路
	\caption{Training loss versus communicated bits under adaptive $s$ and fixed $s$.}
	\label{fig_4}
\end{figure}

%\begin{definition}[\textbf{Number of Bits Communicated}]
%	The total number of bits communicated by a node $i$ to node $j$ at a given time point in training process is denoted by $B$.
%\end{definition}
Because LM-DFL performs synchronized local updates and inter-node communications and each node perform LM vector quantization with the same $s$ in iteration $k$, the $B$ is the same for all nodes. Then we have $K=B/2C_s$ where $C_s$ is the number of bits communicated  by a node $i$ to node $j$ in an inter-node communication stage, defined in \eqref{numberofbits}. We can obtain the convergence bound of LM-DFL versus $s$ as follows.

\begin{theorem}[\textbf{LM-DFL Convergence versus} \textit{s}]\label{thes}
    The expectation of the gradient norm average of LM-DFL is bound by
	\begin{align}
	\notag\mathbb{E}\left[\frac{2C_s}{B}\sum_{k=1}^{B/2C_s}\|\nabla F(\textup{\textbf{u}}_k)\|^2\right]\leq
	A_1\log_2(2s)+\frac{A_2}{s^2}+A_3,
	%\frac{2C_s[F(\textup{\textbf{u}}_1)-F_{\textup{inf}}]}{\eta B\tau}+\\\frac{4C_sL^2\alpha\eta^2\tau\sigma^2}{NB}+
	%\frac{L\eta\tau\sigma^2(\omega+N)}{N}+\frac{2}{3}L^2\eta^2\sigma^2\tau^2+\delta^2,
	\end{align}
	where 
	\begin{align*}
	A_1&=\frac{4[F(\textup{\textbf{u}}_1)-F_{\textup{inf}}]d}{\eta\tau B},\ A_2=\frac{L\eta\tau\sigma^2d}{12N},\\
	A_3&=\frac{A_1}{d}(d+32)+(2\alpha+\frac{2}{3})L^2\eta^2\sigma^2\tau^2+\delta^2+L\eta\tau\sigma^2.
	\end{align*}
\end{theorem}
\begin{proof}
	The detailed proof is presented in Appendix C. % of Supplemental Material.
\end{proof}

Then we can obtain the optimal $s$ to minimize the convergence bound presented in Theorem \ref{thes}. By differentiating $A_1\log_2(2s)+\frac{A_2}{s^2}+A_3$ with respect to $s$, one has 
 \begin{align*}
 s^*=\sqrt{\frac{A_4}{A_5[F(\textup{\textbf{u}}_1)-F_{\textup{inf}}]}}.
 \end{align*}
where $A_4=L\eta^2\tau^2\sigma^2B$, and $A_5=24N^2\log_2e$. %Note that $A_6/A_4=\frac{72L\eta\zeta^2\log_2e}{B(1-\zeta^2)}$, where $\eta L\leq1$ and communicated bits $B=KC_s$ is sufficiently large. Thus $A_6/A_4\approx0$ and $$s^*\approx\sqrt{\frac{A_4}{A_5[F(\textup{\textbf{u}}_1)-F_{\textup{inf}}]}}.$$ 

\begin{figure}[!t]
	\centering
	\includegraphics[scale=0.5]{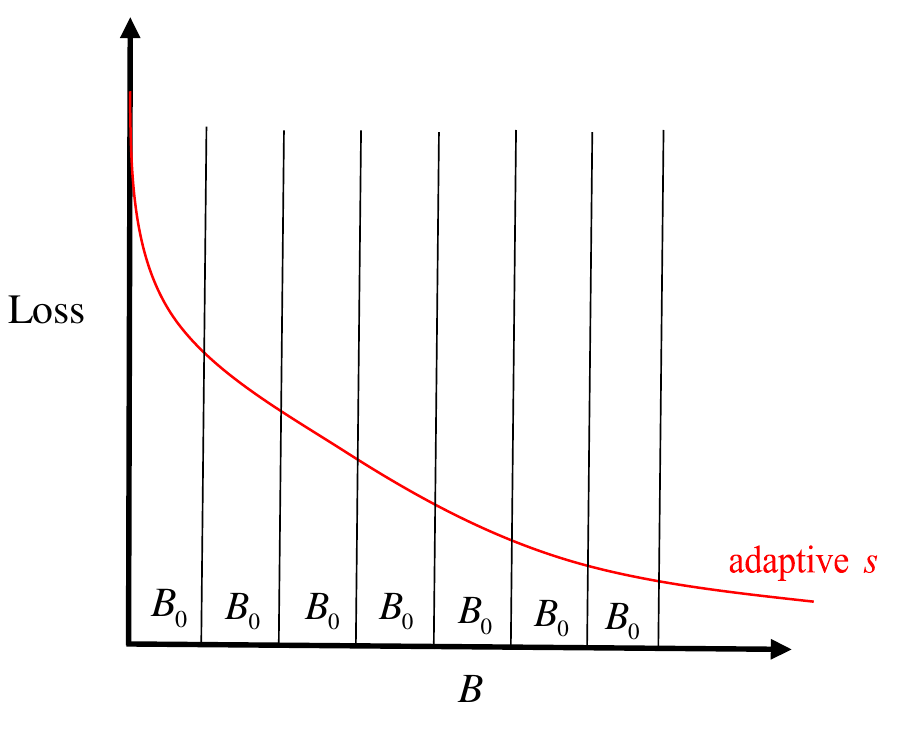}        %这个是在LaTeX文件夹中的相对路
	\caption{The intervals of communicated bits}
	\label{fig_5}
\end{figure}

\begin{algorithm}[!t]
	\caption{\textbf{Doubly-Adaptive DFL}} 
	\label{alg_3}
	\begin{algorithmic}[1]
		\REQUIRE ~~\\ %算法的输入参数：Input
		Learning rate $\eta$\\
		Total number of iterations $K$\\
		Number of local update $\tau$ in an iteration\\
		Confusion matrix $\textbf{C}$
		
		\STATE Set the initial value of $\textbf{X}_{1,0}$. 
		\STATE Set the initial number of quantization levels $s_1$ for all nodes.
		
		\FOR{$k=1,2,...,K$}
		\FOR{$t=0,1,...,\tau-1$} 
		\STATE Nodes calculate stochastic gradient
		\STATE Nodes perform SGD in parallel: \\ 
		$\textbf{X}_{k,t+1}=\textbf{X}_{k,t}-\eta\textbf{G}_{k,t}$ \hfill$//$\textit{local updates}
		\ENDFOR
		\STATE Each node $i$ evaluates $s_k^{(i)}$ according to $s_k^{(i)}\!=\!\sqrt{\frac{F_i(\textbf{x}_1^{(i)})}{F_i(\textbf{x}_k^{(i)})}}$.
		\STATE Each node $i$ calculates $\textbf{x}_{k,\tau}^{(i)}-\textbf{x}_{k}^{(i)}$, $\textbf{x}_{k}^{(i)}-\textbf{x}_{k-1,\tau}^{(i)}$ and computes the statistics to construct its probability density function $\phi_k^{(i)}(x)$.
		\STATE Each node $i$ quantizes the differential model parameters to obtain $Q_{L}(\textbf{x}_{k,\tau}^{(i)}-\textbf{x}_{k}^{(i)})$ and
        $Q_{L}(\textbf{x}_{k}^{(i)}-\textbf{x}_{k-1,\tau}^{(i)})$ by LM vector quantizer.
%performs quantization for differential model parameter by using LM vector quantizer based on  $\phi_k^{(i)}(x)$:\\
%$\quad Q_{L}:\textbf{x}_{k,\tau}^{(i)}-\textbf{x}_{k}^{(i)} \to Q_{L}(\textbf{x}_{k,\tau}^{(i)}-\textbf{x}_{k}^{(i)})$.
        \STATE Nodes exchange the quantized differential model parameters and calculate the model estimated parameters
        $$	\hat{\textbf{X}}_k=\hat{\textbf{X}}_{k-1}+Q_L(\textbf{X}_{k-1,\tau}-\textbf{X}_{k-1})+Q_L(\textbf{X}_k-\textbf{X}_{k-1,\tau}).$$
        \STATE Then update the initial model parameter for next iteration by \\$\quad\textbf{X}_{k+1}\!=[\hat{\textbf{X}}_k\!+\!Q_L(\textbf{X}_{k,\tau}\!-\!\textbf{X}_{k})]\textbf{C}$. \hfill$//$\textit{communication}
		\ENDFOR
	\end{algorithmic}
\end{algorithm}
\begin{figure*}[!t]
	\centering
	\subfloat[MNIST: Training loss]{\includegraphics[scale=0.35]{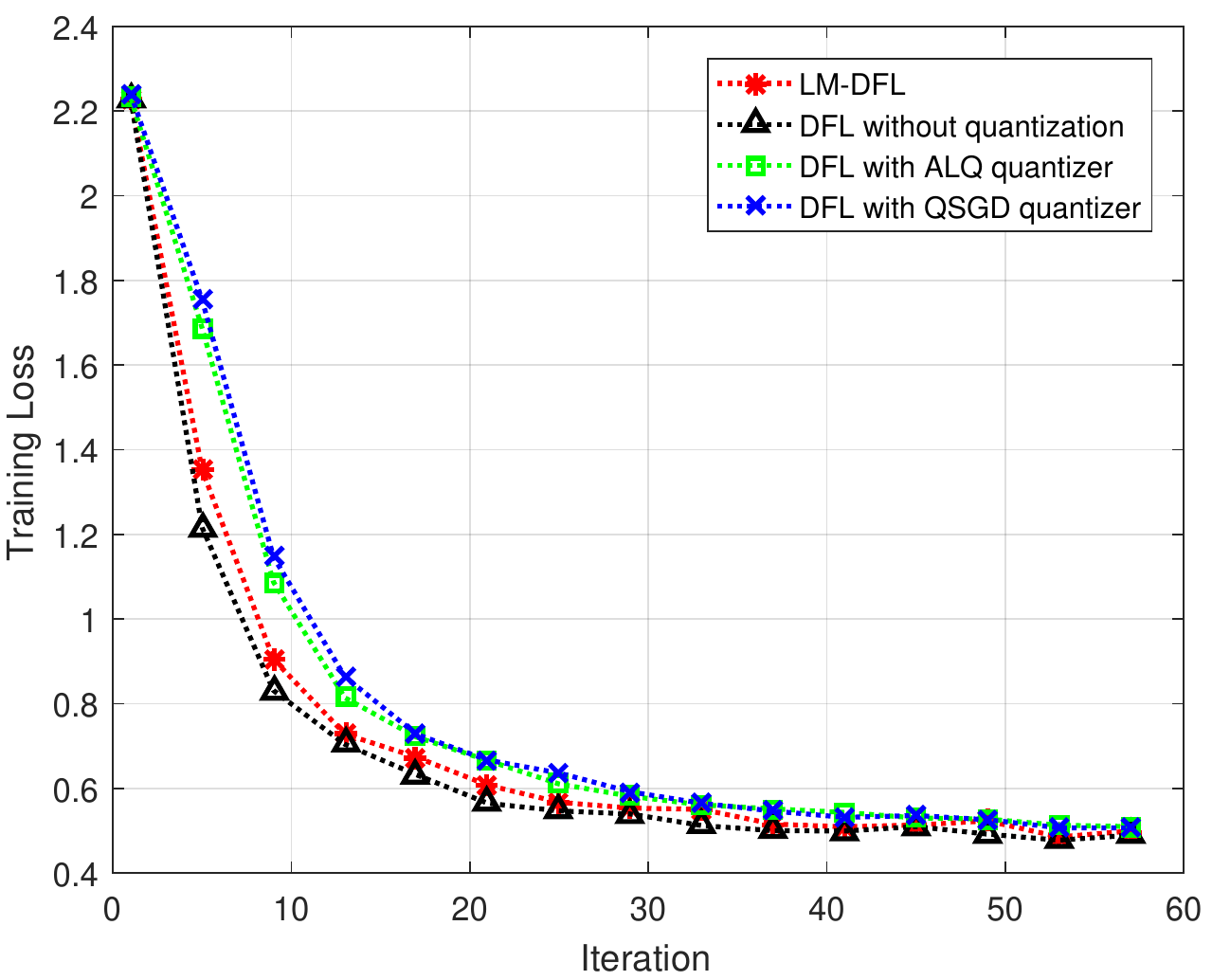}\label{loss_mnist}}
	\subfloat[MNIST: Time Progression]{\includegraphics[scale=0.35]{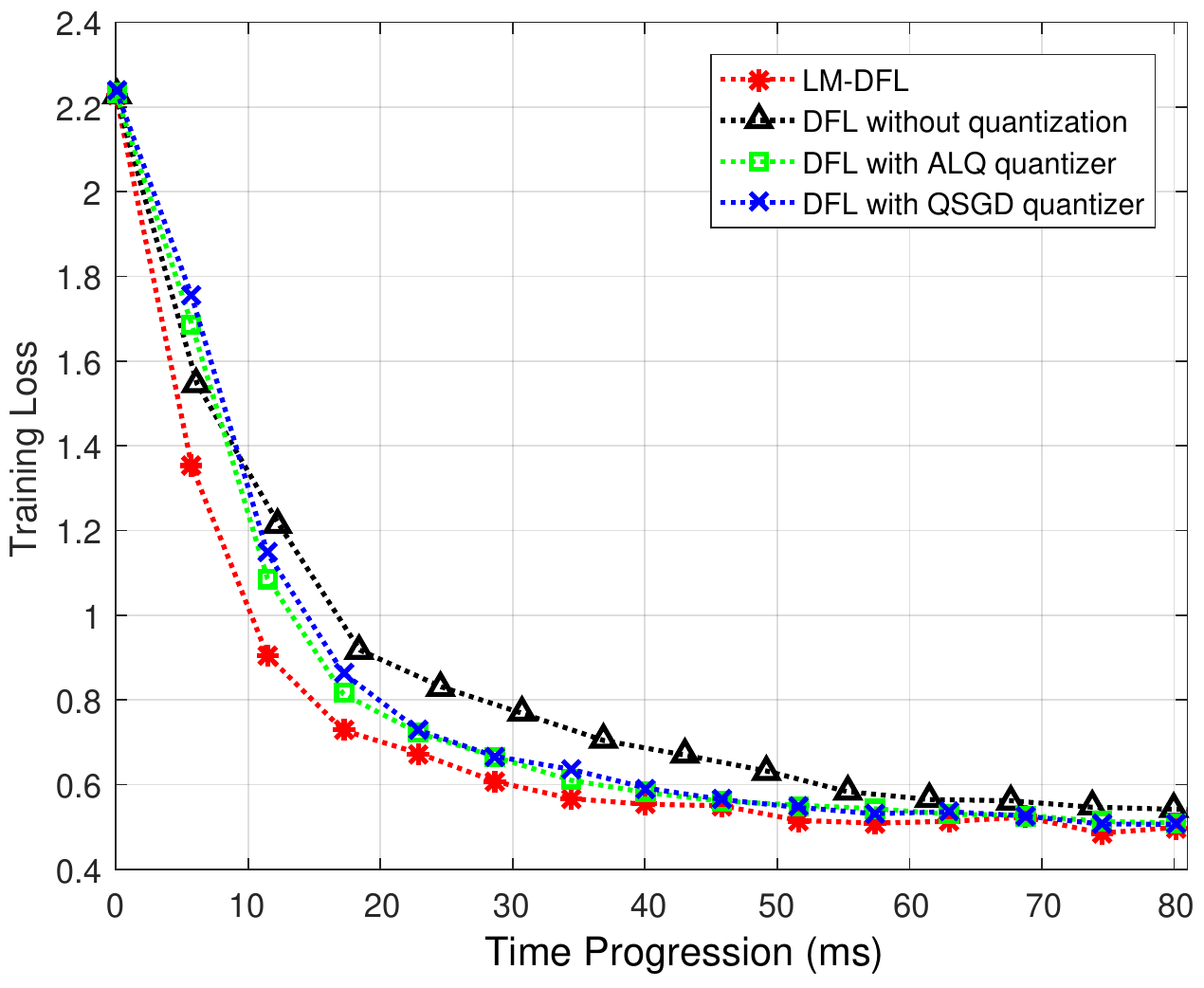}\label{bits_mnist}}
	\subfloat[MNIST: Testing accuracy]{\includegraphics[scale=0.35]{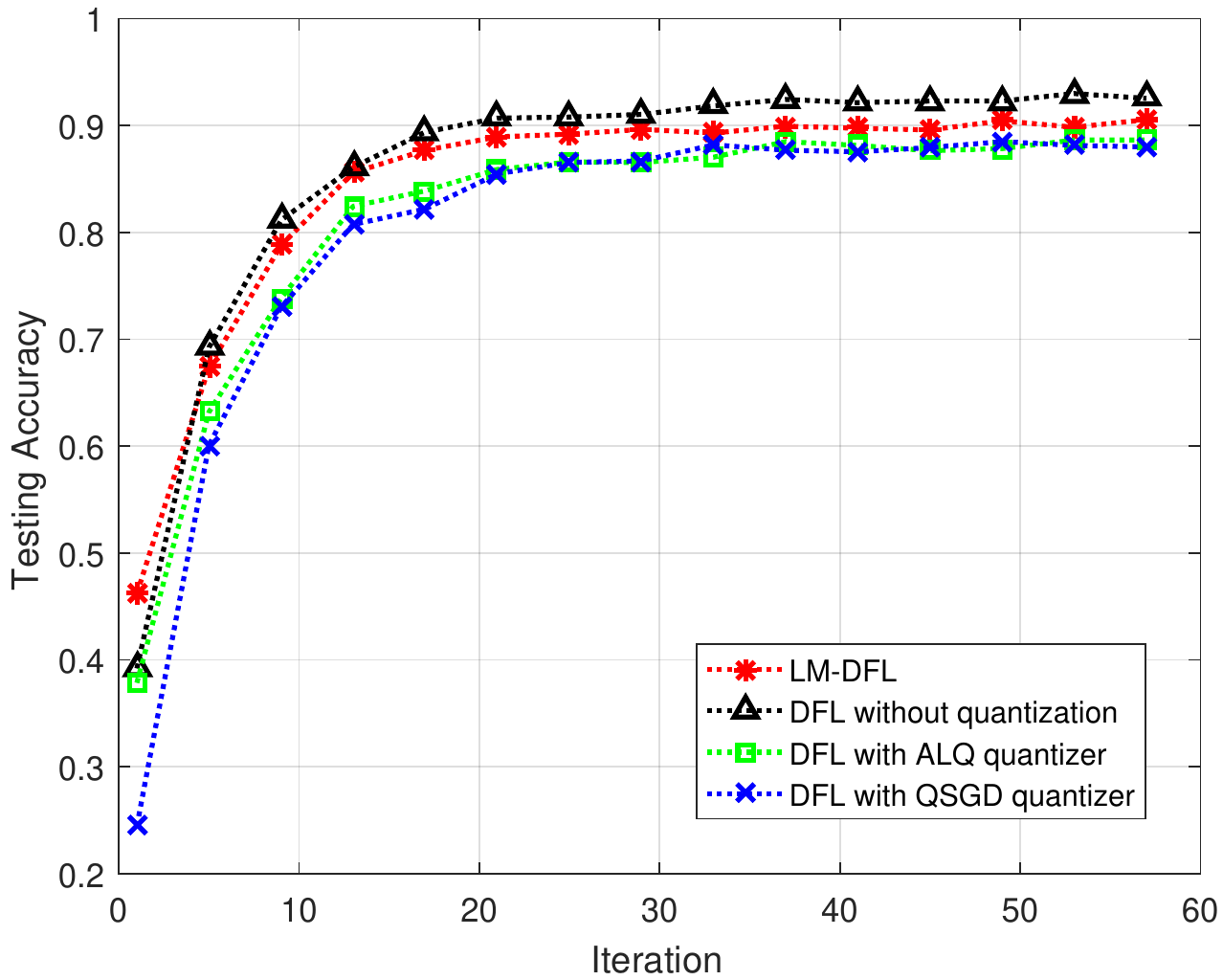}\label{acc_mnist}}
	\subfloat[MNIST: Quantization distortion]{\includegraphics[scale=0.35]{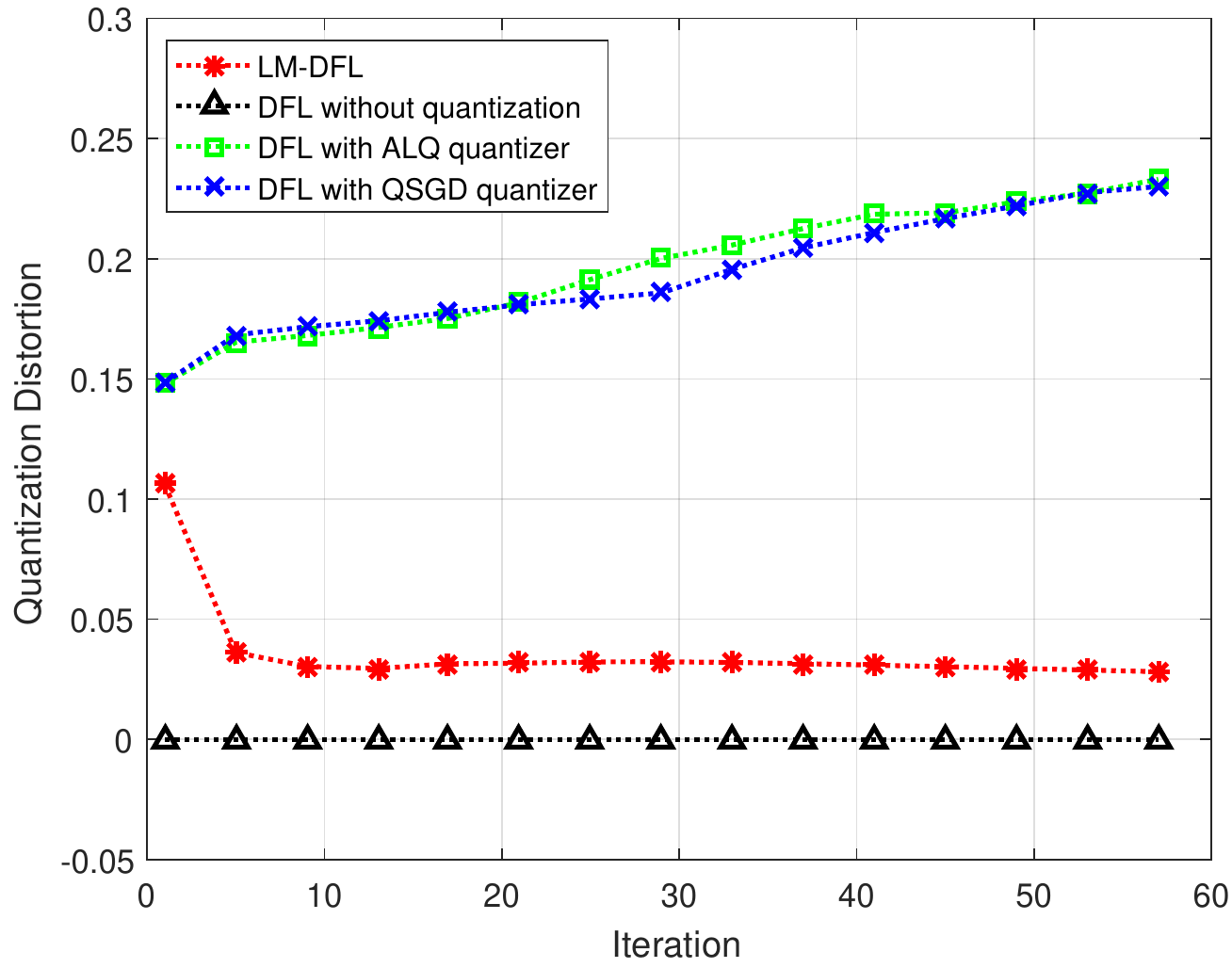}\label{distortion_mnist}}
	\\
	\subfloat[CIFAR-10: Training loss]{\includegraphics[scale=0.35]{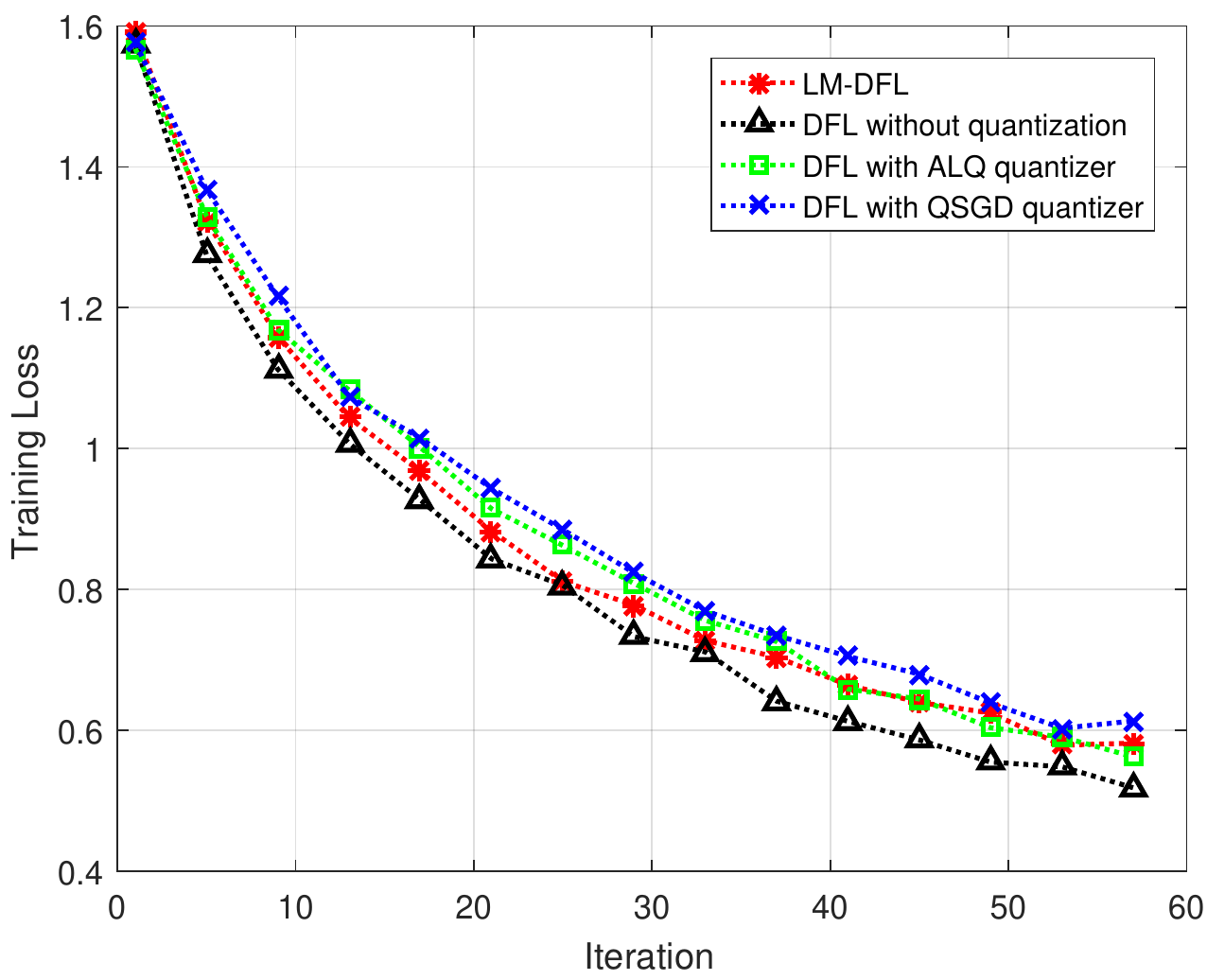}\label{loss_cifar}}
	\subfloat[CIFAR-10: Time Progression]{\includegraphics[scale=0.35]{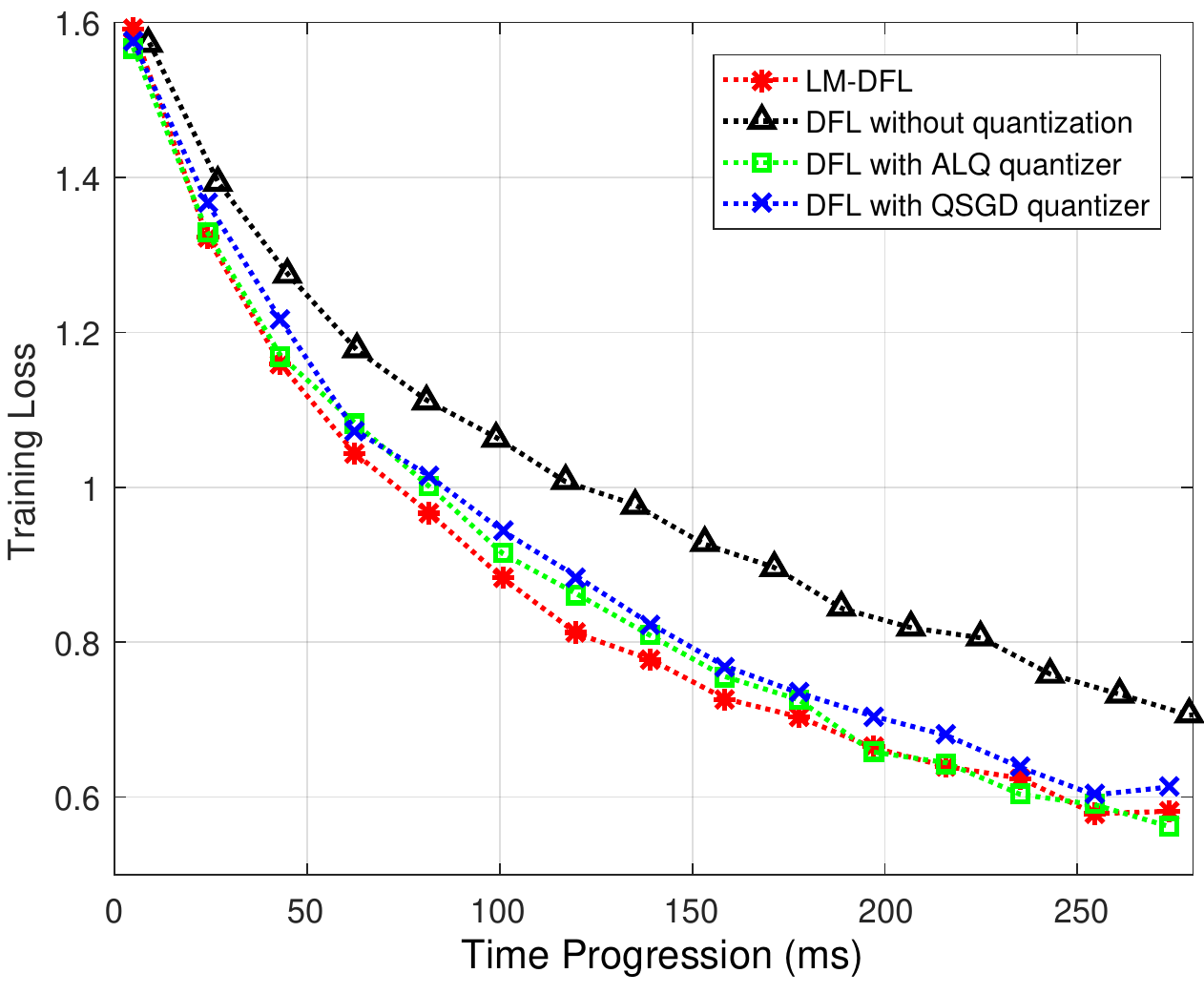}\label{bits_cifar}}
	\subfloat[CIFAR-10: Testing accuracy]{\includegraphics[scale=0.35]{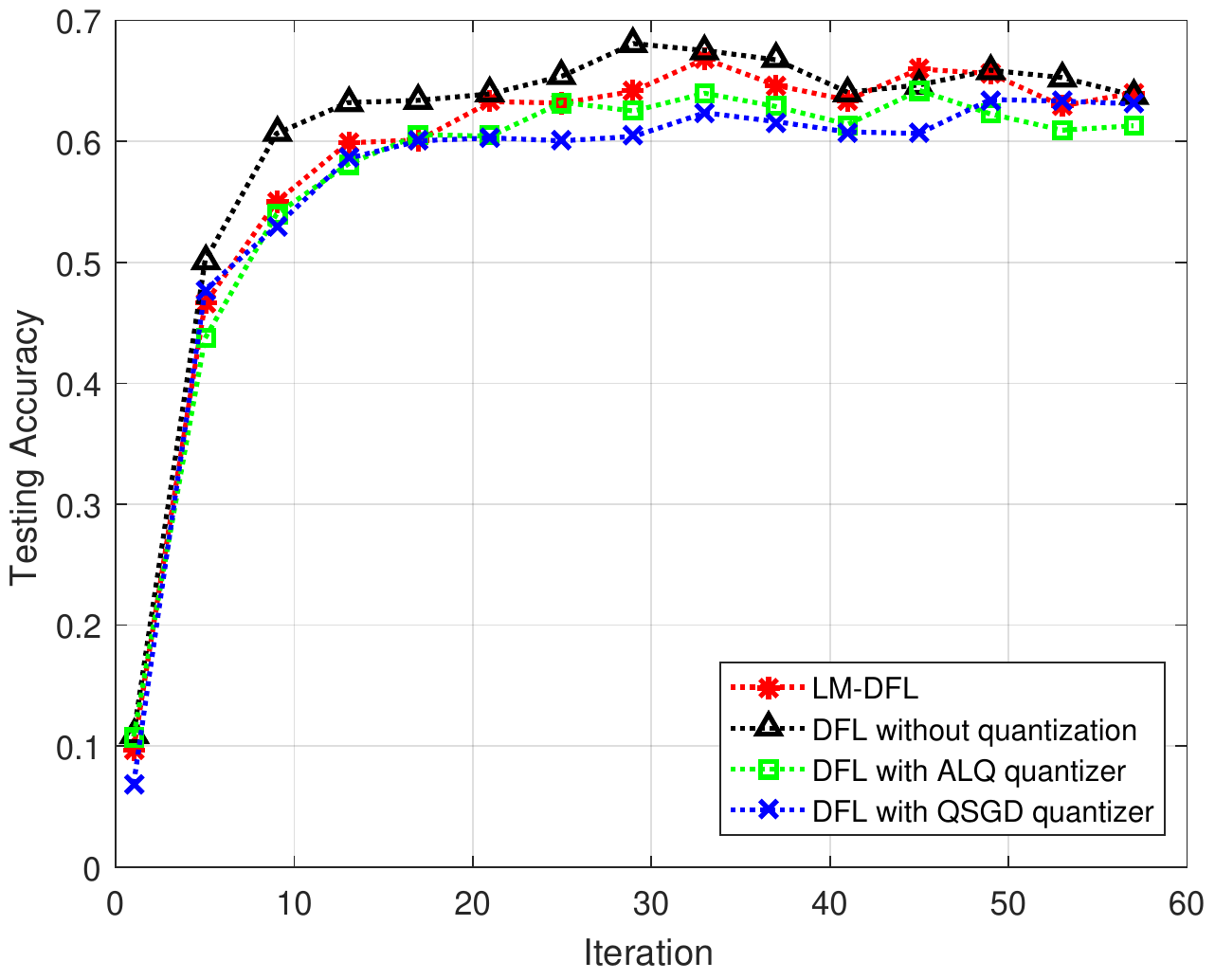}\label{acc_cifar}}
	\subfloat[CIFAR-10: Quantization distortion]{\includegraphics[scale=0.35]{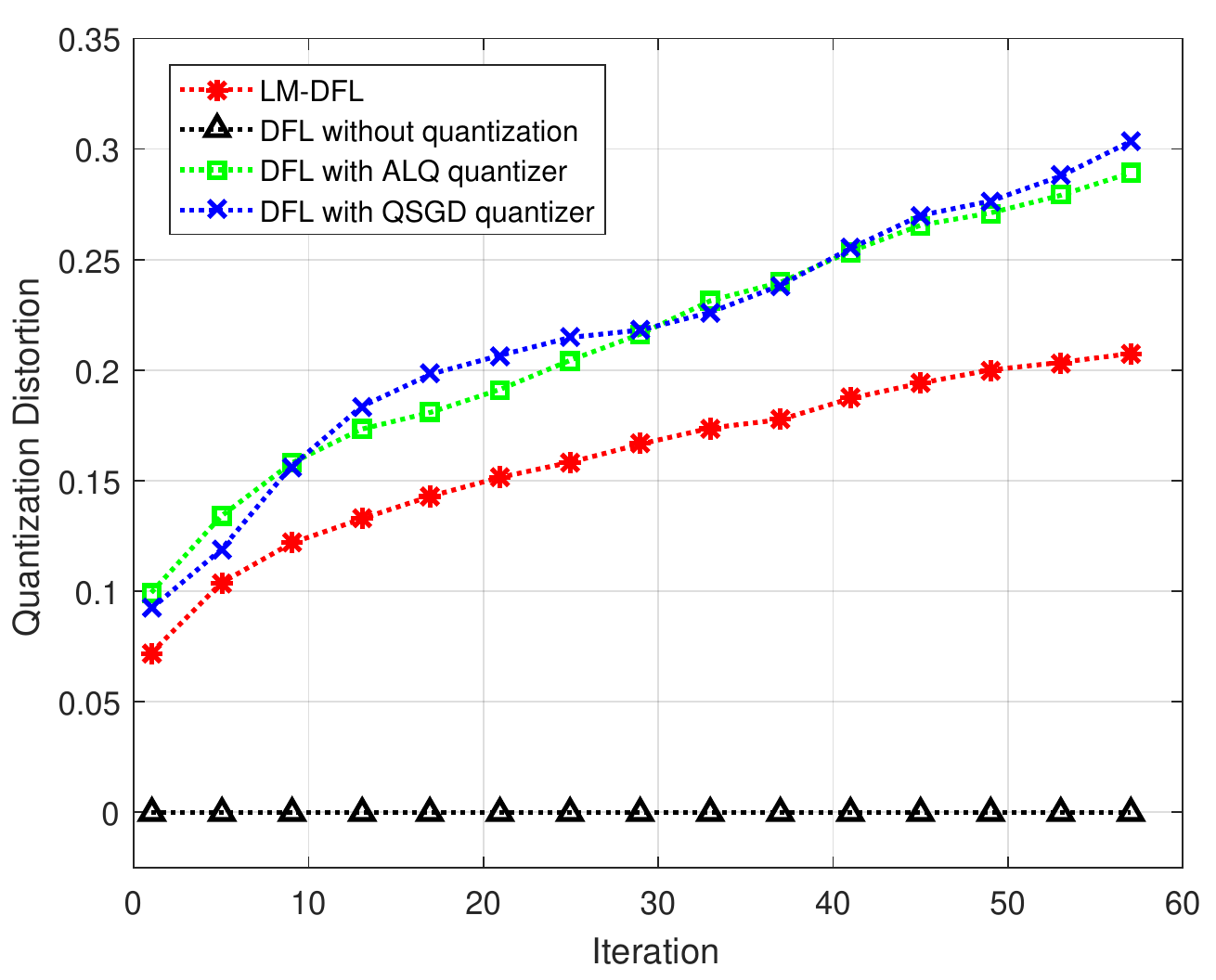}\label{distortion_cifar}}
	\caption{Experiments on MNIST and CIFAR-10 datasets with CNN under different approaches. There are 10 nodes for the DFL topology. We set the number of local updates $\tau=4$ in an iteration. For the training on MNIST, we set the learning rate $\eta=0.002$. For the training on CIFAR-10, we set the learning rate $\eta=0.001$.}
	\label{convergencecurve}
\end{figure*}
We partition the entire training process into many uniform communication intervals as presented in Fig. \ref{fig_5}. In each interval, node $i$ communicates $B_0$ bits with node $j$. Therefore, %if we regard  the $k$-th iteration as the beginning of the training in an interval of communication bits. For this fragment of the training, 
the optimal number of quantization levels of iteration $k$ is
\begin{align}\label{s_k}
s_k=\sqrt{\frac{A_4^0}{A_5[F(\textup{\textbf{u}}_k)-F_{\textup{inf}}]}}
\end{align}
where $A_4^0=L\eta^2\tau^2\sigma^2B_0$. Therefore, according to \eqref{s_k} and suppose the minimal value of $F(\textbf{x})$ is zero, i.e., $F_\textup{inf}=0$, we obtain 
\begin{align}\label{bests}
s_k\approx\sqrt{\frac{F(\textup{\textbf{u}}_1)}{F(\textup{\textbf{u}}_k)}}s_1.
\end{align}
From \eqref{bests}, doubly-adaptive DFL adopts ascending number of quantization levels for optimal convergence. Compared with the methods using fixed number of quantization levels,  doubly-adaptive DFL generates less communication bits to reach the same convergence performance. The intuition is the beginning of training presents fast descent of loss function, where coarse-grained quantizing operation is enough for model exchange. When the training process is almost convergent, fine-grained quantization is needed for any potential gradient descent to achieve a high generalizing performance of DFL network.

The specific operating steps of doubly-adaptive DFL are presented in Algorithm \ref{alg_3}. 
We summarize the two adaptive aspects of doubly-adaptive DFL as follows:
\begin{enumerate}
	\setlength{\itemsep}{0.5ex}
	\item[{1.}] \textbf{Adaptive number of quantization levels}  $s_k$: In order to achieve a given targeted convergence with much fewer communicated bits, doubly-adaptive DFL adopts the adaptive number of quantization levels $s_k\approx\sqrt{\frac{F(\textup{\textbf{u}}_1)}{F(\textup{\textbf{u}}_k)}}s_1$. Note that in Algorithm \ref{alg_3}, we evaluate the adaptive $s_k^{(i)}$ in each node by using local model $\textbf{x}_k^{(i)}$ and local loss function $F_i(\textbf{x}_k^{(i)})$ because the averaging model $\textbf{u}_k$ and the global loss function $F(\textbf{u}_k)$ can not be observed in a DFL framework.
	\item[{2.}] \textbf{Adaptive sequence of quantization levels} $\boldsymbol{\ell}_k^{(i)}$: In order to minimize quantization distortion, doubly-adaptive DFL uses LM vector quantizer for quantization of differential model parameters based on $s_k$.
\end{enumerate}
The convergence of the proposed double-adaptive DFL can be guaranteed, because the proposed double-adaptive DFL can be regarded as a special case of QDFL. The convergence bound is hard to derive due to variable learning rate. Thus, we first extend the convergence condition in Definition \ref{def1} from fixed learning rate to variable learning rate, i.e.,

	\begin{definition}[\textbf{Convergence Condition with Variable Learning Rate}]\label{def6}
 Consider the variable learning rate, the algorithm converges to a stationary point have changed from (\ref{defF}) into 

 \begin{align}
\mathbb{E}\left[ {\frac{{\sum\limits_{k = 1}^K {{\eta _k}{{\left\| {\nabla F\left( {{u_k}} \right)} \right\|}^2}} }}{{\sum\limits_{k = 1}^K {{\eta _k}} }}} \right]\leq\epsilon.
 \end{align}
	\end{definition} 

Then, we can provide the following convergence bound of doubly-adaptive DFL as follows.

	\begin{theorem}[\textbf{Convergence of Doubly-Adaptive DFL}]
Consider the average models over iteration $K$ according to the doubly-adaptive DFL method outlined in Algorithm \ref{alg_3} and the variable learning rate condition in Definition \ref{def6}. Suppose the conditions 1-5 in Assumption \ref{ass1} are satisfied with i.i.d data distribution ($\delta=0$). Adopt the adaptive number of quantization level  in (\ref{bests}), and if the variable learning rate is

\begin{align}
{\eta _k} \le \frac{{\sqrt {{{\left( {{\varpi _k} + N} \right)}^2} + 4{N^2}\left( {2\alpha  + 1} \right)}  - {\varpi _k} - N}}{{2NL\tau \left( {2\alpha  + 1} \right)}},
\end{align}
where $\alpha  = \frac{{{\zeta ^2}}}{{1 - {\zeta ^2}}} + \frac{\zeta }{{{{\left( {1 - \zeta } \right)}^2}}},{\varpi _k} = \frac{d}{{12s_k^2}}.$ Then the expectation of the gradient norm average after $K$ iterations in bounded as
\begin{equation}
  \begin{aligned}
  &\mathbb{E}\left[ {\frac{{\sum\limits_{k = 1}^K {{\eta _k}{{\left\| {\nabla F\left( {{u_k}} \right)} \right\|}^2}} }}{{\sum\limits_{k = 1}^K {{\eta _k}} }}} \right] \le \frac{{2\left[ {F\left( {{u_1}} \right) - {F_{\inf }}} \right]}}{{\tau \sum\limits_{k = 1}^K {{\eta _k}} }} + \frac{{L\tau {\sigma ^2}\sum\limits_{k = 1}^K {\eta _k^2} }}{{\sum\limits_{k = 1}^K {{\eta _k}} }} \\
  &+ \frac{{L\tau {\sigma ^2}\sum\limits_{k = 1}^K {\eta _k^2\left( {d/s_k^2} \right)} }}{{12N\sum\limits_{k = 1}^K {{\eta _k}} }} 
   + \left( {2\alpha  + \frac{2}{3}} \right){L^2}{\tau ^2}{\sigma ^2}\frac{{\sum\limits_{k = 1}^K {\eta _k^3} }}{{\sum\limits_{k = 1}^K {{\eta _k}} }}
\end{aligned}
\end{equation}
	\end{theorem}
	\begin{proof}
The detailed proof is presented in Appendix E.
	\end{proof}

\section{Simulation and discussion}
In this section, we present the simulation results of the proposed LM-DFL and doubly-adaptive DFL frameworks. %Discussions about these results are provided to verify our theoretical analysis and the effectiveness on reducing quantization distortion and improving communication efficiency.

\subsection{Setup}
In order to evaluate the performance of LM-DFL and doubly-adaptive DFL framework, we first need to build a DFL network and conduct our experiments using the built decentralized network. Using python environment and pytorch platform, we establish a DFL framework with 10 nodes and the second largest absolute eigenvalue is $\zeta=0.87$.
All 10 nodes have local datasets. Based on the local datasets, local model training is conducted. After finishing $\tau$ local updates, each node communicates its differential model parameter with its connected nodes. Note that the local models distributed in each nodes have the same structure for model aggregating and the training parameters in each nodes have the same setup. For example, LM-DFL sets the same learning rate $\eta$ and doubly-adaptive DFL sets the same number of quantization levels $s_k$ in iteration $k$.

\subsubsection{Baselines}
We compare LM-DFL with the following baseline approaches:
\begin{enumerate}
	\setlength{\itemsep}{0.5ex}
	\item[{(a)}] \textbf{DFL without quantization}: DFL without quantization is studied in \cite{lian2017can} \cite{wang2021cooperative}, where the model parameters are exchanged between connected nodes with full precision. In our experiments, we use the number of quantization distortion $s=16,000$ to achieve the full precision. Therefore, the transmission of model parameters is lossless.
	\item[{(b)}] \textbf{DFL with ALQ quantizer}: We deploy ALQ from \cite{faghri2020adaptive}, which is a centralized FL framework with adaptive quantization, into DFL framework as a baseline of LM-DFL. Coordinate descent is performed in DFL with quantizer, where the sequence of quantization levels is changed with iterations based on coordinate descent.
	\item[{(c)}] \textbf{DFL with QSGD quantizer}: We deploy the quantizer of QSGD \cite{alistarh2017qsgd}, which is a centralized FL framework, into DFL framework as a baseline of LM-DFL. DFL with QSGD quantizer performs uniform quantization and the quantized values are chosen unbiasedly.
\end{enumerate}
In the comparison, the number of quantization levels $s$ is fixed. %Then we introduce the baseline of Doubly-Adaptive DFL. 
We use DFL with QSGD quantizer under $s=4, 16, 256$ (corresponding to $2,4,8$ bits quantization respectively) as the baselines for doubly-adaptive DFL. 

\begin{figure}[!t]
	\centering
	\includegraphics[scale=0.5]{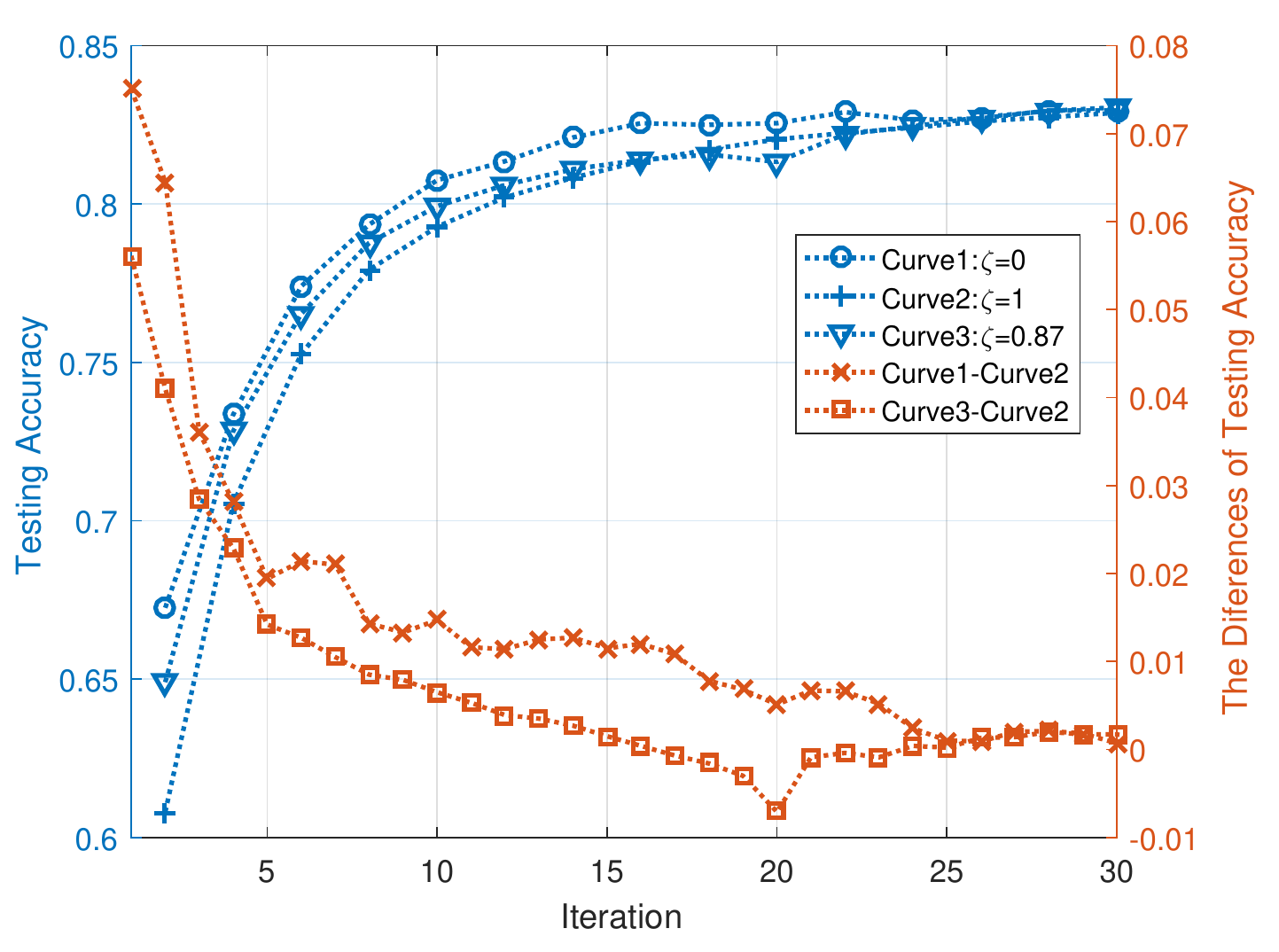}        %这个是在LaTeX文件夹中的相对路
	\caption{Testing accuracies versus iteration under different network topologies. We evaluate three topology with different connection densities, whose $\zeta$ is 0, 0.87 and 1, respectively.}
	\label{zeta-fig}
\end{figure}

\subsubsection{Models and Datasets} 
We evaluate two different model based on two different datasets. We deploy two different Convolutional Neural Networks (CNN) for model training. And these two CNN models are trained and tested using
MNIST \cite{lecun1998gradient} and CIFAR-10 \cite{krizhevsky2009learning} dataset. The MNIST dataset contains $70,000$ handwritten digits with $60,000$ for training and $10,000$ for testing. For the $j$-th sample $(\textbf{x}_j, y_j)$ in MNIST, $\textbf{x}_j$ is a $1\times28\times28$-dimensional input matrix and $y_j$ is a scalar label from $0$ to $9$ corresponding to $\textbf{x}_j$.
The CIFAR-10 dataset has 10 different types of objects, including $50,000$ color images for training and $10,000$ color images for testing. For the $j$-th sample $(\textbf{x}_j, y_j)$ in CIFAR-10, $\textbf{x}_j$ is a $3\times32\times32$-dimensional input matrix and $y_j$ is a scalar label from 0 to 9 for 10 different objects. Considering actual application scenarios of DFL where different devices have different distribution of data, we adopt non-i.i.d data distribution for model training. For half of the data samples, we allocate the data samples with the same label into a individual node. For another half of the data samples, we distribute the data samples uniformly. In the experiments of evaluating both LM-DFL and doubly-adaptive DFL, we conduct model training based on MNIST and CIFAR-10 datasets.

\begin{figure*}[!t]
	\centering
	\subfloat[MNIST: Training loss with fixed $\eta$]{\includegraphics[scale=0.43]{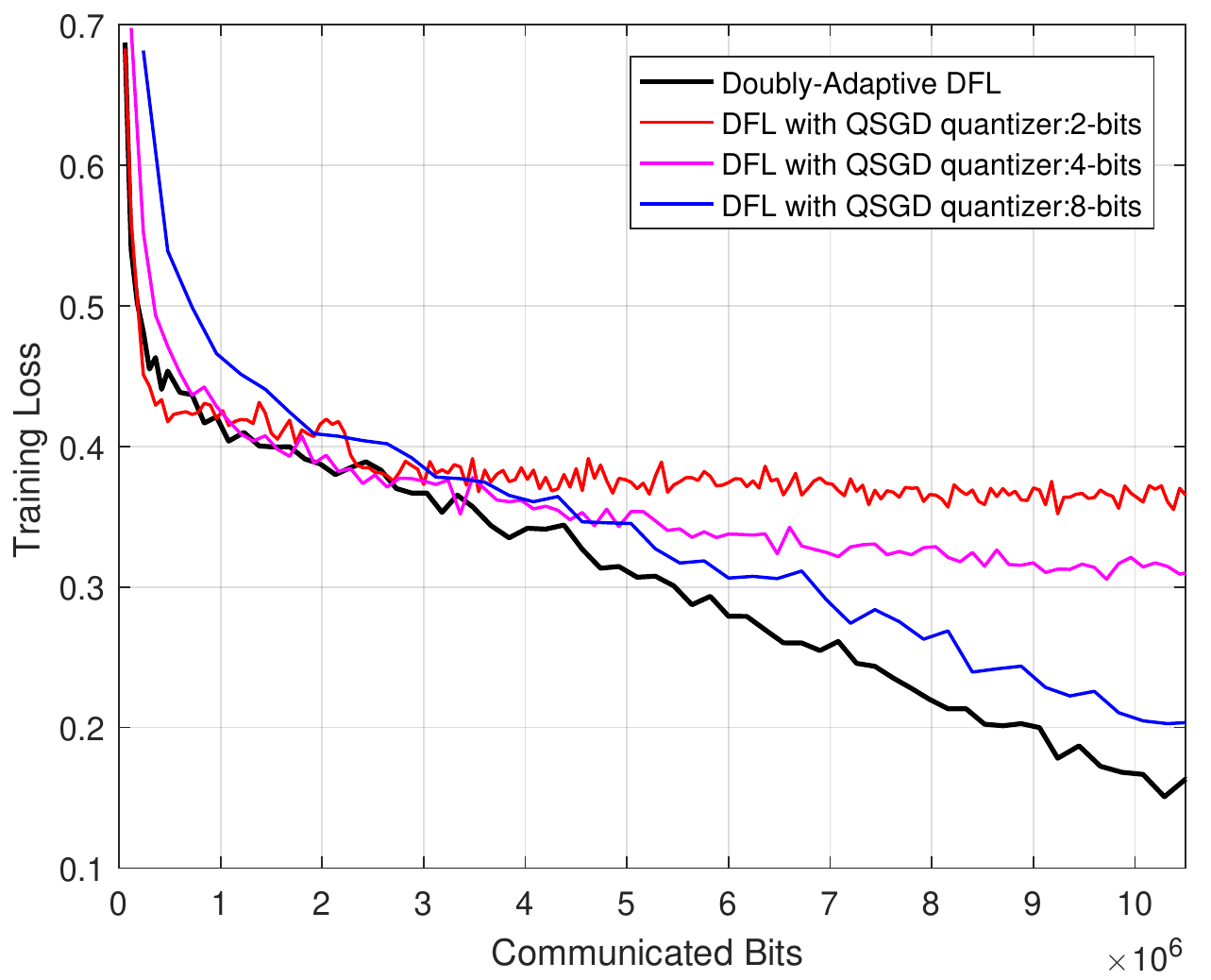}\label{eta1}}	
	\subfloat[MNIST: Training loss with variable $\eta$]{\includegraphics[scale=0.43]{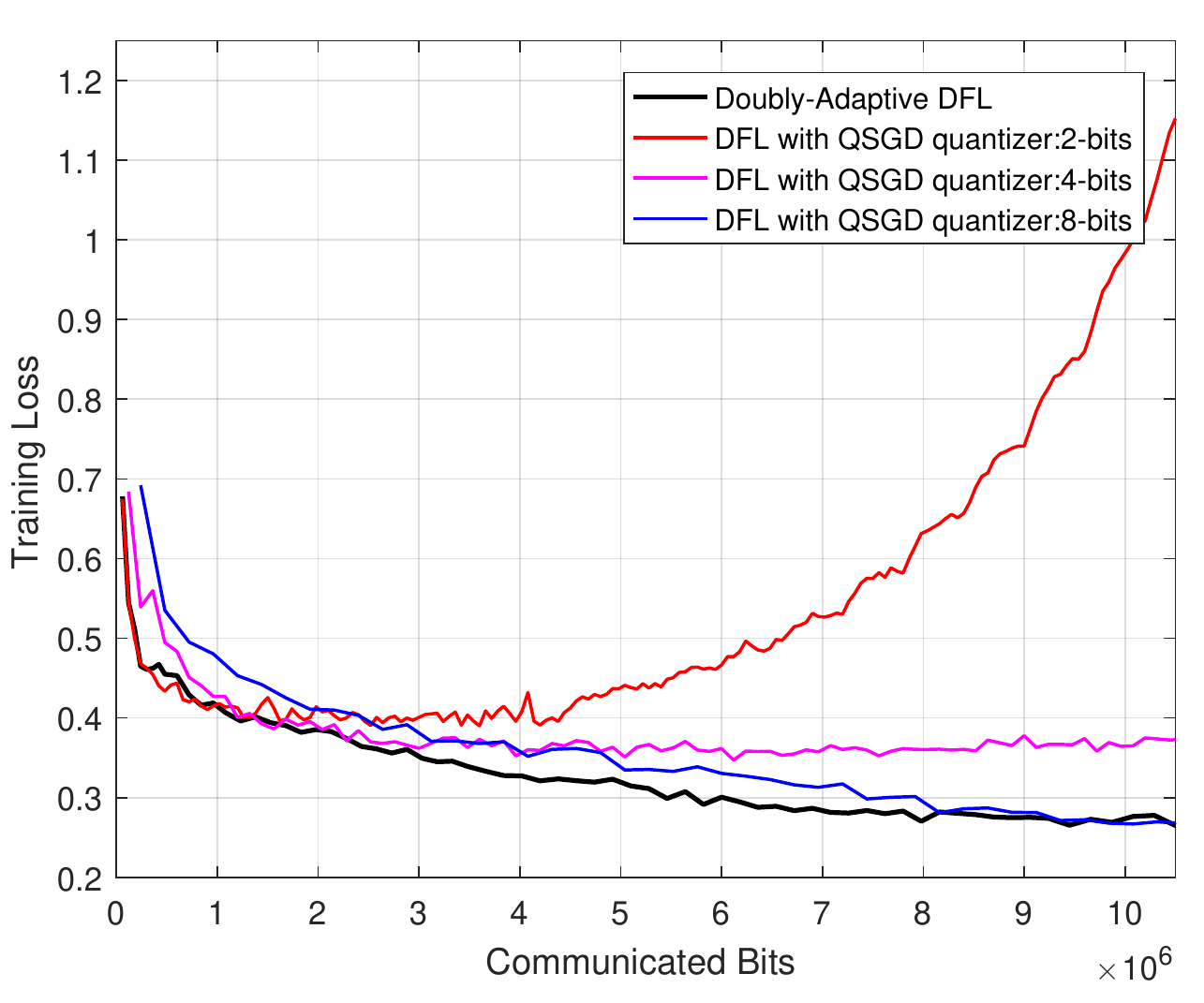}\label{eta2}}
	\subfloat[MNIST: Quantized bits for a model parameter element]{\includegraphics[scale=0.43]{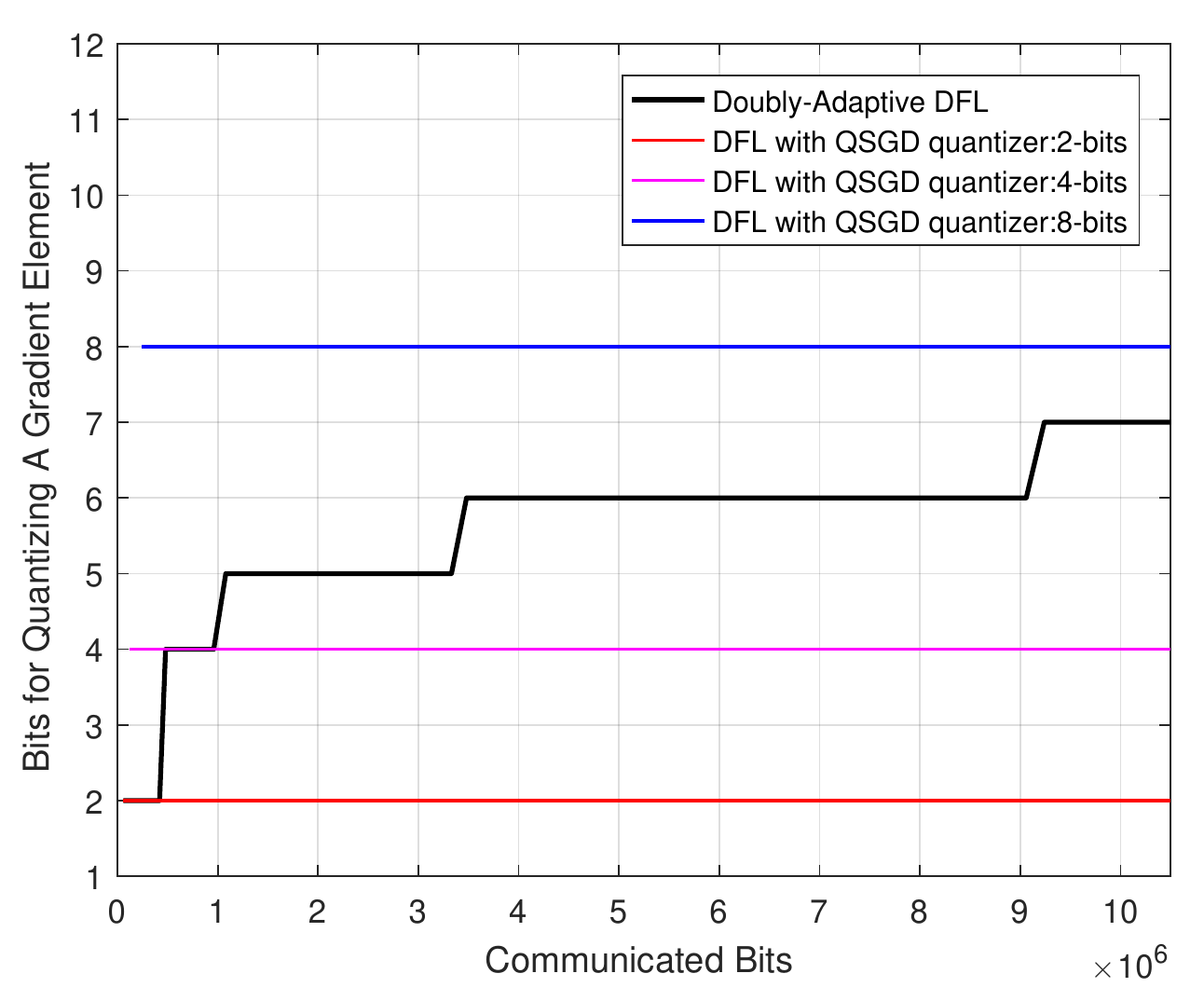}\label{eta3}}

	\subfloat[CIFAR-10: Training loss with fixed $\eta$]{\includegraphics[scale=0.43]{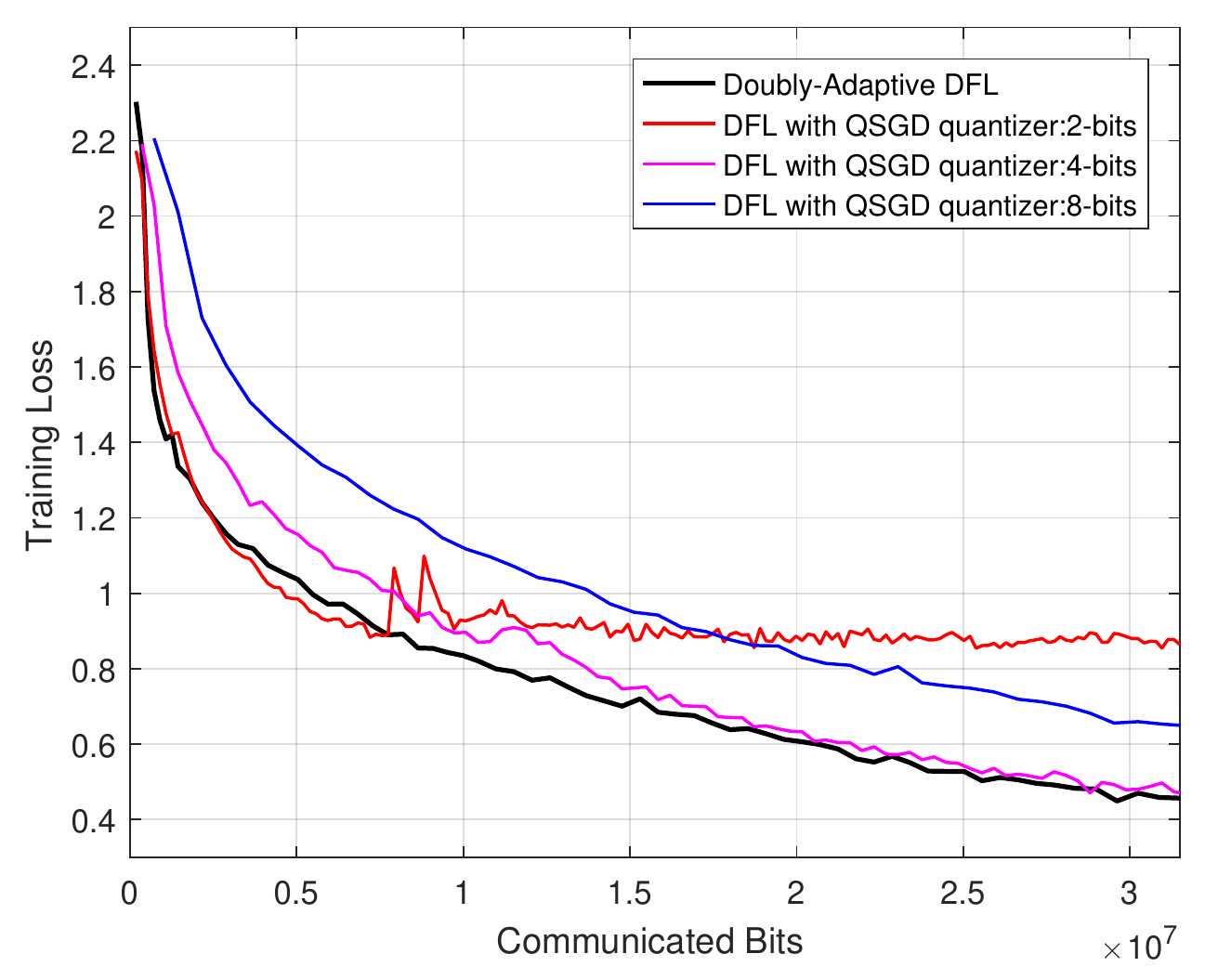}\label{eta11}}
	\subfloat[CIFAR-10: Training loss with variable $\eta$]{\includegraphics[scale=0.43]{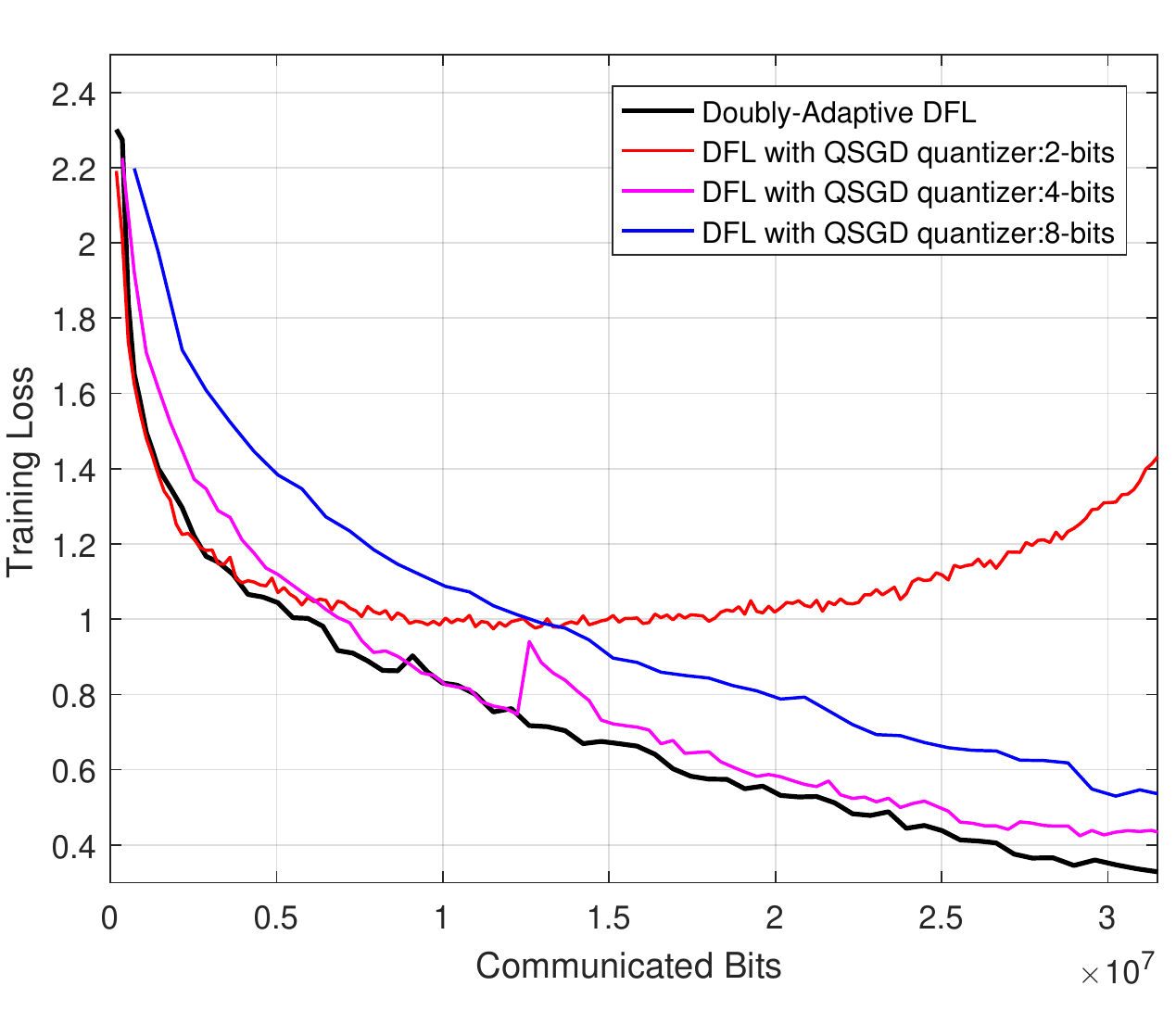}\label{eta22}}	
	\subfloat[CIFAR-10: Quantized bits for a model parameter element]{\includegraphics[scale=0.43]{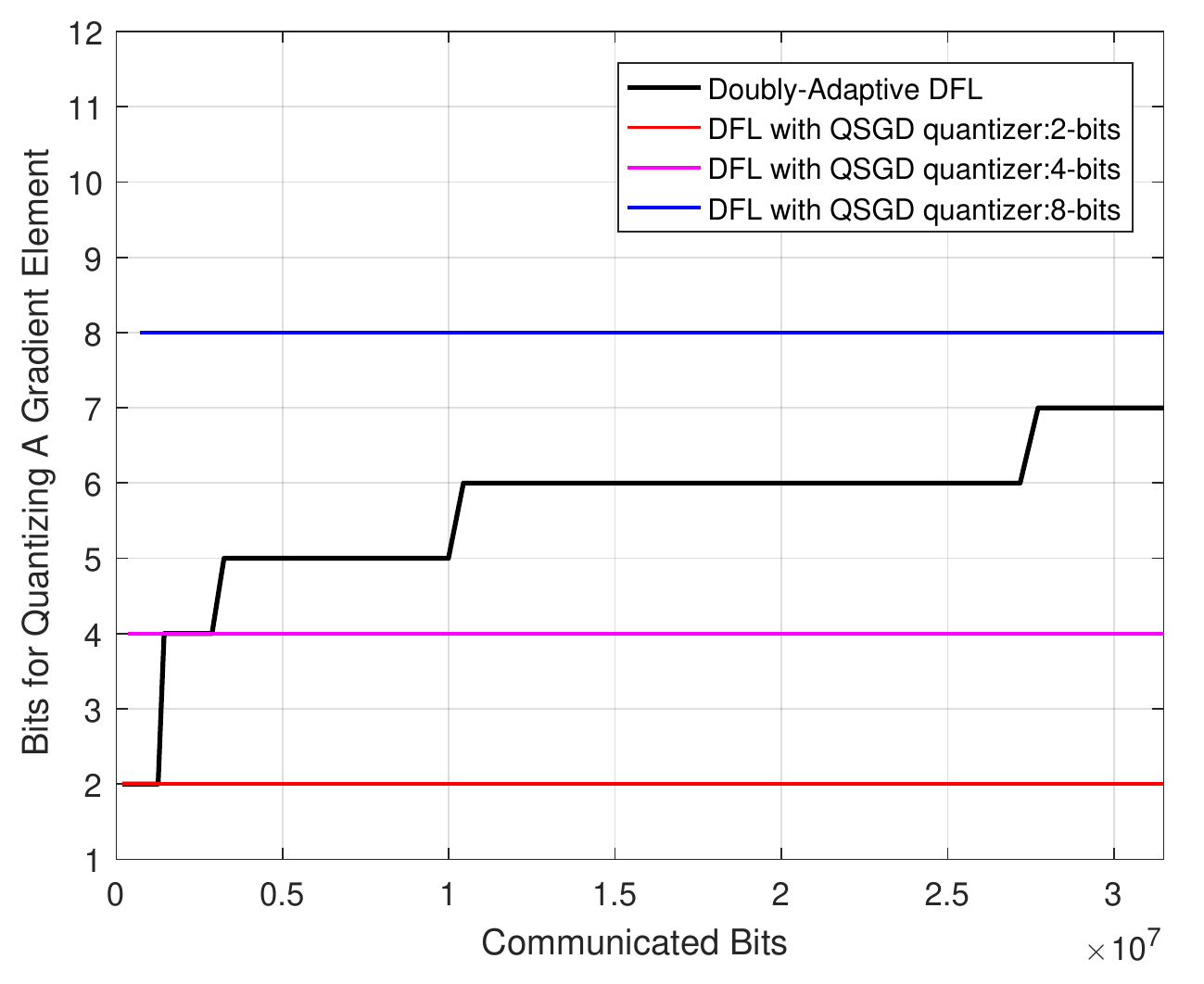}\label{eta33}}
	%\subfloat[Training loss with variable $\eta$: variable $\eta$ in terms of iterations]{\includegraphics[scale=0.5]{changedetacurve}\label{eta4}}
	\caption{Experiments on doubly-adaptive DFL based on MNIST and CIFAR-10 dataset. Doubly-adaptive DFL is conducted with the ascending number of quantization levels, which is compared with QSGD quantizer under 2, 4 and 8 bits.}
	\label{etacurve}
\end{figure*}
 
\subsubsection{Training and Quantization Parameters} In the experiment using MNIST dataset, we set the learning rate $\eta=0.002$ and the number of quantization levels $s=50$ while in the experiment using CIFAR-10 dataset, we set the learning rate $\eta=0.001$ and the number of quantization levels $s=100$. We set the number of local updates $\tau=4$, and set the initial value of model parameter $\textbf{x}_{1,0}^{(1)}=\textbf{x}_{1,0}^{(2)}=\cdots=\textbf{x}_{1,0}^{(N)}$ with an initialization of Gaussian distribution. We use mini-batch SGD in local updates synchronously. 

\subsection{Evaluation}
In this subsection, we first simulate LM-DFL and its baselines and the results can verify the effective quantization distortion improvement of LM-DFL. Then we present simulation of doubly-adaptive DFL and compare it with DFL using QSGD quantizer. %We can conclude that Doubly-Adaptive DFL saves communicated bits and improve communication efficiency when achieving the same convergence compared with DFL with QSGD quantizer.

\subsubsection{Results of LM-DFL}
Based on MNIST and CIFAR-10 datasets, the training loss curves versus iteration are presented in Fig. \ref{convergencecurve}(a) and Fig. \ref{convergencecurve}(e). Training loss of LM-DFL converges gradually to a minimal value with iteration. Hence, LM-DFL algorithm shows its convergence property. Furthermore, from the two figures, we can see that DFL without quantization shows the best convergence performance due to the 
minimal training loss value under the same iteration. 
This is because the convergence bound is smallest when quantization distortion $\omega=0$ (which means DFL without quantization), as discussed in Remark \ref{re3}. Compared with ALQ and QSGD quantizer, LM-DFL has a smaller training loss value under the same iteration. This is because the quantization distortion of LM-DFL is $d/12s^2$ which is smaller than that of ALQ and QSGD quantizer as presented in Table \ref{table-comparison}.
From Fig. \ref{convergencecurve}(c) and Fig. \ref{convergencecurve}(g), the classification performance of LM-DFL is better than DFL using ALQ and QSGD quantizer. % and DFL without quantization has the best classification performance. That still matches the analysis of quantization distortion of the four approaches.

Fig. \ref{convergencecurve}(d) and Fig. \ref{convergencecurve}(h) show quantization distortions of the four approaches. We can find that under the 50-th iteration, quantization distortion of LM-DFL decreases by 88$\%$ and 28$\%$ compared with ALQ and QSGD quantizer based on MNIST and CIFAR-10, respectively. The results shows LM-DFL conducts quantization with a smaller distortion significantly. Less proportion of distortion reduction of 
LM-DFL based on CIFAR-10 may be caused by the more complicated CNN model.

Fig. \ref{convergencecurve}(b) and Fig. \ref{convergencecurve}(f) show training loss versus time progression. Note that the time progression is based on the communication rate of 100 Mbps, where the communicated bits are recorded over a single directed connection of any node $i$ to node $j$. And the time progression is proportional to the communicated bits with fixed communication rate. We can find that LM-DFL shows the best convergence performance under the same time progression. For example, under MNIST dataset, when the time progression reaches to $30$ ms, the training loss of LM-DFL reduces by 23$\%$ compared with DFL without quantization. Under CIFAR-10 dataset, when the time progression is $200$ ms, the training loss of LM-DFL reduces by 18$\%$ compared with DFL without quantization. The results means LM-DFL can use the minimum time progression  achieve the same convergence performance compared with DFL without quantization, DFL with ALQ and QSGD quantizer.

\subsubsection{Impact of Network Topology}

Fig. \ref{zeta-fig} shows the convergence of LM-DFL under three different network topologies for MNIST dataset. The second largest absolute value $\zeta$ is 0, 0.87 and 1, which correspond to fully-connected network ($\textbf{C}=\textbf{J}$), ring topology (each node can communicate with its two neighboring nodes) and connectionless network ($\textbf{C}=\textbf{I}$), respectively. 
In order to highlight the differences of testing performance between the three topologies, we plot curves of testing accuracy differences. We can see that the fully-connected network shows the best test performance, and performance of ring topology is better than connectionless network. This is because the convergence bound of LM-DFL increases with $\zeta$, which means that larger $\zeta$ causes worse convergence performance, as analyzed in Remark \ref{topology}.

\subsubsection{Doubly-Adaptive DFL}
In this experiment, based on MNIST and CIFAR-10 dataset, we evaluate the training loss of doubly-adaptive DFL to verify the analysis that ascending number of quantization levels can convergent with less bits. Trained on MNIST and CIFAR-10, Fig. \ref{etacurve}(c) and Fig. \ref{etacurve}(f) show the variable quantized bits for a single model parameter element under both a fixed learning rate and a variable learning rate, i.e., $\lceil \log_2 s_k\rceil$, respectively. Note that for the variable learning rate $\eta_k$ under MNIST and CIFAR-10 dataset, their values decrease by 20$\%$ per 10 iterations.

From Fig. \ref{etacurve}, we can see that compared with any case of QSGD quantizer, doubly-adaptive DFL with ascending number of quantization levels achieves the best convergence performance under any communicated bits with either a fixed or variable learning rate. For example, when the number of communicated bits is $6\times 10^6$, the training loss of doubly-adaptive DFL based on MNIST reduces by 9.7$\%$ compared with 8-bits QSGD quantizer under a fixed $\eta$, and reduces by 
9.1$\%$ compared with 8-bits QSGD quantizer under a variable $\eta_k$. When the number of communicated bits is $2\times 10^7$, the training loss of doubly-adaptive DFL based on CIFAR-10 reduces by 28.6$\%$ compared with 8-bits QSGD quantizer under a fixed $\eta$, and reduces by 
31.3$\%$ compared with 8-bits QSGD quantizer under a variable $\eta_k$.
Therefore, the results are consistent with equation \eqref{bests} that DFL with the ascending number of quantization levels has the optimal convergence to save communicated bits compared with other methods of fixed levels number.

\section{Conclusion}
%In this paper, we have proposed a general decentralized machine learning framework named DFL to jointly consider model consensus and communication efficiency. DFL performs multiple local updates and multiple inter-node communications. We have established the strong convergence guarantees for DFL without the assumption of convex objective. The analysis about the convergence upper bound of DFL shows that DFL possesses superior convergence properties compared with C-SGD. Then, we have proposed C-DFL with compressed communication to further improve communication efficiency of DFL. The convergence analysis of C-DFL has indicated a linear convergence behavior. Finally, based on CNN with MNIST and CIFAR-10 datasets, experimental simulations have been made. The results have validated the theoretical analysis of DFL and C-DFL.

In this paper, we have proposed LM-DFL to minimize quantization distortion. In the stage of inter-node communication, LM-DFL performs Llyod-Max quantizing method to adaptively adjust quantization levels in terms of the probability distribution of model parameters. We have derived the quantization distortion of LM-DFL to show its superiority compared with existing FL quantizers, and established the convergence upper bound without convex loss assumption. Furthermore, we have proposed doubly-adaptive DFL, which jointly considers adaptive number of quantization levels in the training course and variable quantization levels matching the distribution of model parameters. Doubly-adaptive DFL can use much fewer communicated bits to achieve a given targeted training loss. Finally, based on CNN with MNIST and CIFAR-10 datasets, experimental results have validated our theoretical analysis of LM-DFL and doubly-adaptive DFL.
\iffalse
\section*{Acknowledgments}
This should be a simple paragraph before the References to thank those individuals and institutions who have supported your work on this article.
\fi

\iffalse
{\appendix[Proof of the Zonklar Equations]
Use $\backslash${\tt{appendix}} if you have a single appendix:
Do not use $\backslash${\tt{section}} anymore after $\backslash${\tt{appendix}}, only $\backslash${\tt{section*}}.
If you have multiple appendixes use $\backslash${\tt{appendices}} then use $\backslash${\tt{section}} to start each appendix.
You must declare a $\backslash${\tt{section}} before using any $\backslash${\tt{subsection}} or using $\backslash${\tt{label}} ($\backslash${\tt{appendices}} by itself
 starts a section numbered zero.)}
\fi

\appendices
\section{Proof of Theorem 2}

According to equation \eqref{mse}, quantization distortion of LM-DFL can be obtained by considering the part  $\sum_{i=1}^{d}\mathbb{E}[(q(r_i)-r_i)^2]$. Considering a scalar $r$ with probability density function $\phi(r)$ in $[0,1]$,
the quantization levels are denoted by $\ell_1,\ell_2,...,\ell_s$. The value between $[\ell_k,\ell_{k+1}]$ is denoted by $\ell_{k+1/2}$.
Further, assume that the quantized scalar value $r$ satisfies
\begin{align*}
\ell_{k-1/2}< r< \ell_{k+1/2}.
\end{align*}
Thus, quantization distortion in the bin $[\ell_{k-1/2},\ell_{k+1/2}]$ is written as 
\begin{align*}
d_k=\int_{\ell_{k-1/2}}^{\ell_{k+1/2}}(r-\ell_k)^2\phi(r)dr.
\end{align*}

Assume that the bin $[\ell_{k-1/2},\ell_{k+1/2}]$ is so small even considered as nearly constant which equal to $\ell_m$ where 
\begin{align*}
\ell_m=\frac{\ell_{k-1/2}+\ell_{k+1/2}}{2}.
\end{align*}
Then quantization distortion in the bin $[\ell_{k-1/2},\ell_{k+1/2}]$ can be rewritten as 
\begin{align}\label{sigmak}
d_k=\frac{\phi(\ell_m)}{3}[(\ell_{k+1/2}-\ell_k)^3+(\ell_k-\ell_{k-1/2})^3].
\end{align}
Differentiating $\sigma_k$ with respect to $\ell_k$ gives
\begin{align*}
\frac{\textup{d}d_k}{\textup{d}\ell_k}=P(\ell_m)[-(\ell_{k+1/2}-\ell_k)^2+(\ell_k-\ell_{k-1/2})^2]=0.
\end{align*}
We have 
\begin{align}
\ell_k=\frac{\ell_{k-1/2}+\ell_{k+1/2}}{2}=\ell_m.
\end{align}
Therefore, the half-way value in $[\ell_{k-1/2},\ell_{k+1/2}]$ can make $d_k$ a minimum. We set
\begin{align}
\label{k1}\ell_{k+1/2}&=\ell_k+\Delta \ell_k,\\
\label{k2}\ell_{k-1/2}&=\ell_k-\Delta \ell_k.
\end{align}
By substituting \eqref{k1} and \eqref{k2} into \eqref{sigmak}, it follows 
\begin{align}
d_k=\frac{2}{3}\phi(\ell_k)\Delta\ell_k^3.
\end{align}
The total quantization distortion is a sum of all bins. It follows
\begin{align}\label{DDD}
D=\frac{2}{3}\sum_{k=1}^{s}\phi(\ell_k)\Delta \ell_{k}^3.
\end{align}

Now we will prove the quantization distortion $D$ is a minimum when $d_k$ is a constant, independent with $k$.

According to the definition of an integral, we have 
\begin{align}\label{CCC}
\sum_{k=1}^{s}\phi^{1/3}(\ell_k)\cdot 2\Delta \ell_k=\int_{0}^{1}\phi^{1/3}(r)dr=2C,
\end{align}
where $C$ is a constant. Let $\mu_k=\phi^{1/3}(\ell_k)\Delta \ell_k$. Then \eqref{DDD} and \eqref{CCC} become
\begin{align}
D=&\frac{2}{3}\sum_{k=1}^{s}\mu_k^3,\\
C=&\sum_{k=1}^{s}\mu_k,
\end{align}
respectively.

This problem is reduced to minimizing the sum of cubes subject to the condition the sum of $\mu_k$ is a constant. Based on Lagrange's method, when 
\begin{align}
\mu_1=\mu_2=\cdots=\mu_s=\frac{C}{s},
\end{align}
the quantization distortion $D$ is a minimum with 
\begin{align}
D=&\frac{2}{3}\frac{C^3}{s^2}\notag\\
=&\frac{1}{12s^2}\left(\int_{0}^{1}\phi^{1/3}(r)dr\right)^3\notag\\
\leq&\frac{1}{12s^2}.
\end{align}
The last inequality is from H\"{o}lder's inequality. Therefore, according to \eqref{mse}, we have finished the proof of the theorem.

\section{Proof of Lemma 2}
In order to prove Lemma 2, we first present some preliminaries of necessary definitions and lemmas for convenience, and then show the detailed proof based on these preliminaries.

\subsection{Preliminaries of Proof}
We define the Frobenius norm and the operator norm of matrix as follows to make analysis easier.
\begin{definition}[\textbf{The Frobenius norm}]\label{fronorm}
	\rm The Frobenius norm defined for matrix $\textbf{\textup{A}}\in M_n$ by
	\begin{align}
	\|\textbf{\textup{A}}\|_{\textup{F}}^2=|\textup{Tr}(\textbf{A}\textbf{A}^\top)|=\sum_{i,j=1}^{n}|a_{ij}|^2.
	\end{align}
\end{definition}
\begin{definition}[\textbf{The operator norm}]\label{op-norm}
	\rm The operator norm defined for matrix $\textbf{\textup{A}}\in M_n$ by
	\begin{align}
	\|\textbf{A}\|_{\textup{op}}=\max_{\|\textbf{x}\|=1}\|\textbf{A}\textbf{x}\|=\sqrt{\lambda_{\max}(\textbf{A}^\top \textbf{A})}.
	\end{align}
\end{definition}

Then we provide three lemmas for assisting proof as follows.
\begin{lemma}\label{lem5}
	For two real matrices $\textbf{\textup{A}}\in \mathbb{R}^{d\times m}$ and $\textbf{\textup{B}}^{m\times m}$, if $\textbf{\textup{B}}$ is  symmetric, then we have 
	\begin{align}
	\|\textbf{\textup{A}}\textbf{\textup{B}}\|_{\textup{F}}\leq \|\textbf{\textup{B}}\|_{\textup{op}}\|\textbf{\textup{A}}\|_{\textup{F}}.
	\end{align} 
\end{lemma}
\begin{proof}
	Assuming that $\textbf{a}_1^\top,\dots,\textbf{a}_d^\top$ denote the rows of matrix $\textbf{A}$ and $\mathcal{I}=\{i\in[1,d]:\|\textbf{a}_i\|\neq 0\}$. Then we obtain
	\begin{align*}
	\|\textbf{A}\textbf{B}\|^2_{\textup{F}}&=\sum_{i=1}^{d}\|\textbf{a}_i^\top \textbf{B}\|^2=\sum_{i\in\mathcal{I}}^{d}\|\textbf{B}\textbf{a}_i\|^2\\
	&=\sum_{i\in\mathcal{I}}^{d}\frac{\|\textbf{B}\textbf{a}_i\|^2}{\|\textbf{a}_i\|^2}\|\textbf{a}_i\|^2\\
	&\leq\sum_{i\in\mathcal{I}}^{d}\|\textbf{B}\|^2_{\textup{op}}\|\textbf{a}_i\|^2=\|\textbf{B}\|^2_{\textup{op}}\sum_{i\in\mathcal{I}}^{d}\|\textbf{a}_i\|^2=\|\textbf{B}\|^2_{\textup{op}}\|\textbf{A}\|^2_{\textup{F}},
	\end{align*}
	where the last inequality is due to Definition \ref{op-norm} about matrix operator norm. 
\end{proof}
\begin{lemma}\label{lem6}
	Consider two matrices $\textbf{\textup{A}}\in\mathbb{R}^{m\times n}$ and $\textbf{\textup{B}}\in \mathbb{R}^{n\times m}$. We have 
	\begin{align}
	|\textup{Tr}(\textbf{\textup{A}}\textbf{\textup{B}})|\leq\|\textbf{\textup{A}}\|_{\textup{F}}\|\textbf{\textup{B}}\|_{\textup{F}}.
	\end{align}
\end{lemma}
\begin{proof}
	Assume $\textbf{a}_{i}^\top\in \mathbb{R}^{n}$ is the $i$-th row of matrix $\textbf{A}$ and $\textbf{b}_{i}^\top\in \mathbb{R}^{n}$ is the $i$-th column of matrix $\textbf{B}$. According to the definition of matrix trace, we have
	\begin{align}
	\textup{Tr}(\textbf{A}\textbf{B})&=\sum_{i=1}^{m}\sum_{j=1}^{n}\textbf{A}_{ij}\textbf{B}_{ji}\notag\\
	\label{544}&=\sum_{i=1}^{m}\textbf{a}_{i}^\top\textbf{b}_{i}.
	\end{align}
	Then according to Cauchy-Schwartz inequality, we further obtain
	\begin{align}
	|\sum_{i=1}^{m}\textbf{a}_{i}^\top\textbf{b}_{i}|^2&\leq\left(\sum_{i=1}^{m}\|\textbf{a}_{i}\|^2\right)\left(\sum_{i=1}^{m}\|\textbf{b}_{i}\|^2\right)\notag\\
	\label{555}&=\|\textbf{\textup{A}}\|_{\textup{F}}^2\|\textbf{\textup{B}}\|_{\textup{F}}^2.
	\end{align}
	Then combining \eqref{544} and \eqref{555}, we finish the proof.
\end{proof}
\begin{lemma}\label{lem7}
	Consider a matrix $\textbf{\textup{C}}\in \mathbb{R}^{m\times m}$ which satisfies condition 6 of Assumption 1. Then we have 
	\begin{align}
	\|\textbf{\textup{C}}^j-\textbf{\textup{J}}\|_{\textup{op}}=\zeta^j,
	\end{align}
	where $\zeta=\max\{|\lambda_2(\textbf{\textup{C}})|,|\lambda_m|(\textbf{\textup{C}})\}$.
\end{lemma}
\newcommand{\bm}[1]{\mbox{\boldmath{$#1$}}}
\begin{proof}
	Because $\textbf{C}$ is a real symmetric matrix, it can be decomposed as $\textbf{C}=\textbf{Q}\bm{\Lambda}\textbf{Q}^\top$, where $\textbf{Q}$ is an orthogonal matrix and $\textbf{\bm{\Lambda}}=\textup{diag}\{\lambda_1(\textbf{C}),\lambda_2(\textbf{C}),\dots,\lambda_m(\textbf{C})\}$. Similarly, matrix $\textbf{J}$ can be decomposed as $\textbf{J}=\textbf{Q}\bm{\Lambda}_0\textbf{Q}^\top$ where $\bm{\Lambda}_0=\textup{diag}\{1,0,\dots,0\}$. Then we have 
	\begin{align}\label{577}
	\textbf{C}^j-\textbf{J}=(\textbf{Q}\bm{\Lambda}\textbf{Q}^\top)^j-\textbf{J}=\textbf{Q}(\bm{\Lambda}^j-\bm{\Lambda}_0)\textbf{Q}^\top.
	\end{align} 
	According to the definition of the matrix operator norm, we further obtain
	\begin{align*}
	\|\textbf{C}^j-\textbf{J}\|_{\textup{op}}&=\sqrt{\lambda_{\max}((\textbf{C}^j-\textbf{J})^\top(\textbf{C}^j-\textbf{J}))}\\
	&=\sqrt{\lambda_{\max}(\textbf{C}^{2j}-\textbf{J})},
	\end{align*}
	where the last equality is due to
	\begin{align*}
	(\textbf{C}^j-\textbf{J})^\top(\textbf{C}^j-\textbf{J})&=(\textbf{C}^j-\textbf{J})(\textbf{C}^j-\textbf{J})\\
	&=\textbf{C}^{2j}+\textbf{J}^2-\textbf{C}^j\textbf{J}-\textbf{J}\textbf{C}^j\\
	&=\textbf{C}^{2j}+\textbf{J}-2\textbf{J}=\textbf{C}^{2j}-\textbf{J}.
	\end{align*}
	Because $\textbf{C}^{2j}-\textbf{J}=\textbf{Q}(\bm{\Lambda}^{2j}-\bm{\Lambda}_0)\textbf{Q}^\top$ from \eqref{577}, the maximum eigenvalue is $\max\{0,\lambda_2(\textbf{C})^{2j},\dots,\lambda_m(\textbf{C})^{2j}\}=\zeta^{2j}$ where $\zeta$ is the second largest eigenvalue of $\textbf{C}$. Therefore, we have 
	$$\|\textbf{C}^j-\textbf{J}\|_{\textup{op}}=\sqrt{\lambda_{\max}(\textbf{C}^{2j}-\textbf{J})}=\zeta^j.$$
\end{proof}

\subsection{Proof}
According to L-smooth gradient assumption, we have 
\begin{align}
F(\textbf{u}_{k+1})-F(\textbf{u}_k)\leq
\langle\nabla F(\textbf{u}_k),\textbf{u}_{k+1}-\textbf{u}_{k}\rangle+
\frac{L}{2}\|\textbf{u}_{k+1}-\textbf{u}_{k}\|^2.
\end{align}
According to \eqref{est4}, $\textbf{u}_k$ equals to the expectation of average estimated model $\mathbb{E}\hat{\textbf{u}}_k$. Therefore, we use the estimated value $\hat{\textbf{u}}_k$ to approximate $\textbf{u}_k$ and according to \eqref{DFLQ}, we can obtain
\begin{align}\label{lemma2-1}
&F(\textbf{u}_{k+1})-F(\textbf{u}_k)\leq\notag\\
&\langle\nabla F(\textbf{u}_k),\frac{1}{N}\sum_{i=1}^{N}Q(\textbf{x}_{k,\tau}^{(i)}-\textbf{x}_k^{(i)})\rangle+
\frac{L}{2}\|\frac{1}{N}\sum_{i=1}^{N}Q(\textbf{x}_{k,\tau}^{(i)}-\textbf{x}_k^{(i)})\|^2.
\end{align}

We first consider the term $\langle\nabla F(\textbf{u}_k),\frac{1}{N}\sum_{i=1}^{N}Q(\textbf{x}_{k,\tau}^{(i)}-\textbf{x}_k^{(i)})\rangle$ in \eqref{lemma2-1}. It follows
\begin{align*}
&\mathbb{E}\langle\nabla F(\textbf{u}_k),\frac{1}{N}\sum_{i=1}^{N}Q(\textbf{x}_{k,\tau}^{(i)}-\textbf{x}_k^{(i)})\rangle\\
=&\frac{1}{N}\sum_{i=1}^{N}\langle\nabla F(\textbf{u}_k), \mathbb{E}[Q(\textbf{x}_{k,\tau}^{(i)}-\textbf{x}_k^{(i)})]\rangle\\
=&\frac{1}{N}\sum_{i=1}^{N}\langle\nabla F(\textbf{u}_k), \textbf{x}_{k,\tau}^{(i)}-\textbf{x}_k^{(i)}\rangle\\
=&\frac{1}{N}\sum_{i=1}^{N}\langle\nabla F(\textbf{u}_k),\!\! -\eta(\widetilde{\nabla}f_i(\textbf{x}_k^{(i)})\!+\!\widetilde{\nabla}f_i(\textbf{x}_{k,1}^{(i)})\!+\!\cdots+\widetilde{\nabla}f_i(\textbf{x}_{k,\tau-1}^{(i)}))\rangle\\
=&-\frac{\eta}{N}\sum_{i=1}^{N}\sum_{t=0}^{\tau-1}\langle\nabla F(\textbf{u}_k),\widetilde{\nabla}f_i(\textbf{x}_{k,t}^{(i)})\rangle\\
=&\!\!-\!\!\frac{\eta}{2N}\!\!\sum_{i=1}^{N}\!\!\sum_{t=0}^{\tau-1}[\|\nabla F(\textbf{u}_k)\|^2\!\!+\!\!\|\widetilde{\nabla}f_i(\textbf{x}_{k,t}^{(i)})\|^2\!\!-\!\!\|\nabla F(\textbf{u}_k)\!\!-\!\!\widetilde{\nabla}f_i(\textbf{x}_{k,t}^{(i)})\|^2],
\end{align*}
where the last equation comes from $2\textbf{a}^\top \textbf{b}=\|\textbf{a}\|^2+\|\textbf{b}\|^2-\|\textbf{a}-\textbf{b}\|^2$.
Then we focus on $\|\nabla F(\textbf{u}_k)-\widetilde{\nabla}f_i(\textbf{x}_{k,t}^{(i)})\|^2$, which is 
\begin{align}
&\mathbb{E}\|\nabla F(\textbf{u}_k)-\widetilde{\nabla}f_i(\textbf{x}_{k,t}^{(i)})\|^2\notag\\
=&\mathbb{E}\|\widetilde{\nabla}f_i(\textbf{x}_{k,t}^{(i)})-\nabla F(\textbf{u}_k)-\nabla F_i(\textbf{x}_{k,t}^{(i)})+\nabla F(\textbf{u}_k)\|^2+\notag\\
&\label{44}\|\nabla F_i(\textbf{x}_{k,t}^{(i)})-\nabla F(\textbf{u}_k)\|^2\\
=&\sigma_i^2+\|\nabla F_i(\textbf{x}_{k,t}^{(i)})-\nabla F(\textbf{u}_k)\|^2\notag\\
=&\sigma_i^2+\|\nabla F_i(\textbf{x}_{k,t}^{(i)})-\nabla F(\textbf{u}_k)-\nabla F(\textbf{x}_{k,t}^{(i)})+\nabla F(\textbf{u}_k)\|^2+\notag\\
&\label{45}\|\nabla F(\textbf{x}_{k,t}^{(i)})-\nabla F(\textbf{u}_k)\|^2\\
\label{46}\leq&\sigma_i^2+\delta_i^2+L^2\|\textbf{x}_{k,t}^{(i)}-\textbf{u}_k\|^2,
\end{align}
where \eqref{44} and \eqref{45} come from variance $\textup{Var}(\textbf{X})=\mathbb{E}(\textbf{X}^2)-(\mathbb{E}\textbf{X})^2$.
Substituting \eqref{46} in the term $\langle\nabla F(\textbf{u}_k),\frac{1}{N}\sum_{i=1}^{N}Q(\textbf{x}_{k,\tau}^{(i)}-\textbf{x}_k^{(i)})\rangle$, we have
\begin{align}\label{47}
&\langle\nabla F(\textbf{u}_k),\frac{1}{N}\sum_{i=1}^{N}Q(\textbf{x}_{k,\tau}^{(i)}-\textbf{x}_k^{(i)})\rangle\notag\\
\leq&-\frac{\eta\tau}{2}\|\nabla F(\textbf{u}_k)\|^2-\frac{\eta}{2N}\sum_{i=1}^{N}\sum_{t=0}^{\tau-1}\|\widetilde{\nabla}f_i(\textbf{x}_{k,t}^{(i)})\|^2+\notag\\
&\frac{\eta}{2N}\sum_{i=1}^{N}\sum_{t=0}^{\tau-1}(\sigma_i^2+\delta_i^2+L^2\|\textbf{x}_{k,t}^{(i)}-\textbf{u}_k\|^2)
\end{align}

Then we consider the term $\frac{L}{2}\|\frac{1}{N}\sum_{i=1}^{N}Q(\textbf{x}_{k,\tau}^{(i)}-\textbf{x}_k^{(i)})\|^2$ as follows,
\begin{align}\label{48}
&\mathbb{E}\|\frac{1}{N}\sum_{i=1}^{N}Q(\textbf{x}_{k,\tau}^{(i)}-\textbf{x}_k^{(i)})\|^2\notag\\
=&\frac{1}{N^2}\mathbb{E}\|\sum_{i=1}^{N}Q(\textbf{x}_{k,\tau}^{(i)}\!\!-\!\!\textbf{x}_k^{(i)})\!\!-\!\!\sum_{i=1}^{N}(\textbf{x}_{k,\tau}^{(i)}\!\!-\!\!\textbf{x}_k^{(i)})\!\!+\!\!\sum_{i=1}^{N}(\textbf{x}_{k,\tau}^{(i)}\!\!-\!\!\textbf{x}_k^{(i)})\|^2\notag\\
=&\frac{1}{N^2}\mathbb{E}\|\sum_{i=1}^{N}[Q(\textbf{p}_i)-\textbf{p}_i]\|^2+\frac{1}{N^2}\mathbb{E}\|\sum_{i=1}^{N}\textbf{p}_i\|^2\\
\label{49}=&\frac{1}{N^2}\sum_{i=1}^{N}\mathbb{E}\|Q(\textbf{p}_i)-\textbf{p}_i\|^2+\frac{1}{N^2}\mathbb{E}\|\sum_{i=1}^{N}\textbf{p}_i\|^2\\
\label{50}\leq&\frac{\omega}{N^2}\sum_{i=1}^{N}\|\textbf{x}_{k,\tau}^{(i)}-\textbf{x}_k^{(i)}\|^2+\frac{1}{N}\sum_{i=1}^{N}
\|\textbf{x}_{k,\tau}^{(i)}-\textbf{x}_k^{(i)}\|^2\\
=&\frac{\omega+N}{N^2}\sum_{i=1}^{N}\|\textbf{x}_{k,\tau}^{(i)}-\textbf{x}_k^{(i)}\|^2\notag\\
\label{51}\leq&\frac{(\omega+N)\eta^2\tau}{N^2}\sum_{i=1}^{N}\sum_{t=0}^{\tau-1}\|\widetilde{\nabla}f_i(\textbf{x}_{k,t}^{(i)})\|^2,
\end{align}
where we define $\textbf{p}_i=\textbf{x}_{k,\tau}^{(i)}-\textbf{x}_k^{(i)}$ for convenience, and \eqref{48} and \eqref{49} are because $Q(\textbf{x})$ is unbiased, while \eqref{50} is based on the definition of quantization distortion in Definition 1.

By substituting \eqref{47} and \eqref{51} into \eqref{lemma2-1}, we obtain
\begin{align}\label{eq70}
&\mathbb{E}F(\textbf{u}_{k+1})-\mathbb{E}F(\textbf{u}_k)\notag\\
\leq&-\frac{\eta\tau}{2}\|\nabla F(\textbf{u}_k)\|^2-\frac{\eta}{2N}\sum_{i=1}^{N}\sum_{t=0}^{\tau-1}\|\widetilde{\nabla}f_i(\textbf{x}_{k,t}^{(i)})\|^2\notag\\
&+\frac{\eta}{2N}\sum_{i=1}^{N}\sum_{t=0}^{\tau-1}(\sigma_i^2+\delta_i^2+L^2\|\textbf{x}_{k,t}^{(i)}-\textbf{u}_k\|^2)\notag\\
&+\frac{(\omega+N)L\eta^2\tau}{2N^2}\sum_{i=1}^{N}\sum_{t=0}^{\tau-1}\|\widetilde{\nabla}f_i(\textbf{x}_{k,t}^{(i)})\|^2\\
=&-\frac{\eta\tau}{2}\|\nabla F(\textbf{u}_k)\|^2\!\!-\!\!\frac{\eta}{2N}\left[1\!\!-\!\!\frac{(\omega+N)L\eta\tau}{N}\right]\!\!\sum_{i=1}^{N}\!\sum_{t=0}^{\tau-1}\|\widetilde{\nabla}f_i(\textbf{x}_{k,t}^{(i)})\|^2\notag\\
&+\frac{\eta}{2N}\sum_{i=1}^{N}\sum_{t=0}^{\tau-1}(\sigma_i^2+\delta_i^2+L^2\|\textbf{x}_{k,t}^{(i)}-\textbf{u}_k\|^2)\\
=&\!\!-\frac{\eta\tau}{2}\|\nabla F(\textbf{u}_k)\|^2\!-\!\frac{\eta}{2N}\left[1\!\!-\!\!\frac{(\omega+N)L\eta\tau}{N}\right]\sum_{i=1}^{N}\sum_{t=0}^{\tau-1}\|\widetilde{\nabla}f_i(\textbf{x}_{k,t}^{(i)})\|^2\notag\\
&+\frac{\eta\sigma^2\tau}{2}+\frac{\eta\tau\delta^2}{2}+\frac{\eta L^2}{2N}\sum_{i=1}^{N}\sum_{t=0}^{\tau-1}\|\textbf{x}_{k,t}^{(i)}-\textbf{u}_k\|^2\notag\\
\leq&\!\!-\!\!\frac{\eta\tau}{2}\|\nabla F(\textbf{u}_k)\|^2\!\!
+\!\frac{\eta\sigma^2\tau}{2}+\frac{\eta\tau\delta^2}{2}\!+\!\frac{\eta L^2}{2N}\sum_{i=1}^{N}\sum_{t=0}^{\tau-1}\|\textbf{x}_{k,t}^{(i)}\!-\!\textbf{u}_k\|^2\notag\\
&\label{54}-\!\!\frac{\eta}{2N}\left[1\!\!-\!\!\frac{(\omega+N)L\eta\tau}{N}\right]\sum_{i=1}^{N}\sum_{t=0}^{\tau-1}(\sigma_i^2\!\!+\!\!\|\nabla F_i(\textbf{x}_{k,t}^{(i)})\|^2)\\
=&-\!\!\frac{\eta\tau}{2}\|\nabla F(\textbf{u}_k)\|^2\!\!+\frac{(\omega+N)L\eta^2\tau^2\sigma^2}{2N}+\frac{\eta\tau\delta^2}{2}\notag\\
&-\frac{\eta}{2N}\left[1\!\!-\!\!\frac{(\omega+N)L\eta\tau}{N}\right]\sum_{i=1}^{N}\sum_{t=0}^{\tau-1}\|\nabla F_i(\textbf{x}_{k,t}^{(i)})\|^2\notag\\
&\label{55}+\frac{\eta L^2}{2N}\sum_{i=1}^{N}\sum_{t=0}^{\tau-1}\|\textbf{x}_{k,t}^{(i)}-\textbf{u}_k\|^2,
\end{align}
where \eqref{54} comes from condition 3 of Assumption 1 that the variance of gradient estimation is $\sigma$.
By minor rearranging \eqref{55}, it follows
\begin{align}
&\frac{\eta\tau}{2}\|\nabla F(\textbf{u}_k)\|^2\notag\\
\leq&\mathbb{E}F(\textbf{u}_k)-\mathbb{E}F(\textbf{u}_{k+1})+\frac{(\omega+N)L\eta^2\tau^2\sigma^2}{2N}+\frac{\eta\tau\delta^2}{2}\notag\\
&-\frac{\eta}{2N}\left[1\!\!-\!\!\frac{(\omega+N)L\eta\tau}{N}\right]\sum_{i=1}^{N}\sum_{t=0}^{\tau-1}\|\nabla F_i(\textbf{x}_{k,t}^{(i)})\|^2\notag\\
&\label{56}+\frac{\eta L^2}{2N}\sum_{i=1}^{N}\sum_{t=0}^{\tau-1}\|\textbf{x}_{k,t}^{(i)}-\textbf{u}_k\|^2.
\end{align} 
Then we have
\begin{align}
&\|\nabla F(\textbf{u}_k)\|^2\notag\\
\leq&\frac{2[\mathbb{E}F(\textbf{u}_k)-\mathbb{E}F(\textbf{u}_{k+1})]}{\eta\tau}+\frac{(\omega+N)L\eta\tau\sigma^2}{N}+\delta^2\notag\\
&-\frac{1}{N\tau}\left[1\!\!-\!\!\frac{(\omega+N)L\eta\tau}{N}\right]\sum_{i=1}^{N}\sum_{t=0}^{\tau-1}\|\nabla F_i(\textbf{x}_{k,t}^{(i)})\|^2\notag\\
\label{57}&+\frac{L^2}{N\tau}\sum_{i=1}^{N}\sum_{t=0}^{\tau-1}\|\textbf{x}_{k,t}^{(i)}-\textbf{u}_k\|^2.
\end{align}
Based on \eqref{57}, taking the total expectation and averaging over all iterations, we have
\begin{align}
	&\mathbb{E}\left[\frac{1}{K}\sum_{k=1}^{K}\|\nabla F(\textbf{u}_k)\|^2\right]\notag\\
	\leq& \frac{2[F(\textbf{u}_1)-F_{\textup{inf}}]}{\eta K\tau}+\frac{(\omega+N)L\eta\tau\sigma^2}{N}+\delta^2\notag\\
	&-\frac{1}{N\tau K}\left[1\!\!-\!\!\frac{(\omega+N)L\eta\tau}{N}\right]\sum_{k=1}^{K}\sum_{i=1}^{N}\sum_{t=0}^{\tau-1}\|\nabla F_i(\textbf{x}_{k,t}^{(i)})\|^2\notag\\
	&\label{58}+\frac{L^2}{NK\tau}\sum_{k=1}^{K}\sum_{i=1}^{N}\sum_{t=0}^{\tau-1}\|\textbf{x}_{k,t}^{(i)}-\textbf{u}_k\|^2.
\end{align}

Focusing on \eqref{58}, we try to further find an upper bound of the term $\frac{L^2}{NK\tau}\sum_{k=1}^{K}\sum_{i=1}^{N}\sum_{t=0}^{\tau-1}\|\textbf{x}_{k,t}^{(i)}-\textbf{u}_k\|^2$. 

Because 
\begin{align}
&\sum_{i=1}^{N}\|\textbf{x}_{k,t}^{(i)}-\textbf{u}_k\|^2\notag\\
\leq&2\sum_{i=1}^{N}\|\textbf{x}_{k}^{(i)}-\textbf{u}_k\|^2+2\eta^2\sum_{i=1}^{N}t\sum_{p=0}^{t-1}\|\widetilde{\nabla}f_i(\textbf{x}_{k,p}^{(i)})\|^2\notag\\
\label{61}=&2\sum_{i=1}^{N}\|\textbf{x}_{k}^{(i)}-\textbf{u}_k\|^2+2\eta^2t\sum_{i=1}^{N}\sum_{p=0}^{t-1}\|\widetilde{\nabla}f_i(\textbf{x}_{k,p}^{(i)})\|^2,
\end{align}
we consider $\sum_{i=1}^{N}\|\textbf{x}_{k}^{(i)}-\textbf{u}_k\|^2$ in \eqref{61}. According to Definition \ref{fronorm}, we have
\begin{align}
&\sum_{i=1}^{N}\|\textbf{x}_{k}^{(i)}-\textbf{u}_k\|^2\notag\\
=&\|\textbf{X}_k-\textbf{u}_k\cdot\textbf{1}^\top \|_{\textup{F}}^2\notag\\
=&\|\textbf{X}_k-\frac{\textbf{X}_k\textbf{1}\textbf{1}^\top}{N}\|_{\textup{F}}^2\notag\\
=&\|\textbf{X}_k(\textbf{I}-\textbf{J})\|_{\textup{F}}^2.
\end{align}

From \eqref{QDFLlearningrule}, we have
\begin{align}
\textbf{X}_{k}=[\hat{\textbf{X}}_{k-1}+Q(\textbf{X}_{k-1,\tau}-\textbf{X}_{k-1})]\textbf{C}.
\end{align}
By taking an expectation for the above equation, we have 
\begin{align}
\mathbb{E}\textbf{X}_{k}&=[\mathbb{E}\hat{\textbf{X}}_{k-1}+\mathbb{E}Q(\textbf{X}_{k-1,\tau}-\textbf{X}_{k-1})]\textbf{C}\notag\\
&=[\textbf{X}_{k-1}+\textbf{X}_{k-1,\tau}-\textbf{X}_{k-1}]\textbf{C}.
\end{align}

Noting that $\textbf{C}\textbf{J}=\textbf{J}\textbf{C}=\textbf{J}$, we have
\begin{align}
&\textbf{X}_{k}(\textbf{I}-\textbf{J})\notag\\
=&[\textbf{X}_{k-1}+(\textbf{X}_{k-1,\tau}-\textbf{X}_{k-1})]\textbf{C}(\textbf{I}-\textbf{J})\notag\\
=&\textbf{X}_{k-1}(\textbf{I}-\textbf{J})\textbf{C}+(\textbf{X}_{k-1,\tau}-\textbf{X}_{k-1})(\textbf{C}-\textbf{J})\notag\\
=&\textbf{X}_{k-2}(\textbf{I}-\textbf{J})\textbf{C}^2+(\textbf{X}_{k-2,\tau}-\textbf{X}_{k-2})(\textbf{C}^2-\textbf{J})\notag\\
&+(\textbf{X}_{k-1,\tau}-\textbf{X}_{k-1})(\textbf{C}-\textbf{J}).
\end{align}
By repeating the same procedure from $k-2,k-3,...,1$, we finally get
\begin{align}
&\textbf{X}_{k}(\textbf{I}-\textbf{J})\notag\\
=&\textbf{X}_1(\textbf{I}-\textbf{J})\textbf{C}^{k-1}+\sum_{s=1}^{k-1}(\textbf{X}_{s,\tau}-\textbf{X}_{s})(\textbf{C}^{k-s}-\textbf{J})
\end{align}
Because the variable $\textbf{X}$ has the same initialized point with $\textbf{X}_1(\textbf{I}-\textbf{J})=\textbf{0}$, the squared norm of $\textbf{X}_{k}(\textbf{I}-\textbf{J})$ can written as 
\begin{align}\label{72}
&\mathbb{E}\|\textbf{X}_k(\textbf{I}-\textbf{J})\|_{\textup{F}}^2\notag\\
=&\mathbb{E}\|\sum_{s=1}^{k-1}(\textbf{X}_{s,\tau}-\textbf{X}_{s})(\textbf{C}^{k-s}-\textbf{J})\|_{\textup{F}}^2\notag\\
=&\mathbb{E}\|\sum_{s=1}^{k-1}\textbf{q}_s(\textbf{C}^{k-s}-\textbf{J})\|_{\textup{F}}^2,
\end{align}
where we use $\textbf{q}_s$ to denote $\textbf{X}_{s,\tau}-\textbf{X}_{s}$ for convenience. Following \eqref{72}, we have
\begin{align}
	&\|\sum_{s=1}^{k-1}\textbf{q}_s(\textbf{C}^{k-s}-\textbf{J})\|_{\textup{F}}^2\notag\\
	=&\sum_{s=1}^{k-1}\|\textbf{q}_s(\textbf{C}^{k-s}-\textbf{J})\|_{\textup{F}}^2\notag\\
	&+\sum_{n=1}^{k-1}\sum_{l=1,l\neq n}^{k-1}\textup{Tr}((\textbf{C}^{k-n}-\textbf{J})\textbf{q}_n^\top\textbf{q}_l(\textbf{C}^{k-l}-\textbf{J}))\\
	\leq &\sum_{s=1}^{k-1}\zeta^{2(k-s)}\|\textbf{q}_s\|_{\textup{F}}^2\notag\\
	\label{79}&+\frac{1}{2}\sum_{n=1}^{k-1}\sum_{l=1,l\neq n}^{k-1}\zeta^{2k-n-l}[\|\textbf{q}_n\|_{\textup{F}}^2+\|\textbf{q}_l\|_{\textup{F}}^2],
\end{align}
where \eqref{79} is from Lemma \ref{lem5}, \ref{lem6} and \ref{lem7}.

Then we further have
\begin{align}
	&\mathbb{E}\|\sum_{s=1}^{k-1}\textbf{q}_s(\textbf{C}^{k-s}-\textbf{J})\|_{\textup{F}}^2\notag\\
	\leq&\sum_{s=1}^{k-1}\zeta^{2(k-s)}\|\textbf{q}_s\|_{\textup{F}}^2\!\!+\!\!\frac{1}{2}\sum_{n=1}^{k-1}\sum_{l=1,l\neq n}^{k-1}\!\!\!\zeta^{2k-n-l}[\|\textbf{q}_n\|_{\textup{F}}^2+\|\textbf{q}_l\|_{\textup{F}}^2]\notag\\
	=&\sum_{s=1}^{k-1}\zeta^{2(k-s)}\|\textbf{q}_s\|_{\textup{F}}^2+\sum_{n=1}^{k-1}\sum_{l=1,l\neq n}^{k-1}\zeta^{2k-n-l}\|\textbf{q}_n\|_{\textup{F}}^2\notag\\
	=&\sum_{s=1}^{k-1}\zeta^{2(k-s)}\|\textbf{q}_s\|_{\textup{F}}^2+\sum_{n=1}^{k-1}\zeta^{k-n}\|\textbf{q}_n\|_{\textup{F}}^2\sum_{l=1,l\neq n}^{k-1}\zeta^{k-l}\notag\\
	\leq&\sum_{s=1}^{k-1}\zeta^{2(k-s)}\|\textbf{q}_s\|_{\textup{F}}^2+\sum_{n=1}^{k-1}\frac{\zeta^{k-n}}{1-\zeta}\|\textbf{q}_n\|_{\textup{F}}^2\notag\\
	\label{80}=&\sum_{s=1}^{k-1}\left[\zeta^{2(k-s)}+\frac{\zeta^{k-s}}{1-\zeta}\right]\|\textbf{q}_s\|_{\textup{F}}^2.
\end{align}

Considering $\frac{L^2}{NK\tau}\sum_{k=1}^{K}\sum_{i=1}^{N}\sum_{t=0}^{\tau-1}\|\textbf{x}_{k,t}^{(i)}-\textbf{u}_k\|^2$, because
\begin{align}\label{81}
	&\sum_{i=1}^{N}\|\textbf{x}_{k}^{(i)}-\textbf{u}_k\|^2\notag\\
	\leq&2\sum_{i=1}^{N}\|\textbf{x}_{k}^{(i)}-\textbf{u}_k\|^2+2\eta^2t\sum_{i=1}^{N}\sum_{p=0}^{t-1}\|\widetilde{\nabla}f_i(\textbf{x}_{k,p}^{(i)})\|^2,
\end{align}
this inequality consists of two parts, which are $\frac{L^2}{NK\tau}\sum_{k=1}^{K}\sum_{t=0}^{\tau-1}2\sum_{i=1}^{N}\|\textbf{x}_{k}^{(i)}-\textbf{u}_k\|^2$ and $\frac{L^2}{NK\tau}\sum_{k=1}^{K}\sum_{t=0}^{\tau-1}2\eta^2t\sum_{i=1}^{N}\sum_{p=0}^{t-1}\|\widetilde{\nabla}f_i(\textbf{x}_{k,p}^{(i)})\|^2$, respectively. According to \eqref{80}, the first part  $\frac{L^2}{NK\tau}\sum_{k=1}^{K}\sum_{t=0}^{\tau-1}2\sum_{i=1}^{N}\|\textbf{x}_{k}^{(i)}-\textbf{u}_k\|^2$ can be written as 
\begin{align}
	&\frac{L^2}{NK\tau}\sum_{k=1}^{K}\sum_{t=0}^{\tau-1}2\sum_{i=1}^{N}\|\textbf{x}_{k}^{(i)}-\textbf{u}_k\|^2\notag\\
	\leq&\frac{L^2}{NK\tau}\sum_{k=1}^{K}\sum_{t=0}^{\tau-1}2\sum_{s=1}^{k-1}\left[\zeta^{2(k-s)}+\frac{\zeta^{k-s}}{1-\zeta}\right]\|\textbf{q}_s\|_{\textup{F}}^2\notag\\
	\label{82}=&\frac{2L^2}{NK}\sum_{k=1}^{K}\sum_{s=1}^{k-1}\left[\zeta^{2(k-s)}+\frac{\zeta^{k-s}}{1-\zeta}\right]\|\textbf{q}_s\|_{\textup{F}}^2.
\end{align}
In \eqref{82}, the coefficient of $\|\textbf{q}_s\|_{\textup{F}}^2$ is $\sum_{k=1}^{K-s}\left[\zeta^{2k}+\frac{\zeta^k}{1-\zeta}\right]$, and $$ \sum_{k=1}^{K-s}\left[\zeta^{2k}+\frac{\zeta^k}{1-\zeta}\right] \leq \frac{\zeta^2}{1-\zeta^2}+\frac{\zeta}{(1-\zeta)^2}.$$
Then we have 
\begin{align}
	&\frac{L^2}{NK\tau}\sum_{k=1}^{K}\sum_{t=0}^{\tau-1}2\sum_{i=1}^{N}\|\textbf{x}_{k}^{(i)}-\textbf{u}_k\|^2\notag\\
	\leq&\frac{2L^2}{NK}\left[\frac{\zeta^2}{1-\zeta^2}+\frac{\zeta}{(1-\zeta)^2}\right]\sum_{k=1}^{K}\|\textbf{q}_k\|_{\textup{F}}^2\notag\\
	\leq&\frac{2L^2}{NK}\left[\frac{\zeta^2}{1-\zeta^2}+\frac{\zeta}{(1-\zeta)^2}\right]\sum_{k=1}^{K}\sum_{i=1}^{N}\eta^2\tau\sum_{t=0}^{\tau-1}\|\widetilde{\nabla}f_i(\textbf{x}_{k,t}^{(i)})\|\notag\\
	=&\frac{2L^2\eta^2\tau}{NK}\left[\frac{\zeta^2}{1-\zeta^2}\!\!+\!\!\frac{\zeta}{(1-\zeta)^2}\right]\sum_{k=1}^{K}\sum_{i=1}^{N}\sum_{t=0}^{\tau-1}(\sigma_i^2+\|\nabla F_i(\textbf{x}_{k,t}^{(i)})\|^2)\notag\\
	=&2L^2\eta^2\tau^2\sigma^2\left[\frac{\zeta^2}{1-\zeta^2}+\frac{\zeta}{(1-\zeta)^2}\right]\notag\\
	&+\frac{2L^2\eta^2\tau}{NK}\left[\frac{\zeta^2}{1-\zeta^2}+\frac{\zeta}{(1-\zeta)^2}\right]\sum_{k=1}^{K}\sum_{i=1}^{N}\sum_{t=0}^{\tau-1}\|\nabla F_i(\textbf{x}_{k,t}^{(i)})\|^2\notag\\
	\label{83}=&2L^2\eta^2\tau^2\sigma^2\alpha+\frac{2L^2\eta^2\tau\alpha}{NK}\sum_{k=1}^{K}\sum_{i=1}^{N}\sum_{t=0}^{\tau-1}\|\nabla F_i(\textbf{x}_{k,t}^{(i)})\|^2,
\end{align}
where $\alpha=\left[\frac{\zeta^2}{1-\zeta^2}+\frac{\zeta}{(1-\zeta)^2}\right]$.

Then we focus on the second term in \eqref{81}. We have
\begin{align}
	&\frac{L^2}{NK\tau}\sum_{k=1}^{K}\sum_{t=0}^{\tau-1}2\eta^2t\sum_{i=1}^{N}\sum_{p=0}^{t-1}\|\widetilde{\nabla}f_i(\textbf{x}_{k,p}^{(i)})\|^2\notag\\
	=&\frac{L^2}{NK\tau}\sum_{k=1}^{K}\sum_{t=0}^{\tau-1}2\eta^2t\sum_{i=1}^{N}\sum_{p=0}^{t-1}(\sigma_i^2+\|\nabla F_i(\textbf{x}_{k,t}^{(i)})\|^2)\notag\\
	\leq&\frac{2}{3}L^2\eta^2\sigma^2\tau^2\notag\\
	&+\frac{2L^2\eta^2}{NK\tau}\sum_{k=1}^{K}\sum_{i=1}^{N}\sum_{t=0}^{\tau-2}\frac{1}{2}(\tau-t-1)(\tau+t)\|\nabla F_i(\textbf{x}_{k,t}^{(i)})\|^2\notag\\
	\label{84}\leq&\frac{2}{3}L^2\eta^2\sigma^2\tau^2+\frac{L^2\eta^2\tau}{NK}\sum_{k=1}^{K}\sum_{i=1}^{N}\sum_{t=0}^{\tau-1}\|\nabla F_i(\textbf{x}_{k,t}^{(i)})\|^2
\end{align}
Therefore, we have derived an upper bound of the term $\frac{L^2}{NK\tau}\sum_{k=1}^{K}\sum_{i=1}^{N}\sum_{t=0}^{\tau-1}\|\textbf{x}_{k,t}^{(i)}-\textbf{u}_k\|^2$.
Then based on \eqref{83} and \eqref{84}, \eqref{58} can be rewritten as 
\begin{align}
&\mathbb{E}\left[\frac{1}{K}\sum_{k=1}^{K}\|\nabla F(\textbf{u}_k)\|^2\right]\notag\\
\leq& \frac{2[F(\textbf{u}_1)-F_{\textup{inf}}]}{\eta K\tau}+\frac{(\omega+N)L\eta\tau\sigma^2}{N}+\delta^2\notag\\
&-\frac{1}{N\tau K}\left[1\!\!-\!\!\frac{(\omega+N)L\eta\tau}{N}\right]\sum_{k=1}^{K}\sum_{i=1}^{N}\sum_{t=0}^{\tau-1}\|\nabla F_i(\textbf{x}_{k,t}^{(i)})\|^2\notag\\
&+2L^2\eta^2\tau^2\sigma^2\alpha+\frac{2L^2\eta^2\tau\alpha}{NK}\sum_{k=1}^{K}\sum_{i=1}^{N}\sum_{t=0}^{\tau-1}\|\nabla F_i(\textbf{x}_{k,t}^{(i)})\|^2\notag\\
&+\frac{2}{3}L^2\eta^2\sigma^2\tau^2+\frac{L^2\eta^2\tau}{NK}\sum_{k=1}^{K}\sum_{i=1}^{N}\sum_{t=0}^{\tau-1}\|\nabla F_i(\textbf{x}_{k,t}^{(i)})\|^2\\
=&\frac{2[F(\textbf{u}_1)-F_{\textup{inf}}]}{\eta K\tau}+\frac{(\omega+N)L\eta\tau\sigma^2}{N}+\delta^2+2L^2\eta^2\tau^2\sigma^2\alpha\notag\\
&+\frac{2}{3}L^2\eta^2\sigma^2\tau^2+\frac{2L^2\eta^2\tau\alpha}{NK}\sum_{k=1}^{K}\sum_{i=1}^{N}\sum_{t=0}^{\tau-1}\|\nabla F_i(\textbf{x}_{k,t}^{(i)})\|^2\notag\\
&+\frac{L^2\eta^2\tau}{NK}\sum_{k=1}^{K}\sum_{i=1}^{N}\sum_{t=0}^{\tau-1}\|\nabla F_i(\textbf{x}_{k,t}^{(i)})\|^2\notag\\
&-\frac{1}{N\tau K}\left[1\!\!-\!\!\frac{(\omega+N)L\eta\tau}{N}\right]\sum_{k=1}^{K}\sum_{i=1}^{N}\sum_{t=0}^{\tau-1}\|\nabla F_i(\textbf{x}_{k,t}^{(i)})\|^2.
\end{align}
If the coefficient of $\sum_{k=1}^{K}\sum_{i=1}^{N}\sum_{t=0}^{\tau-1}\|\nabla F_i(\textbf{x}_{k,t}^{(i)})\|^2$ satisfies
\begin{align}
	\frac{2L^2\eta^2\tau\alpha}{NK}+\frac{L^2\eta^2\tau}{NK}-\frac{1}{N\tau K}\left[1\!\!-\!\!\frac{(\omega+N)L\eta\tau}{N}\right]\leq 0,
\end{align}
then it follows
\begin{align}
	&\mathbb{E}\left[\frac{1}{K}\sum_{k=1}^{K}\|\nabla F(\textbf{u}_k)\|^2\right]\notag\\
	\leq&\frac{2[F(\textbf{u}_1)-F_{\textup{inf}}]}{\eta K\tau}+\frac{(\omega+N)L\eta\tau\sigma^2}{N}+\delta^2+2L^2\eta^2\tau^2\sigma^2\alpha\notag\\
	&+\frac{2}{3}L^2\eta^2\sigma^2\tau^2,
\end{align}
where $\alpha=\left[\frac{\zeta^2}{1-\zeta^2}+\frac{\zeta}{(1-\zeta)^2}\right]$.
Considering the condition that coefficient satisfies, we have 
\begin{align*}
	L^2\tau(2\alpha+1)\eta^2+\frac{L(\omega+N)\eta}{N}-\frac{1}{\tau}\leq 0.
\end{align*}
Then we get
\begin{align}
	\eta\leq \frac{\sqrt{(\omega+N)^2+4N^2(2\alpha+1)}-\omega-N}{2NL\tau(2\alpha+1)}. 
\end{align}
Here, we complete the proof of Lemma 2.

\section{Proof of Theorem 4}
From Lemma 2, after $K$ iterations, we have 
\begin{align}
&\mathbb{E}\left[\frac{1}{K}\sum_{k=1}^{K}\|\nabla F(\textbf{u}_k)\|^2\right]\notag\\
\leq&\frac{2[F(\textbf{u}_1)-F_{\textup{inf}}]}{\eta K\tau}+\frac{(\omega+N)L\eta\tau\sigma^2}{N}+\delta^2+2L^2\eta^2\tau^2\sigma^2\alpha\notag\\
&+\frac{2}{3}L^2\eta^2\sigma^2\tau^2,
\end{align}
where 
\begin{align}
\eta\leq \frac{\sqrt{(\omega+N)^2+4N^2(2\alpha+1)}-\omega-N}{2NL\tau(2\alpha+1)}
\end{align}
and $\alpha=\left[\frac{\zeta^2}{1-\zeta^2}+\frac{\zeta}{(1-\zeta)^2}\right]$.

According to Definition 3, we have $K=B/2C_s$. So we get
\begin{align}
&\mathbb{E}\left[\frac{2C_s}{B}\sum_{k=1}^{B/2C_s}\|\nabla F(\textbf{u}_k)\|^2\right]\notag\\
\leq&\frac{4C_s[F(\textbf{u}_1)-F_{\textup{inf}}]}{B\eta \tau}+\frac{(\omega+N)L\eta\tau\sigma^2}{N}+\delta^2+2L^2\eta^2\tau^2\sigma^2\alpha\notag\\
&+\frac{2}{3}L^2\eta^2\sigma^2\tau^2.
\end{align}
Due to $C_s=d\lceil \log_2 s \rceil+d+32\leq d \log_2(2s)+d+32$ according to \eqref{numberofbits} and $\omega=\frac{d}{12s^2}$ from Theorem 2, we have
\begin{align}
&\mathbb{E}\left[\frac{2C_s}{B}\sum_{k=1}^{B/2C_s}\|\nabla F(\textbf{u}_k)\|^2\right]\notag\\
\leq&\frac{4(d \log_2(2s)+d+32)[F(\textbf{u}_1)-F_{\textup{inf}}]}{B\eta \tau}+\frac{(\frac{d}{12s^2}+N)L\eta\tau\sigma^2}{N}\!\\
&+\!\delta^2
\!+\!\left(2\alpha+\frac{2}{3}\right)L^2\eta^2\sigma^2\tau^2\notag\\
=&\frac{4[F(\textbf{u}_1)-F_{\textup{inf}}]d}{\eta\tau B}\log_2(2s)+\frac{dL\eta\tau\sigma^2}{12Ns^2}\\
&+\frac{4[F(\textbf{u}_1)-F_{\textup{inf}}](d+32)}{\eta\tau B}\notag\\
&+\left(2\alpha+\frac{2}{3}\right)L^2\eta^2\sigma^2\tau^2+\delta^2+L\eta\tau\sigma^2\notag\\
=&A_1\log_2(2s)+\frac{A_2}{s^2}+A_3,
\end{align}
where 
\begin{align*}
A_1&=\frac{4[F(\textup{\textbf{u}}_1)-F_{\textup{inf}}]d}{\eta\tau B},\ A_2=\frac{L\eta\tau\sigma^2d}{12N},\\
A_3&=\frac{A_1}{d}(d+32)+(2\alpha+\frac{2}{3})L^2\eta^2\sigma^2\tau^2+\delta^2+L\eta\tau\sigma^2.
\end{align*}
Therefore, we finish the proof of Theorem 4.

\section{Proof of quantization distortion of LM-DFL: another expression}
In this section, we give the proof of another expression of LM-DFL's quantization distortion, which is presented in the following theorem:
\begin{theorem}[\textbf{Another Expression of LM-DFL Quantization Distortion}]
	Let $\textbf{\textup{v}}\in\mathbb{R}^d$. The quantization distortion of LM-DFL can be expressed as
	\begin{align}
	\mathbb{E}[\|Q_L(\textbf{\textup{v}})-\textbf{\textup{v}}\|^2]\leq\|\textbf{\textup{v}}\|^2\left(\frac{\ell_{j^*+1}/\ell_{j^*}-1}{\ell_{j^*+1}/\ell_{j^*}+1}\right)^2,
	\end{align}
	where $j^*=\arg\max_{1\leq j\leq s-1}\ell_{j+1}/\ell_j$.
\end{theorem}

\begin{proof}
	We define a boundary sequence $\boldsymbol{b}=[b_0,b_1,...,b_{s-1},b_s]$, where $b_0=0,b_s=1,b_j\in(0,1)$, $j=1,2,...,s-1$.
	And quantization level $\ell_j$ falls in the bin $[b_{j-1},b_{j}]$ for $j=1,...,s$. The distortion of $Q_L(\textbf{v})$ can be expressed as 
	\begin{align}\label{95}
		\mathbb{E}[\|Q_L(\textbf{v})-\textbf{v}\|^2]=\|\textbf{v}\|^2\left(\sum_{j=1}^{s}\sum_{r_i\in\mathcal{I}_j}(\ell_j-r_i)^2\right),
	\end{align} 
	where $\mathcal{I}_j=[b_{j-1},b_{j}]$ for $j=1,...,s$. We let $(\ell_j-r)^2=k_jr^2$, and try to find the maximum of $k_j$.
	We have 
	\begin{align}
		k_j=\left(\frac{\ell_j-r}{r}\right)^2,
	\end{align}
	where $r\in[b_{j-1},b_j]$. By differentiating above equation, we get
	$$\frac{dk_j}{dr}=\frac{2(r-\ell_j)\ell_j}{r^3}.$$
	Then when $b_{j-1}<r<\ell_j$, $\frac{2(r-\ell_j)\ell_j}{r^3}<0$; when $\ell_j<r<b_j$, $\frac{2(r-\ell_j)\ell_j}{r^3}>0$. So the maximum of $k_j$ is obtained in $b_{j-1}$ or $b_j$. Because $b_{j-1}=\frac{\ell_{j-1}+\ell_j}{2}$ and $b_j=\frac{\ell_j+\ell_{j+1}}{2}$, we have
	\begin{align*}
		k_j\leq\left(\frac{\ell_j-\frac{\ell_{j-1}+\ell_j}{2}}{\frac{\ell_{j-1}+\ell_j}{2}}\right)^2=\left(\frac{\ell_j-\ell_{j-1}}{\ell_j+\ell_{j-1}}\right)^2=\left(\frac{\ell_j/\ell_{j-1}-1}{\ell_j/\ell_{j-1}+1}\right)^2
	\end{align*}
	or
	\begin{align*}
	k_j\leq\left(\frac{\ell_j-\frac{\ell_j+\ell_{j+1}}{2}}{\frac{\ell_j+\ell_{j+1}}{2}}\right)^2=\left(\frac{\ell_{j+1}-\ell_{j}}{\ell_{j+1}+\ell_j}\right)^2=\left(\frac{\ell_{j+1}/\ell_{j}-1}{\ell_{j+1}/\ell_{j}+1}\right)^2.
	\end{align*}
	Because $\left(\frac{x-1}{x+1}\right)^2$ is increasing with $x$ when $x>1$,
	then we get
	\begin{align}
		k_j\leq \left(\frac{\ell_{j^*+1}/\ell_{j^*}-1}{\ell_{j^*+1}/\ell_{j^*}+1}\right)^2,
	\end{align}
	where $j^*=\arg\max_{1\leq j\leq s-1}\ell_{j+1}/\ell_j$ and $s=1,2,...,s-1$.
	Thus \eqref{95} can be rewritten as 
	\begin{align}
	\mathbb{E}[\|Q_L(\textbf{v})-\textbf{v}\|^2]&\leq\|\textbf{v}\|^2\left(\frac{\ell_{j^*+1}/\ell_{j^*}-1}{\ell_{j^*+1}/\ell_{j^*}+1}\right)^2\sum_{j=1}^{s}\sum_{r_i\in\mathcal{I}_j}r_i^2\notag\\
	&=\|\textbf{v}\|^2\left(\frac{\ell_{j^*+1}/\ell_{j^*}-1}{\ell_{j^*+1}/\ell_{j^*}+1}\right)^2\sum_{i=1}^{d}r_i^2\notag\\
	&=\|\textbf{v}\|^2\left(\frac{\ell_{j^*+1}/\ell_{j^*}-1}{\ell_{j^*+1}/\ell_{j^*}+1}\right)^2\sum_{i=1}^{d}\frac{|v_i|^2}{\|\textbf{v}\|^2}\notag\\
	&=\|\textbf{v}\|^2\left(\frac{\ell_{j^*+1}/\ell_{j^*}-1}{\ell_{j^*+1}/\ell_{j^*}+1}\right)^2
	\end{align}
	Here we finish the proof.
\end{proof}

\section{Proof of Theorem 5}
In eq.(\ref{eq70}) of Appendix B, replacing the learning rate $\eta $ into  ${\eta _k}$, it becomes

\begin{equation}
 \begin{aligned}
&{\eta _k}{\left\| {\nabla F({u_k})} \right\|^2} \le \frac{{2({\bf{E}}\left[ {F({u_k})} \right] - {\bf{E}}\left[ {F({u_{k + 1}})} \right])}}{\tau } + \frac{{L\eta _k^2\tau {\sigma ^2}d}}{{12Ns_k^2}} \\
& + L\eta _k^2\tau {\sigma ^2} + {\eta _k}{\delta ^2}  + \frac{{{\eta _k}{L^2}}}{{N\tau }}\sum\limits_{i = 1}^N {\sum\limits_{t = 0}^{\tau  - 1} {{{\left\| {x_{k,t}^{(i)} - {u_k}} \right\|}^2}} } \\
& - \frac{{{\eta _k}}}{{N\tau }}\left[ {1 - \frac{{L{\eta _k}\tau }}{{12Ns_k^2}} - L{\eta _k}\tau } \right]\sum\limits_{i = 1}^N {\sum\limits_{t = 0}^{\tau  - 1} {{{\left\| {\nabla {F_i}(x_{k,t}^{(i)})} \right\|}^2}} },
\end{aligned}   
\end{equation}
where ${\varpi _k} = \frac{d}{{12s_k^2}}$. We now summing over all rounds  $k \in \left\{ {1, \cdots ,K} \right\}$. Then, one has

\begin{equation}
 \begin{aligned}
&{\bf{E}}\left[ {\sum\limits_{k = 1}^K {{\eta _k}{{\left\| {\nabla F({u_k})} \right\|}^2}} } \right] \le \frac{{2(F({u_k}) - {F_{\inf }})}}{\tau } \\
&+ \frac{{L\tau {\sigma ^2}d\sum\limits_{k = 1}^K {\left( {\eta _k^2/s_k^2} \right)} }}{{12N}} 
 + L\tau {\sigma ^2}\sum\limits_{k = 1}^K {\eta _k^2}  + {\delta ^2}\sum\limits_{k = 1}^K {{\eta _k}} \\
&+ \left[ {-\frac{1}{{N\tau }}\sum\limits_{k = 1}^K {{\eta _k}}  + \frac{L}{{12{N^2}}}\sum\limits_{k = 1}^K {\left( {\eta _k^2/s_k^2} \right)}  + \frac{L}{N}\sum\limits_{k = 1}^K {\eta _k^2} } \right]\cdot\\
&\sum\limits_{i = 1}^N {\sum\limits_{t = 0}^{\tau  - 1} {{{\left\| {\nabla {F_i}(x_{k,t}^{(i)})} \right\|}^2}} }  + \frac{{{L^2}}}{{N\tau }}\sum\limits_{k = 1}^K {{\eta _k}} \sum\limits_{i = 1}^N {\sum\limits_{t = 0}^{\tau  - 1} {{{\left\| {x_{k,t}^{(i)} - {u_k}} \right\|}^2}} }  \\
&\le \frac{{2(F({u_k}) - {F_{\inf }})}}{\tau } + \frac{{L\tau {\sigma ^2}d\sum\limits_{k = 1}^K {\left( {\eta _k^2/s_k^2} \right)} }}{{12N}} + L\tau {\sigma ^2}\sum\limits_{k = 1}^K {\eta _k^2} \\
&+ {\delta ^2}\sum\limits_{k = 1}^K {{\eta _k}}  + \left( {2\alpha  + \frac{2}{3}} \right){L^2}{\tau ^2}{\sigma ^2}\sum\limits_{k = 1}^K {\eta _k^3}  \\
& + D\sum\limits_{i = 1}^N {\sum\limits_{t = 0}^{\tau  - 1} {{{\left\| {\nabla {F_i}(x_{k,t}^{(i)})} \right\|}^2}} } ,
\end{aligned}   
\end{equation}
where $\alpha  = \frac{{{\zeta ^2}}}{{1 - {\zeta ^2}}} + \frac{\zeta }{{{{\left( {1 - \zeta } \right)}^2}}}$, and $D =  - \frac{1}{{N\tau }}\sum\limits_{k = 1}^K {{\eta _k}}  + \frac{L}{N}\sum\limits_{k = 1}^K {\eta _k^2}  + \left( {\frac{{2{L^2}\tau \alpha }}{N} + \frac{{{L^2}\tau }}{N}} \right)\sum\limits_{k = 1}^K {\eta _k^3 + \frac{L}{{{N^2}}}} \sum\limits_{k = 1}^K {\eta _k^2{\varpi _k}} $.
Define ${D_k} =  - \frac{1}{{N\tau }}{\eta _k} + \frac{L}{N}\eta _k^2 + \left( {\frac{{2{L^2}\tau \alpha }}{N} + \frac{{{L^2}\tau }}{N}} \right)\eta _k^3 + \frac{L}{{{N^2}}}\eta _k^2{\varpi _k}$. If $D<0$ for all $k$, $D_k<0$ and one has

\begin{align}
 {\eta _k} \le \frac{{\sqrt {{{\left( {{\varpi _k} + N} \right)}^2} + 4{N^2}\left( {2\alpha  + 1} \right)}  - {\varpi _k} - N}}{{2NL\tau \left( {2\alpha  + 1} \right)}},   
\end{align}
and

\begin{equation}
 \begin{aligned}
&{\bf{E}}\left[ {\sum\limits_{k = 1}^K {{\eta _k}{{\left\| {\nabla F({u_k})} \right\|}^2}} } \right] \le \frac{{2(F({u_k}) - {F_{\inf }})}}{\tau } \\
&+ \frac{{L\tau {\sigma ^2}d\sum\limits_{k = 1}^K {\left( {\eta _k^2/s_k^2} \right)} }}{{12N}} + L\tau {\sigma ^2}\sum\limits_{k = 1}^K {\eta _k^2}  + {\delta ^2}\sum\limits_{k = 1}^K {{\eta _k}} \\
&+ \left( {2\alpha  + \frac{2}{3}} \right){L^2}{\tau ^2}{\sigma ^2}\sum\limits_{k = 1}^K {\eta _k^3} 
\end{aligned}
\end{equation}
Dividing both sides by $\sum\limits_{k = 1}^K {{\eta _k}} $ and make $\delta  = 0$  we have,

\begin{equation}
\begin{aligned}
&{\bf{E}}\left[ {\frac{{\sum\limits_{k = 1}^K {{\eta _k}{{\left\| {\nabla F\left( {{u_k}} \right)} \right\|}^2}} }}{{\sum\limits_{k = 1}^K {{\eta _k}} }}} \right] \le \frac{{2\left[ {F\left( {{u_1}} \right) - {F_{\inf }}} \right]}}{{\tau \sum\limits_{k = 1}^K {{\eta _k}} }} \\
&+ \frac{{L\tau {\sigma ^2}\sum\limits_{k = 1}^K {\eta _k^2\left( {d/s_k^2} \right)} }}{{12N\sum\limits_{k = 1}^K {{\eta _k}} }} + \frac{{L\tau {\sigma ^2}\sum\limits_{k = 1}^K {\eta _k^2} }}{{\sum\limits_{k = 1}^K {{\eta _k}} }} \\
&+ \left( {2\alpha  + \frac{2}{3}} \right){L^2}{\tau ^2}{\sigma ^2}\frac{{\sum\limits_{k = 1}^K {\eta _k^3} }}{{\sum\limits_{k = 1}^K {{\eta _k}} }},
\end{aligned}
\end{equation}
which completes the proof.


%\section{References Section}
%You can use a bibliography generated by BibTeX as a .bbl file.
% The IEEEtran BibTeX style support page is:
 
%\section{Simple References}
%You can manually copy in the resultant .bbl file and set second argument of $\backslash${\tt{begin}} to the number of references
% (used to reserve space for the reference number labels box).

\small
\bibliographystyle{unsrt} 
\bibliography{bare_jrnl_new_sample4}

\iffalse

\newpage

\section{Biography Section}
If you have an EPS/PDF photo (graphicx package needed), extra braces are
 needed around the contents of the optional argument to biography to prevent
 the LaTeX parser from getting confused when it sees the complicated
 $\backslash${\tt{includegraphics}} command within an optional argument. (You can create
 your own custom macro containing the $\backslash${\tt{includegraphics}} command to make things
 simpler here.)
 
\vspace{11pt}

\bf{If you include a photo:}\vspace{-33pt}
\begin{IEEEbiography}[{\includegraphics[width=1in,height=1.25in,clip,keepaspectratio]{fig1}}]{Michael Shell}
Use $\backslash${\tt{begin\{IEEEbiography\}}} and then for the 1st argument use $\backslash${\tt{includegraphics}} to declare and link the author photo.
Use the author name as the 3rd argument followed by the biography text.
\end{IEEEbiography}

\vspace{11pt}

\bf{If you will not include a photo:}\vspace{-33pt}
\begin{IEEEbiographynophoto}{John Doe}
Use $\backslash${\tt{begin\{IEEEbiographynophoto\}}} and the author name as the argument followed by the biography text.
\end{IEEEbiographynophoto}

\fi

\vfill

\end{document}